\documentclass[amsart,11pt,oneside,english]{article}
\usepackage{amssymb,amsmath,amsfonts, amsmath,mathtools,
amsthm,euscript,mathrsfs,MnSymbol,verbatim,enumerate,multirow,bbding,slashed,multicol,color,array, esint,babel, tikz,tikz-cd,tikz-3dplot,tkz-graph,pgfplots,ytableau,graphicx,float,rotating,hyperref,geometry,mathdots,savesym,cite, wasysym,amscd,graphicx,pifont,float,setspace,multirow,wrapfig,picture,subfigure, amsthm, enumitem}
\usepackage[normalem]{ulem} 
\usepackage[utf8]{inputenc}
\usepackage{prettyref}
\usepackage[T1]{fontenc}
\usepackage{datetime}

\numberwithin{equation}{section}
\usetikzlibrary{arrows,positioning,decorations.pathmorphing,   decorations.markings, matrix, patterns}
\hypersetup{colorlinks=true}
\hypersetup{linkcolor=black}
\hypersetup{citecolor=black}
\hypersetup{urlcolor=black}
\theoremstyle{plain}
\newtheorem*{thm*}{Theorem}
\theoremstyle{plain}
\newtheorem{thm}{Theorem}[section]

\newtheorem{lem}[thm]{Lemma}
\newtheorem{prop}[thm]{Proposition}
\theoremstyle{definition}
\newtheorem{defn}[thm]{Definition}
	\newtheorem{exmp}[thm]{Example}
		
\newtheorem{rem}[thm]{Remark}

\newcommand*{\medcap}{\mathbin{\scalebox{1.0}[1.5]{\ensuremath{\cap}}}}%

\numberwithin{equation}{section}

\tikzset{
  big arrow/.style={
    decoration={markings,mark=at position 1 with {\arrow[scale=1.5,#1]{>}}},
    postaction={decorate},
    shorten >=0.4pt},
  big arrow/.default=black}
\newtheorem*{assumption*}{\assumptionnumber}
\providecommand{\assumptionnumber}{}


\begin{document}

\begin{titlepage}
\begin{center}
\vspace{2cm}
{\Huge\bfseries    The Geometry of $\text{G$_2$},$ Spin($7$), and Spin($8$)-models \\  }
\vspace{2cm}
{\Large
Mboyo Esole$^{1}$, Ravi Jagadeesan$^{2,3,4}$, and Monica Jinwoo Kang$^{5,6}$\\}
\vspace{.6cm}
{\large $^1$ Department of Mathematics, Northeastern University}\par
{\large  360 Huntington Avenue, Boston, MA 02115, USA.}\par
\vspace{.3cm}
{\large $^2$ Department of Mathematics, Harvard University}\par
1 Oxford Street, Cambridge, MA 02138, USA. \par
currently at \par
 $^3$ Harvard Business School\par 
Wyss Hall, Soldiers Field, Boston, MA 02163, USA.\par 
$^4$ Department of Economics, Harvard University\par 
1805 Cambridge St, Cambridge, MA 02138, USA.\par 
\vspace{.3cm}
{\large $^5$ Department of Physics,  Harvard University}\par
17 Oxford Street, Cambridge, MA 02138, USA.\par
currently at \par 
$^6$ Walter Burke Institute for Theoretical Physics\par
 California Institute of Technology \par
 1200 East California Boulevard,   Pasadena, CA 91125, USA.
  \par 
\vspace{2cm}
{ \bf{Abstract}}\\
\end{center}

We study the geometry of elliptic fibrations given by  Weierstrass models resulting from Step 6 of Tate's algorithm. Such elliptic fibrations have a discriminant locus containing  an irreducible component $S$, over which the generic fiber is of  Kodaira  type  I$^*_0$. 
In string geometry, these geometries are used to geometrically engineer 
G$_2$, Spin($7$), and Spin($8$) gauge  theories. 
We give sufficient conditions for the existence of crepant resolutions. 
When they exist, we give a complete description of all crepant resolutions and show explicitly how the network of flops matches the Coulomb branch of the associated gauge theories. 
We also compute the triple intersection numbers in each chamber.  Physically, they  correspond to the  Chern--Simons levels  of the gauge theory and depend on the choice of a Coulomb branch. 
We determine the representations associated with these elliptic fibrations by computing intersection numbers with fibral divisors and then interpreting them as weights of a representation. 
For a five-dimensional gauge theory, we compute the number of hypermultiplets in each representation by matching the triple intersection numbers with the superpotential of the theory. 
 We also  discuss anomaly cancellations of a six-dimensional supergravity theory obtained by a compactification of F-theory on an elliptically fibered Calabi--Yau threefold corresponding to a G$_2$, Spin($7$), or Spin($8$) gauge theory. 

\vfill 

\end{titlepage}

\tableofcontents
\newpage

\section{Introduction}
 The existence and geometric properties of  crepant resolutions for Weierstrass models over higher-dimensional bases are an important  but still  incomplete aspects 
 of the theory of elliptic fibrations. This study has many applications, such as in the theory of  Calabi--Yau varieties,  string geometry, and the classification of superconformal field theories.  For examples of crepant resolutions of Weierstrass models,  see \cite{Bershadsky:1996nh,Braun:2013cb,EY,ES,SU2U1,SO6, ESY1,ESY2,Euler, F4,Lawrie:2012gg,Marsano,Morrison:2011mb}. 

The goal of this paper is to study the geometry of crepant resolutions of Weierstrass models corresponding to Step 6 of Tate's algorithm \cite{Tate.Alg}.  
 Such a Weierstrass model has a discriminant locus $\Delta$ containing a nonsingular and irreducible divisor $S$ of the base $B$ such that the  geometric
fiber over the generic point of $S$ 
is of Kodaira type  I$^*_0$ (whose dual graph is the affine Dynkin diagram of type $\widetilde{\text{D}}_4$).
The divisor $S$ appears with multiplicity six in the discriminant $\Delta$ and the remainder of the discriminant $\Delta'= \Delta S^{-6}$  is typically singular. 
The Kodaira fibers over generic points of $\Delta'$ are of type I$_1$ and the divisors $S$ and $\Delta'$  do not intersect transversally. 
  At the intersection of $S$ and $\Delta'$, we have a ``collision of singularities''\footnote{When two irreducible components $\Delta_1$ and $\Delta_2$ of the discriminant locus intersect, if we denote by $T_1$ and $T_2$ the Kodaira type of their respective generic fibers, their intersection is called a collision of singularity of type $T_1+T_2$. 
The collision is said to be transverse when $\Delta_1$ and $\Delta_2$ intersect transversally.  The type of the  generic fiber over their intersection is usually not one of the Kodaira type. 
The possible types are classified for Miranda's models \cite{Miranda.smooth,Szydlo.Thesis}, which are regularizations of Weierstrass models that gives  flat elliptic fibrations  such that the $j$-invariant is a morphism. 
 Miranda's models only allow transverse collisions.
} of type  I$_1$+I$^*_0$
 yielding non-Kodaira fibers whose structures are explained in detail in later sections.
   Such a collision is not allowed in Miranda's models since the fiber I$_1$ has an infinite $j$-invariant while I$^*_0$ can only take finite values for its  $j$-invariant.

\subsection{I$_0^*$ as the most versatile Kodaira type}

The fiber    I$^*_0$  has at least three unique properties that distinguish it from all   other Kodaira fibers.
Firstly, like smooth elliptic curves,  the I$^*_0$  fiber can have a $j$-invariant of any value in the ground field of the elliptic fibration.
On the other hand, all other singular
Kodaira fibers have a $j$-invariant taking the values $0$ or $1728$, with the exceptions of the Kodaira fiber I$_{n>0}^*$ and I$_{n>0}$, whose $j$-invariant has a pole and takes an infinite value.

 Secondly, while all the other Kodaira fibers have at most two splitting types---split and non-split---the fiber I$^*_0$ distinguishes itself by having three splitting types---split, semi-split, and non-split--- corresponding to three distinct Lie algebras---namely 
G$_2$, B$_3$, and D$_4$.
As we assume that the Mordell-Weil group is trivial, these three Lie algebras correspond to the simply-connected groups G$_2$, Spin($7$), and Spin($8$). 
 It is also more natural to distinguish two different versions of the non-split case: there are two possible Galois groups, which give rise to distinct fiber degenerations, as we shall explain  below.

Lastly, the fiber of type I$^*_0$ plays a central role in Miranda's models \cite{Miranda.smooth} and in the classification of ``non-Higgsable clusters'' \cite{Morrison:2012np}.  The fiber  I$^*_0$  appears in collisions of the types $j=0$ (namely  II+I$_0^*$ and  IV+I$_0^*$)  and $j=1728$  (namely  III+I$_0^*$) \cite{Miranda.smooth}. 
 These collisions are building blocks for the non-Higgsable clusters of elliptically fibered Calabi--Yau threefolds \cite{Morrison:2012np}. 
For example, an isolated curve of self-intersection $-4$ can support a Kodaira fiber of type I$_0^*$ and an associated Lie algebra $\mathfrak{so}_8$ with a trivial representation \cite{Morrison:2012np}.  
There are also three building blocks consisting of two or three rational curves of negative self-intersection intersecting transversally.  We write $(-n_1,-n_2,-n_3,\ldots ,-n_r)$ for a chain of $r$  rational curves $C_i$, with $1\leq i \leq r$, where $C_i^2=-n_i$ and  two consecutive curves in the chain  intersect transversally. 
 A $(-2,-3)$-chain corresponds to the collision of two Kodaira fibers of type III and I$_0^*$, yielding a semi-simple Lie algebra  $\mathfrak{sp}_1\oplus \mathfrak{g}_2$ \cite{Morrison:2012np}. 
 A $(-2,-2,-3)$-chain corresponds to the chain of Kodaira fibers  II+IV+I$_0^*$, supporting the Lie algebra 
$\mathfrak{su}_2\oplus \mathfrak{g}_2$ \cite{Morrison:2012np}; finally, a  $(-2,-3,-2)$-chain corresponds to the chain of Kodaira fibers  III+I$_0^*$+III,  yielding a Lie algebra $\mathfrak{su}_2\oplus \mathfrak{so}_7\oplus \mathfrak{su}_2$ \cite{Morrison:2012np}.

\begin{table}[hbt]
\begin{center}
\scalebox{.95}{$			
\begin{array}{|c | c | c |c|} 
\hline
\scalebox{1.2}{\text{ Fiber type}} &\scalebox{1.2}{ \text{Dual graph}} & \scalebox{1.2}{\text{Splitting type}} & \scalebox{1.2}{\text{Field extension $[\kappa':\kappa]$}}
\\  \hline 
			\begin{array}{c}
			\text{I}^{*\text{ns}}_{0}
			\\
			\widetilde{\text{G}}_2^t
\end{array}
 &		
	\scalebox{.8}{$\begin{array}{c}\begin{tikzpicture}
				\node[draw,circle,thick,scale=1.25,label=above:{1}, fill=black] (1) at (0,0){};
				\node[draw,circle,thick,scale=1.25,label=above:{2}] (2) at (1.3,0){};
				\node[draw,circle,thick,scale=1.25,label=above:{1}] (3) at (2.6,0){};
				\draw[thick] (1) to (2);
				\draw[thick] (1.5,0.09) --++ (.9,0);
				\draw[thick] (1.5,-0.09) --++ (.9,0);
				\draw[thick] (1.5,0) --++ (.9,0);
				\draw[thick]
					(1.9,0) --++ (60:.25)
					(1.9,0) --++ (-60:.25);
			\end{tikzpicture}\end{array}
		$} 
& 
		\scalebox{.8}{$\begin{array}{c}\begin{tikzpicture}
				\node[draw,circle,thick,scale=1.25,label=30:{2}] (0) at (0,0){};
				\node[draw,circle,thick,scale=1.25,label=above:{1}, fill=black] (1) at (-1,0){};
				\node[draw,circle,thick,scale=1.25,label=right:{1}] (2) at (1,0){};
				\node[draw,circle,thick,scale=1.25,label=above:{1}] (3) at (90:1){};
								\node[draw,circle,thick,scale=1.25,label=below:{1}] (4) at (90:-1){};
				\draw[thick] (0) to (1);
				\draw[thick] (0) to (2);
				\draw[thick] (0) to (3);
				\draw[thick] (0) to (4);
				\draw[<->,>=stealth',semithick,dashed]  (1.1,0.3) arc (25:65:1.7cm);
				\draw[<->,>=stealth',semithick,dashed]  (1.1,-0.3) arc (-25:-65:1.7cm);
			\end{tikzpicture}
			\end{array}$}
			& 
			3 \quad \text{or} \quad 6
\\
\hline 
\begin{array}{c}
			\text{I}^{*\text{ss}}_{0}
			\\
			\widetilde{\text{B}}_3^t
\end{array}
 &
 		\scalebox{.8}{$\begin{array}{c}\begin{tikzpicture}
				\node[draw,circle,thick,scale=1.25,label=below:{1}] (1) at (0,-.5){};
				\node[draw,circle,thick,scale=1.25,label=above:{2}] (2) at (1.3,0){};
				\node[draw,circle,thick,scale=1.25,label=above:{1}] (3) at (2.6,0){};
								\node[draw,circle,thick,scale=1.25,label=above:{1}, fill=black] (4) at (0,.5){};
				\draw[thick] (1) to (2);\draw[thick] (2) to (4);
				\draw[thick] (1.5,0.09) --++ (.9,0);
				\draw[thick] (1.5,-0.09) --++ (.9,0);
				\draw[thick]
					(1.9,0) --++ (60:.25)
					(1.9,0) --++ (-60:.25);
			\end{tikzpicture}\end{array}
		$} 
& 
		\scalebox{.8}{$\begin{array}{c}\begin{tikzpicture}
				\node[draw,circle,thick,scale=1.25,label=30:{2}] (0) at (0,0){};
				\node[draw,circle,thick,scale=1.25,label=below:{1}, fill=black] (1) at (-1,0){};
				\node[draw,circle,thick,scale=1.25,label=below:{1}] (2) at (1,0){};
				\node[draw,circle,thick,scale=1.25,label=above:{1}] (3) at (90:1){};      								\node[draw,circle,thick,scale=1.25,label=below:{1}] (4) at (90:-1){};
				\draw[thick] (0) to (1);
				\draw[thick] (0) to (2);
				\draw[thick] (0) to (3);
				\draw[thick] (0) to (4);
								\draw[<->,>=stealth',semithick,dashed]  (1.1,0.3) arc (25:65:1.7cm);
			\end{tikzpicture}\end{array}
		$} 
& 2 
				\\\hline

			\begin{array}{c}
			\text{I}^{*\text{s}}_{0}
			\\
			\widetilde{\text{D}}_4
\end{array}
 &			\scalebox{.8}{$\begin{array}{c}\begin{tikzpicture}
				\node[draw,circle,thick,scale=1.25,label=30:{2}] (0) at (0,0){};
				\node[draw,circle,thick,scale=1.25,label=below:{1}, fill=black] (1) at (-1,0){};
				\node[draw,circle,thick,scale=1.25,label=below:{1}] (2) at (1,0){};
				\node[draw,circle,thick,scale=1.25,label=above:{1}] (3) at (90:1){};      								\node[draw,circle,thick,scale=1.25,label=below:{1}] (4) at (90:-1){};
				\draw[thick] (0) to (1);
				\draw[thick] (0) to (2);
				\draw[thick] (0) to (3);
				\draw[thick] (0) to (4);
			\end{tikzpicture}\end{array}
		$} 

		& 
		\scalebox{.8}{$\begin{array}{c}\begin{tikzpicture}
				\node[draw,circle,thick,scale=1.25,label=30:{2}] (0) at (0,0){};
				\node[draw,circle,thick,scale=1.25,label=below:{1}, fill=black] (1) at (-1,0){};
				\node[draw,circle,thick,scale=1.25,label=below:{1}] (2) at (1,0){};
				\node[draw,circle,thick,scale=1.25,label=above:{1}] (3) at (90:1){};      								\node[draw,circle,thick,scale=1.25,label=below:{1}] (4) at (90:-1){};
				\draw[thick] (0) to (1);
				\draw[thick] (0) to (2);
				\draw[thick] (0) to (3);
				\draw[thick] (0) to (4);
			\end{tikzpicture}\end{array}
		$} 
		& \displaystyle{1}
	
				\\\hline

\end{array}
		$} 	
	\end{center}
	\caption{
{Fibers of type  I$_0^*$. }  \label{Table:DualGraph}.
}
\end{table}

\subsection{Geometric fibers, ground fields, and field extensions} The geometry and topology of an  I$^*_0$-model  is  subtle.  
To fully appreciate the geometry of a Weierstrass model describing the fiber  I$^*_0$, one has to take into account the scheme structure of the generic fiber. This scheme structure  contains more detailed information than can be seen by the type of the geometric fiber and impacts the values of the  topological invariants. 
Fiber types, generic fibers, and geometric irreducibility play a central role in this paper. See Appendix C of \cite{Euler} for a review.

The type of a singular fiber depends on the ground field used to define its scheme structure. 
For the  fiber over the generic point of an irreducible component of the discriminant locus, the natural ground field is the residue field of the generic point \cite[\S 3.1.2, Remark 1.17]{Liu}. 
This is because in scheme theory, the fiber over a point $p$ is by definition $Y\times_B \mathop{Spec}\,\kappa(p)$\, \cite[\S 3.1.2, Definition 1.13]{Liu} and the second projection $Y_p\to \mathop{Spec} \, \kappa(p)$  gives the fiber $Y_p$  the structure of a scheme defined with  the residue field $\kappa(p)$ as its ground field.  
As the residue field $\kappa(p)$ is not necessarily  an algebraically closed field, the fiber type of $Y_p$ as a scheme over $\kappa(p)$ does not always match its Kodaira type. 
We recall that a Kodaira fiber is by definition a geometric fiber.\footnote{ A geometric fiber is such that its  components are all geometrically irreducible, i.e.  they do  not  factor into more components even after a  field extension \cite[\S 3.2, Definition 2.8]{Liu}.} 
The Kodaira type of $Y_p$ is seen only after a field extension causing all components of the fiber to become geometrically irreducible. 
 In the case of Weierstrass models coming from Tate's algorithm, the required field extension is carefully described  in Tate's algorithm to be the splitting field of an appropriate cubic or quadratic polynomial defined from the Weierstrass coefficients.
For elliptic fibrations, the type of a generic fiber $Y_p$ that is not geometric is either an affine Dynkin diagram of  type $\widetilde{\text{A}}_1$, or a twisted affine Dynkin diagram of type 
$\widetilde{\text{G}}_2^t$, $\widetilde{\text{B}}_{3+n}^t$, $\widetilde{\text{F}}_4^t$, or $\widetilde{\text{C}}_{2+n}^t$.

   \subsection{G$_2$,  Spin($7$), and Spin($8$)-models}
   
   A fiber  I$^*_0$ consists of a rational curve of multiplicity two intersecting transversally with four other rational curves. Its dual graph is the affine Dynkin diagram $\widetilde{\text{D}}_4$.
One of the four is the component touching the section of the elliptic fibration. The points of intersection of the other three nodes with the node of multiplicity two can be described by a cubic polynomial that is essentially  
  the auxiliary polynomial $P(T)$.  
The elliptic fibration  is called the G$_2$, Spin($7$), or  Spin($8$)-model, respectively, if   $P(T)$ has no $\kappa$-rational roots,  a unique $\kappa$-rational root, or three distinct $\kappa$-rational roots. 
The generic fibers over $S$ are then  denoted by
 $$\text{I$_0^{*\text{ns}}$, I$_0^{*\text{ss}}$, and I$_0^{*\text{s}}$},$$
 where ``ns'', ``ss'', and ``s'' stand  for ``non-split'', ``semi-split'', and ``split'' \cite{Bershadsky:1996nh}. 
 The fibers I$_0^{*\text{ns}}$, I$_0^{*\text{ss}}$, and I$_0^{*\text{s}}$  defined with respect to the residue field $\kappa$ are called {\it arithmetic fibers}; this distinguishes them from their geometric fiber  I$_0^{*}$
 defined in the splitting field of $P(T)$. 
 The Galois group $\textnormal{Gal}(\kappa'/ \kappa)$  is  trivial  for Spin($8$)-models, $\mathbb{Z}/2\mathbb{Z}$ for Spin($7$)-models, and 
 can be either  the symmetric group $S_3$ or the cyclic group $\mathbb{Z}/3\mathbb{Z}$ for G$_2$-models. Thus, the Galois group provides a finer invariant than the number of rational solutions of  $P(T)$. We introduce the notion of  G$_2^{S_3}$-models and G$_2^{\mathbb{Z}/3\mathbb{Z}}$-models to distinguish between the two  cases  of a G$_2$-model as they have different fiber structures. One can think of the  G$_2^{\mathbb{Z}/3\mathbb{Z}}$-model as a specialization of the G$_2^{S_3}$-model in which the discriminant $\Delta(P)$ of the auxiliary polynomial $P(T)$ is a perfect square in the residue field $\kappa$.

\subsection{Crepant resolutions, flops, hyperplane arrangements, and Coulomb branches}

Each crepant resolution of a singular Weierstrass model is a relative minimal model (in the sense of the Minimal Model Program) over the Weierstrass model \cite{Matsuki.Weyl}. 
When the base of the fibration is a curve, the Weierstrass model has a unique crepant resolution. 
 When the base is of dimension two or higher, a crepant resolution does not always exist; furthermore, when it does, it is not necessarily unique.  
 Different crepant resolutions of the same Weierstrass model are connected by a finite sequence of  flops. 
 Following F-theory, we attach to a given elliptic fibration a Lie algebra $\mathfrak{g}$, a representation $\mathbf{R}$ of $\mathfrak{g}$, and a hyperplane arrangement I($\mathfrak{g},\mathbf{R}$). 
The Lie algebra $\mathfrak{g}$ and the representation $\mathbf{R}$  are determined by the fibers   over codimension-one  and codimension-two points, respectively, of the base in the discriminant locus.  
The hyperplane arrangement I($\mathfrak{g},\mathbf{R}$) is defined inside the dual fundamental Weyl chamber of $\mathfrak{g}$ (i.e. the  dual cone of the fundamental Weyl chamber of $\mathfrak{g}$), and its hyperplanes are the set of kernels of  the weights of $\mathbf{R}$.

The network of flops is studied using the hyperplane arrangement I($\mathfrak{g}, \mathbf{R}$) inspired from  the theory of Coulomb branches of five-dimensional supersymmetric gauge theories with eight supercharges \cite{IMS}. 
The network of crepant resolutions  is isomorphic  to 
the network of chambers of the hyperplane arrangement I$(\mathfrak{g}, \mathbf{R})$ defined by splitting the dual fundamental Weyl chamber of the Lie algebra $\mathfrak{g}$ by the hyperplanes dual to the weights of $\mathbf{R}$. 
  The hyperplane arrangement  I$(\mathfrak{g}, \mathbf{R})$, its relation to the Coulomb branches of supersymmetric gauge theories and the network of crepant resolutions are studied in 
 \cite{IMS,ESY1,ESY2,Diaconescu:1998cn,Hayashi:2014kca,EJJN1,EJJN2}. 

The algorithm that we use to determine the representation $\mathbf{R}$ works for any base of dimension two or higher and does not require us to impose the Calabi--Yau condition.
The algorithm consists of three steps. 
We start by identifying those vertical curves that are relative extremal curves appearing over divisors of $S$, over which the irreducible components of the generic fiber of $S$  degenerate. 
We then associate to each of these curves a weight  computed geometrically as minus of the intersections  of the  curve  with the fibral  divisors.  
Lastly, the representation $\mathbf{R}$ is then determined from these weights using the notion of saturated set of weights  introduced by Bourbaki \cite{Bourbaki.GLA79}.

The irreducible representations associated with the G$_2$-model are the adjoint representation $\bf{14}$, and the fundamental representation $\bf{7}$. The irreducible representations associated with the Spin($7$)-model are the adjoint representation $\bf{21}$, the vector representation $\bf{7}$, and the spinor representation $\bf{8}$. The irreducible representations for the Spin($8$)-model are the adjoint representation $\bf{28}$, the vector representation $\bf{8}_v$ and the two irreducible spinor representations $\bf{8}_s$ and $\bf{8}_c$. 
We note that all the representations are real. Moreover,  in each case we get the full list of fundamental representations.  
Moreover,  the $\bf{7}$ of G$_2$ and Spin($7$) are quasi-minuscule; the $\bf{8}$ of Spin($7$) is minuscule; the $\bf{8}_v$, $\bf{8}_s$, and $\bf{8}_c$ of Spin($8$) are minuscule while  the adjoint $\bf{28}$ of Spin($8$) is quasi-minuscule.

We show  that each  G$_2$-model has  a unique crepant resolution, each Spin($7$)-model has two crepant resolutions connected by a simple flop, and each 
Spin($8$)-model has six distinct crepant resolutions forming a cycle--- under conditions that ensure the existence of crepant resolutions. Figures \ref{fig:G2Phases}, \ref{fig:Spin7Phases}, and  \ref{fig:Spin8Phases} on page \pageref{fig:G2Phases} depicts the structure of the set of crepant resolutions in relationship to the geometry of the Coulomb branch.

When the Weierstrass coefficients vanish to high order along the component $S$ of the discriminant locus, $\mathbb{Q}$-factorial  terminal singularities are obstructions to  the existence of a crepant resolution depending on the dimension of the base $B$.  See section \ref{Sec:G2Z3Terminal} for the case of the G$_2$-model and section \ref{Sec:Spin7First} for the case of terminal singularities with the Spin($7$)-model.
In contrast, crepant resolutions were recently shown to exist in \cite{F4} for F$_4$-models for generic coefficients of arbitrary valuations, as long as the restrictions of Step 7 of Tate's algorithm are satisfied.
We also show by direct inspection that each of the crepant resolutions we obtained determines a flat fibration. 
 $\mathbb{Q}$-factorial  terminal singularities  have been discussed recently in F-theory in \cite{Arras:2016evy}.

In this paper, we consider elliptic fibrations over bases of arbitrary dimensions. For F-theory and M-theory applications, we focus mostly on compactifications yielding five and six-dimensional gauge theories with eight supercharges. In particular, we do not discuss Sen's weak coupling limit of these theories. However, 
we point out that  the weak coupling limit of G$_2$, Spin($7$), and Spin($8$)-models gives a local  $\mathfrak{so}$($8$)-gauge theory  realized by eight D7 branes on top of an O7 orientifold  \cite[Table 4]{Esole:2012tf}. Such an $\mathfrak{so}$($8$)-gauge theory can also be constructed using K-theory 
as in \cite[\S 4.2.4]{CDE}.

\subsection{Non-Kodaira fibers and Tate's algorithm in higher codimension}

The study of the fiber structure of G$_2$, Spin($7$), and Spin($8$)-models is surprisingly rich in non-Kodaira fibers. We get eight types of non-Kodaira fibers in the fiber degenerations of these elliptic fibrations. They are listed in Figure \ref{Fig:NonKodaira}.
To put this in perspective, we recall that in the theory of Miranda's models, there are only seven non-Kodaira fibers across all possible collisions, while here, 
the Spin($7$)-model on its own already produces six non-Kodaira fibers.  Miranda has observed that non-Kodaira fibers of  Miranda's models of elliptic threefolds \cite{Miranda.smooth}  are contractions of Kodaira fibers (see also  \cite{Cattaneo}). 
Here we see that this is also true in higher codimension for the crepant resolutions of the Weierstrass model resulting from Step 6 of Tate's algorithm. 

With a careful study of the crepant resolutions of the Weierstrass models resulting from Step 6 of Tate's algorithm, it is possible to predict (for Step 6) the possible higher codimension fibers from the valuations of the coefficients of the 
Weierstrass equation. 
In many cases, these valuations do not completely determine the fiber type: different crepant resolutions of the same Weierstrass model can have distinct fiber types over points of codimension two or higher.  

\subsection{Five and six dimensional theories with eight supercharges}

M-theory compactified on elliptically-fibered Calabi--Yau threefolds corresponding to I$_0^*$-models yields five-dimensional $\mathcal{N}=1$ supergravity theories with gauge groups G$_2$, Spin($7$), or Spin($8$). 
In each case, we determine the number of charged hypermuliplets transforming in each irreducible representation of the gauge group  
by comparing the polynomials of triple intersection of fibral divisors with the one-loop prepotential in the appropriate chamber of the Coulomb branch. 

F-theory compactified on the same threefold yields a six-dimensional $\mathcal{N}=(1,0)$ theory with the same gauge group. 
 We consider local anomalies of G$_2$, Spin($7$), and Spin($8$)-models to constrain  the number of charged hypermultiplets in each irreducible representations of the gauge group.

We see that the number of representations found for five-dimensional theories are compatible with an anomaly-free six-dimensional theory. 
The number of charged hypermultiplets we find  matches the expectation from Witten's quantization method generalized by Aspinwall, Katz and Morrison \cite{Aspinwall:2000kf}.
In the G$_2$ and Spin($7$)-model, the anomaly cancellation conditions of the six-dimensional theory reproduce exactly the counting of charged hypermultiplets derived in the five-dimensional theory using triple intersection numbers. 
However, in the Spin($8$)-model, we cannot completely determine the number of charged hypermultiplets transforming in the representation $\mathbf{8_c}$ and $\mathbf{8_s}$, but their sum is determined.

\subsection{Road map to the rest of the paper}

The rest of the paper is organized as follows. In section \ref{Sec:preliminaries}, we introduce  some basic definitions and fix our conventions in section \ref{Sec:DefConv} and review Step 6 of Tate's algorithm in section \ref{Sec:Tate}.  
In  section \ref{Sec:summary} we summarize the main mathematical results of the paper. 
We first derive  canonical forms for G$_2$, Spin($7$), and Spin($8$)-models in section \ref{Sec:Canonical}  by scrutinizing  Step 6 of  Tate's algorithm.
 We distinguish between two types of G$_2$-models using the Galois group of the splitting field of the associated polynomial $P(T)$ used in Step 6 of Tate's algorithm. 
They have  distinct fiber structures and $j$-invariants over $S$. 
 We also distinguish between two  Weierstrass models for Spin($7$)-models  by their fiber degenerations and their $j$-invariants over $S$.  
  We then study the existence of crepant resolutions for these canonical forms and determine the  fiber structure for each resolution. 
We compute the Chern--Simons coefficients as the  triple intersection numbers of the fibral divisors in section \ref{Sec:Triple}. 
We study G$_2$-models in section \ref{Sec:G2}, Spin($7$)-models in section \ref{Sec:Spin7}, and Spin($8$)-models in section \ref{Sec:Spin8}.
 In section \ref{Sec:Phys}, we apply the results collected in the previous sections to  describe the five-dimensional $\mathcal{N}=1$ theories regarding the I$^*_0$ models. We also count the number of representations of each model by comparing the triple intersection numbers and the cubic prepotential for five-dimensional theories. In contrast, 
in section \ref{sec:anomaly6d}, 
we derive from six-dimensional $\mathcal{N}=(1,0)$ theories with gague groups G$_2$, Spin($7$), and Spin($8$) the number of multiplets that cancels anomalies. We review the generalities of 6d anomaly cancellations and find the countings of matter contents. We consider global anomalies for G$_2$-models in section \ref{sec:6dgravG2}.

\section{Preliminaries} \label{Sec:preliminaries}

In this section, we introduce our conventions and some basic definitions   in  section \ref{Sec:DefConv} and review Step 6 of Tate's algorithm in section \ref{Sec:Tate}.  

\subsection{Definitions and conventions}\label{Sec:DefConv}

By a variety, we mean an irreducible algebraic variety over the complex numbers.
Given a morphism $Y \to B,$ we denote by $Y_p$ the fiber of $Y$ over a point $p$.
Let $S=V(s)\subset B$ be an effective Cartier divisor of  $B$ given by the  zero scheme of a section $s$ of a line bundle $\mathscr{S}$. 
We assume that $S$ is a smooth irreducible variety and denote its generic point by $\eta$. 
The function field at $\eta$ is a local field $\mathscr{O}_{B,\eta}$ with maximal ideal  $\mathfrak{m}_\eta$ and  residue field  $\kappa_\eta=\mathscr{O}_{B,\eta}/ \mathfrak{m}_\eta$. 
When there is  no possibility for confusion, we simply write $\kappa_\eta$ as $\kappa$. 
The triplet ($\mathscr{O}_{B,\eta}, \mathfrak{m}_\eta, \kappa_\eta$) defines a discrete valuation ring induced by $S$, and we denote its valuation by  $v_S$ and take $s$ as a uniformizing parameter. 
A function $f\in \mathscr{O}_{B,\eta}$ has valuation $n$ (i.e. $v_S(f)=n$) if and only if $n$ is the order of the zero of $f$ at $s=0$; if $f$ has a pole at $s=0$, the valuation is negative. Furthermore, by definition,  $v_S(0)=\infty$ and $v_S(1)=0$. 

A genus-one fibration is a surjective proper morphism $\varphi: Y\to B$  between algebraic varieties  such that the generic fiber is a smooth curve of genus one. 
An elliptic fibration is  a genus-one fibration  endowed with  a  rational section, which is a rational map $\sigma: B \dashrightarrow Y$ such that $\varphi\circ \sigma$  is the identity when restricted to the domain of $\sigma$.
We will also assume that the base $B$  is a smooth  (quasi)-projective variety, and that the Mordell-Weil group of the fibration (i.e. the group of rational sections) is trivial. 
The locus $\Delta$ of points $p$ such that the fiber $Y_p$ is singular is called the discriminant locus of the fibration. Note that a smooth fiber over a closed point is characterized (up to isomorphism) by its $j$-invariant.

Any elliptic fibration over  a smooth variety defined with an algebraically closed field is birational to a (possibly singular) Weierstrass model \cite{Deligne.Formulaire,MumfordSuominen,Nakayama.OWM}. 
For Weierstrass models, we use the notation of Tate as presented in Deligne's formulaire \cite{Deligne.Formulaire}. 
Tate's algorithm determines the type of the geometric fiber over the special point of a Weierstrass model over a discrete valuation ring by manipulating the coefficients of the Weierstrass model \cite{Tate.Alg}. 
There are well-known formulas to compute the discriminant locus $\Delta$ and  the $j$-invariant of a Weierstrass model \cite{Deligne.Formulaire}.  
The type of singular fibers over a generic point of a divisor $S$ of the base are classified by Kodaira and N\'eron. 
 We denote the type of the  geometric generic fiber over an irreducible component of the discriminant locus by one of the Kodaira symbols  I$_{n}$, IV, III,  II, I$_{n}^*$,  IV$^*$, III$^*$, II$^*$ \cite{Kodaira}.

  An algebraic $d$-cycle of a Noetherian scheme $X$ is an element of the free group $Z_d(X)$ generated by irreducible algebraic subvarieites of  $X$ of dimension $d$. 
 Following Kodaira \cite{Kodaira}, we introduce the following definition. 

 \begin{defn}[Fiber type]
The {\em type} of an algebraic 1-cycle $\Theta\in Z_1(X)$  with irreducible decomposition $\Theta=\sum_i m_i \Theta_i$ consists of the isomorphism class of each irreducible curve $\Theta_i$, together with the topological structure of the reduced polyhedron $\sum_i \Theta_{i}$. 
\end{defn}

The topological structure of the polyhedron $\sum_i \Theta_{i}$ is characterized by the underlying set of the scheme-theoretic intersection $\Theta_i  \medcap \Theta_j$ ($i\neq j$) of the irreducible components $\Theta_i$.

\begin{defn}[Dual graph]
Given an algebraic one-cycle $\Theta$ with irreducible decomposition  
$\Theta=\sum_i m_i \Theta_{i}$,  we associate a  weighted graph (called the {\em dual graph} of $\Theta$) such that
\begin{itemize}
\item 
 the vertices are the irreducible components $\Theta_i$ of the fiber,
\item
 the weight of the vertex  corresponding to an irreducible component $\Theta_i$ is its multiplicity  $m_i$
\item 
 the vertices corresponding to the irreducible components $\Theta_i$ and $\Theta_j$ ($i\neq j$) are connected by $\hat{\Theta}_{i,j}=\mathrm{deg}[\Theta_i \medcap \Theta_j]$ edges.\footnote{The degree $\mathrm{deg}(\alpha_0)$ of a zero-cycle $\alpha_0$ is defined by  passing to the  Chow ring  and using the degree defined in  \cite[Definition 1.4]{Fulton}.}
\end{itemize}
\end{defn}

The type of the geometric fiber over a codimension one point of a minimal elliptic fibration is called the {\em Kodaira type} of the fiber. 
 As shown by Kodaira \cite{Kodaira} and N\'eron \cite{Neron}, there are 10 Kodaira types, and we denote them with the notation of Kodaira.
The dual graph of a Kodaira fiber is always an affine Dynkin diagram of type ADE (see Table \ref{Table.KodairaNeron}), while  the dual graph of a generic fiber itself can be the affine Dynkin diagram of a non-simply-laced simple Lie algebra. 

Let $G$ be a simply-connected  simple Lie group with Dynkin diagram $\mathfrak{g}$. 
 We denote by $\widetilde{\mathfrak{g}}$ the affine Dynkin diagram that reduces, upon removal of its extra node, to the Dynkin diagram $\mathfrak{g}$. Following  Carter \cite{Carter},  we write  $\widetilde{\mathfrak{g}}^t$ for its Langlands dual, namely,  the twisted Dynkin diagram whose Cartan matrix is the transpose of the Cartan matrix of $\widetilde{\mathfrak{g}}$.
 In particular, $\widetilde{\mathfrak{g}}$ and $\widetilde{\mathfrak{g}}^t$ are distinct only when $\mathfrak{g}$ is not simply laced (that is, when $\mathfrak{g}$=B$_k$, C$_k$, G$_2$, or F$_4$).

\begin{table}[tbh]
\begin{center}
$
\begin{array}{|c|c|c|c|c|c|c|c|c|c|c|c|}
\hline 
\text{Tate's Step} & 1 & 2  & 3 & 4 & 5  & 6 & 7 & 8 & 9 & 10 \\
\hline
\text{Kodaira's symbol} & \text{I}_0 & \text{I}_{n>0} & \text{II} & \text{III} & \text{IV} &  \text{I}_0^* & \text{I}_{n>0}^* & \text{IV}^* & \text{III}^* & \text{II}^* \\
\hline
\text{N\'eron's type} & \text{A} & \text{B}_n & \text{C}_1 & \text{C}_2 & \text{C}_3 &  \text{C}_4 & \text{C}_{5,n} & \text{C}_6 & \text{C}_7 & \text{C}_8 \\
\hline 
\vrule width 0pt height 3ex 
\text{Dual graph} & \text{--} & \widetilde{\text{A}}_{n-1} &  \widetilde{\text{A}}_1 &  \widetilde{\text{A}}_1 &  \widetilde{\text{A}}_2 &   \widetilde{\text{D}}_4 &  \widetilde{\text{D}}_{4+n} &  \widetilde{\text{E}}_6 &  \widetilde{\text{E}}_7 &  \widetilde{\text{E}}_8 \\
\hline
v(\Delta)& 0 & n & 2& 3 & 4 & 6& 6+n &   8  &  9 & 10\\
\hline 
v(j)&  0 & -n & 0 & 0 & 0 & 0 & -n & 0 & 0 & 0 \\
\hline
j &  j\in \mathbb{C} &j=\infty & j=0 & j=1728 & j=0 & j\in \mathbb{C} & j=\infty & j=0 & j=1728 & j=0 \\
\hline 
\end{array}
$
\end{center}
\caption{ Kodaira-N\'eron classification. When working over an algebraically closed field, the fiber type is determined uniquely by the valuation of the discriminant  $\Delta$ and the $j$-invariant. \label{Table.KodairaNeron} }
\end{table}

 \begin{defn}[$\mathcal{K}$-model]
 Let $\mathcal{ K}$ be a fiber type.  Let $S\subset B$ be a smooth divisor of a projective variety $B$. 
 An elliptic fibration  $\varphi:Y\longrightarrow B$ over $B$  is said to be  a  \emph{$\mathcal{K}$-model} if
the discriminant locus $\Delta(\varphi)$ contains as an irreducible component a divisor  $S\subset B$ such that the generic fiber over $S$ is of type $\mathcal{ K}$ and any other  fiber away from $S$ is irreducible.
 \end{defn}

 The locus of  points in the base that lie below singular fibers of a non-trivial elliptic fibration is a Cartier divisor $\Delta$ called the discriminant locus of the elliptic fibration. 
We denote the irreducible components of the reduced discriminant by $\Delta_i$. 
If the elliptic fibration is minimal, the type of the  fiber over the generic point of  $\Delta_i$ of $\Delta$ has a dual graph that is  an affine Dynkin diagram $\widetilde{\mathfrak{g}}_i^t$.  If the generic fiber over $\Delta_i$ is irreducible, $\mathfrak{g}_i$ is the trivial Lie algebra since $\widetilde{\mathfrak{g}}_i^t=\widetilde{A}_0$.
The Lie algebra   $\mathfrak{g}$ associated with the elliptic fibration $\varphi:Y\to B$ is then  the direct sum 
$\mathfrak{g}=\bigoplus_i \mathfrak{g}_i, $
where the Lie algebra $\mathfrak{g}_i$ is such that the affine Dynkin diagram  $\widetilde{\mathfrak{g}}^t_i$ is the dual graph of the fiber over the generic point of $\Delta_i$.

When an elliptic $\varphi: Y \to B$ has trivial Mordell-Weil group, the compact Lie group $G$ associated with the elliptic fibration  $\varphi$ is semisimple and simply connected and is given by the formula $
G:= \exp (\bigoplus_i \mathfrak{g}_i)$,
where the index $i$ runs over all the irreducible components of the reduced discriminant locus, 
the Lie algebra $\mathfrak{g}_i$ such that the affine Dynkin diagram  $\widetilde{\mathfrak{g}}^t_i$ is the dual graph of the fiber over the generic point of the irreducible component $\Delta_i$ of the reduced discriminant 
of the elliptic fibration, and
$\exp(\bigoplus_i \mathfrak{g}_i)$ is the compact simply connected Lie group whose Lie algebra is $\bigoplus_i \mathfrak{g}_i$.

\begin{defn}[$G$-model]\label{Defn:Gmodel}
Let $G$ be a compact, simply connected Lie group.
An elliptic fibration  $\varphi:Y\to B$ with an associated Lie group $G$ and trivial Mordell-Weil group is called a \emph{$G$-model}.
\end{defn}

Given a $G$-model, the generic fiber over $S$ is an affine Dynkin diagram. The generic fiber degenerates into a fiber of different type  over points of codimension two in the base.

Given an elliptic fibration $\varphi:Y\to B$, if $S$ is an irreducible component of the discriminant locus,
the generic fiber over $S$ can degenerate further over subvarieties of $S$. We distinguish between two
types of degenerations. A degeneration is said to be {\em  arithmetic} if it modifies the type of the
fiber without changing the type of the geometric fiber. A degeneration is said to be {\em geometric} if it
modifies the geometric type of the fiber.

The irreducible curves of the degenerations over codimension-two loci give weights of a representation $\mathbf{R}$. However, they only give a subset of weights. Hence, we need an algorithm that retrieves the full representation $\mathbf{R}$ given only a few of its weights.  
This problem can be addressed systematically using the notion of a saturated set of weights introduced by  Bourbaki  \cite[Chap.VIII.\S 7. Sect. 2]{Bourbaki.GLA79}.

\begin{defn}[Saturated set of weights]
A set $\Pi$ of integral weights is {\em saturated} if for any weight $\varpi\in\Pi$ and any simple root $\alpha$,  the weight $\varpi-i\alpha$ is also in $\Pi$ for any $i$ such that $0\leq i\leq \langle \varpi,\alpha\rangle$. 
A saturated set has {\em highest weight} $\lambda$ if $\lambda\in\Lambda^+$ and $\mu\prec\lambda$ for any $\mu\in \Pi$. 
\end{defn}
\begin{defn}[Saturation of a subset]
Any subsets $\Pi$ of weights is contained in a unique smallest saturated subset. We call it the saturation of $\Pi$. 
\end{defn}
\begin{prop}\quad 
\begin{enumerate}[label=(\alph*)]
\item A saturated set of weights is invariant under the action of the Weyl group.  
\item The saturation of a set of weights $\Pi$ is finite if and only if  the set $\Pi$ is finite.
\item A saturated set with highest weight $\lambda$  consists of all dominant weights lower than or equal to $\lambda$ and their conjugates under the Weyl group.  
\end{enumerate}
\end{prop}
\begin{proof}
See   \cite[Chap. III \S 13.4]{Humphreys}.
\end{proof}

\begin{thm}\label{Thm:R-Saturation}
Let $\Pi$ be a finite saturated  set of weights. Then there exists a finite dimensional  $\mathfrak{g}$-module  whose set of weights is $\Pi$. 
\end{thm}
\begin{proof}
\cite[Chap.VIII.\S 7. Sect. 2, Corollary to Prop. 5]{Bourbaki.GLA79}.
\end{proof}

\begin{defn}[Weight vector of a vertical curve]\label{Def:WeightVerticalCurve}
Let $C$ be a vertical curve, i.e.  a curve contained in a fiber of the elliptic fibration.
Let $S$ be an irreducible component of the reduced discriminant of the elliptic fibration $\varphi: Y\to B$. 
The pullback of  $\varphi^* S$ has irreducible components $D_0, D_1, \ldots, D_n$, where $D_0$ is the component touching the section of the elliptic fibration.   
The {\em weight vector} of $C$ over $S$ is  by definition the vector ${\varpi}_S(C)=(-D_1\cdot C, \ldots, -D_n\cdot C)$ of intersection numbers $D_i\cdot C$ for $i=1,\ldots, n$. 
\end{defn}
\begin{defn}[Representation of  a $G$-model]\label{Def.Rep}
To a $G$-model, we associate a representation $\mathbf{R}$ of the Lie algebra $\mathfrak{g}$ as follows. 
 The weight vectors of the irreducible vertical   rational curves of the fibers over codimension-two points form  a set $\Pi$  whose  saturation defines uniquely a representation  $\mathbf{R}$ by  Theorem \ref{Thm:R-Saturation}. 
  We call this representation  $\mathbf{R}$  the representation of the $G$-model. 
 \end{defn}
Definition \ref{Def.Rep}  is essentially a formalization of the method of Aspinwall and Gross \cite[\S 4]{Aspinwall:1996nk}.
Note that we always get the adjoint representation as a summand of $\mathbf{R}$.
 There are subtleties when the divisor $S$ is singular \cite{Anderson:2015cqy,Klevers:2016jsz}.

\begin{defn}[Weierstrass model]
Consider a variety $B$ endowed with a line bundle $\mathscr{L}\rightarrow B$.
A Weierstrass model $\mathscr{E} \rightarrow B$ over $B$ is a hypersurface  cut out by the zero locus of a section of the line bundle of $\mathscr{O}(3)\otimes\pi^* \mathscr{L}^{\otimes 6}$ in the projective bundle  $\mathbb{P}(\mathscr{O}_B \oplus\mathscr{L}^{\otimes 2} \oplus\mathscr{L}^{\otimes 3})\rightarrow B$. 
We denote by $\mathscr{O}(1)$  the dual of the tautological line bundle of the projective bundle, and denote by $\mathscr{O}(n)$ ($n>0$) its $n$th-tensor product.  
The relative projective coordinates of the $\mathbb{P}^2$ bundle are denoted by $[x:y:z]$. In particular,  $x$ is a section of $\mathscr{O}(1)\otimes\pi^* \mathscr{L}^{\otimes 2}$, $y$ is a section of $\mathscr{O}(1)\otimes\pi^* \mathscr{L}^{\otimes 3}$, and $z$ is a section of $\mathscr{O}(1)$. 
 Following Tate and Deligne's notation, the defining equation of a  Weierstrass model is
\begin{equation}
\mathscr{E}: \quad  zy(y +a_1 x+ a_3 z) -(x^3 +a_2 x^2 z+a_4 xz^2 +a_6 z^3)=0,
\end{equation}
 where the coefficient $a_i$ ($i=1,2,3,4,6$) is a section of $\mathscr{L}^{\otimes i}$ on $B$. Such a hypersurface is an elliptic fibration since over the generic point of the base, the fiber is a nonsingular  cubic planar curve with a rational point ($x=z=0$). 
    We use the convention of  Deligne's formulaire \cite{Deligne.Formulaire}: 
\begin{equation}
\begin{aligned}
b_2 &= a_1^2+ 4 a_2,\quad
b_4 = a_1 a_3 + 2 a_4 ,\quad
b_6  = a_3^2 + 4 a_6 , \quad
b_8  =b_2 a_6 -a_1 a_3 a_4 + a_2 a_3^2-a_4^2,\\
c_4 & = b_2^2 -24 b_4, \quad
c_6 = -b_2^3+ 36 b_2 b_4 -216 b_6,\\
\Delta &= -b_2^2 b_8 -8 b_4^3 -27 b_6^2 + 9 b_2 b_4 b_6,\quad  j =\frac{c_4^3}{\Delta}.
\end{aligned}
\end{equation}
These quantities are related by the  relations
\begin{equation}
4 b_8 =b_2 b_6 -b_4^2 \quad \text{and}\quad 1728 \Delta=c_4^3 -c_6^2. 
\end{equation}
\end{defn}

 The discriminant locus is the subvariety of $B$ cut out by the equation $\Delta=0$, and is the locus of points  $p$ of the base $B$ such that the fiber over $p$ (i.e. $Y_p$)  is singular. 
 Over a generic point of $\Delta$, the fiber is a nodal cubic that  degenerates to a cuspidal  cubic over the codimension-two locus $c_4=c_6=0$. Up to isomorphisms, the
 $j$-invariant $j=c_4^3/ \Delta$ uniquely characterizes nonsingular elliptic curves.

\begin{defn}[Resolution of singularities]
A resolution of singularities of a variety $Y$ is a proper birational morphism $\varphi:\widetilde{Y}\longrightarrow Y$  such that  
$\widetilde{Y}$ is nonsingular
and  $\varphi$ is an isomorphism away  from the singular  locus of $Y$. In other words, $\widetilde{Y}$ is nonsingular and  if $U$ is the singular locus of $Y$, $\varphi$ maps $\varphi^{-1}(Y\setminus U)$isomorphically  onto $Y\setminus U$.  
\end{defn}

\begin{defn}[Crepant birational map]
A  birational map $\varphi:\widetilde{Y}\to Y$ between two algebraic varieties with  $\mathbb{Q}$-Cartier canonical classes is said to be {\em crepant} if it preserves the canonical class, i.e.  
$
K_{\widetilde{Y}}=\varphi^\ast K_Y.
$
\end{defn}

\subsection{Step 6 of Tate's algorithm}\label{Sec:Tate}

Tate's algorithm consists of eleven steps (see \cite{Tate.Alg},  \cite[\S IV.9]{Silverman.II}, \cite{Papadopoulos}, and \cite[\S 4.8]{Szydlo.Thesis}).\footnote{
In F-theory, Tate's algorithm is discussed in \cite{Bershadsky:1996nh,Morrison:1996pp,Katz:2011qp}, but usually focuses on the minimal valuations of the coefficients of the Weierstrass equation. 
Hence, we use instead the original paper of Tate  \cite{Tate.Alg}, 
 which contains some typos that are listed and corrected in
 in   \cite{Papadopoulos}. }
Step 6 of Tate's algorithm characterizes the fiber of Kodaira type I$_0^*$. 
We start with a general Weierstrass equation 
\begin{equation}
y^2 z+ a_1 xy z + a_3 yz^2=x^3+a_2 x^2 z +a_4 x z^2 +a_6 z^3.
\end{equation}
N\'eron proved in \cite{Neron} that a Weierstrass model over a discrete valuation ring has a special fiber of type  I$_0^*$ (denoted by C$_4$ in N\'eron's notation) if and only if the 
discriminant 
\begin{equation}
\Delta(x^3+a_2 x^2  +a_4 x  +a_6 )=-4 a_2^3 a_6+a_2^2a_4^2 +18 a_2a_4 a_6 -4 a_4^3-27 a_6^2
\end{equation}
  has valuation $6$ with respect to $S$, and the valuation of the Weierstrass 
coefficients satisfies the following inequalities:
\begin{equation}\label{Neron.val.I0*}
v_S(a_1)\geq 1, \quad v_S(a_2)\geq 1, \quad v_S(a_3)\geq 2, \quad v_S(a_4)\geq 2, \quad v_S(a_6)\geq 3.
\end{equation}
This case is further studied by Tate in Step 6 of his algorithm \cite{Tate.Alg}. Tate defines the following auxiliary cubic polynomial in  the polynomial ring $\kappa[T]$ to be
\begin{equation}\label{Eq:Cubic}
P(T)= T^3+ a_{2,1} T^2 + a_{4,2} T + a_{6,3},
\end{equation}
where $a_{i,j}= a_i / s^{j}$. 
The splitting field of the cubic $P(T)$ in $\kappa$ is denoted by $\kappa'$. 
The discriminant $\Delta(P)$ of $P(T)$ is 
\begin{equation}
\Delta(P):=4 a_{2,1}^3 a_{6,3}-a_{2,1}^2 a_{4,2}^2-18 a_{2,1} a_{4,2} a_{6,3}+4 a_{4,2}^3+27 a_{6,3}^2.
\end{equation}
The polynomial $P(T)$ is separable in $\kappa$ if and only if $P(T)$ has three distinct roots in $\kappa'$. 
This is the case if and only if the discriminant  $\Delta(P)$ of  $P(T)$ has valuation zero. 
This condition is equivalent to N\'eron's requirement discussed above. 
In view of the inequalities  in equation \eqref{Neron.val.I0*}, the discriminant of the full Weierstrass equation has valuation $6$ if and only if the valuation of $\Delta(P)$ is zero. 
The type of the geometric fiber is the same as the type of the fiber as seen in the splitting field $\kappa'$ of the polynomial $P(T)$ in $\kappa$. 
 The type of the special fiber as a scheme over $\kappa$ depends on the degree $[\kappa':\kappa]$ of the field extension $\kappa'/\kappa$: 
\begin{center}
\begin{tabular}{lll}
$\bullet$ $[\kappa':\kappa]=6\implies$ I$_0^{*\text{ns}}$ & with Galois group  $S_3$  and dual graph $\widetilde{G}_2^t$,\\
$\bullet$ $[\kappa':\kappa]=3\implies$ I$_0^{*\text{ns}}$ & with Galois group  $\mathbb{Z}/3 \mathbb{Z}$ and dual graph $\widetilde{G}_2^t$,\\
$\bullet$   $[\kappa':\kappa]=2\implies$ I$_0^{*\text{ss}}$ & with Galois group $\mathbb{Z}/2 \mathbb{Z}$ and dual graph $\widetilde{\text{B}}_3^t$,\\
$\bullet$    $[\kappa':\kappa]=1\implies$ I$_0^{*\text{s}}$ & with trivial Galois group and dual graph $\widetilde{\text{D}}_4$,
\end{tabular}
\end{center}
 where $``{}ns"$, $``{}ss"$, and $``{}s"$ stand  for {\em non-split}, {\em semi-split}, and {\em split}, respectively \cite{Bershadsky:1996nh}.\footnote{In the notation of Liu \cite[\S 10.2]{Liu}, the fibers I$_0^{*\text{ns}}$, I$_0^{*\text{ss}}$, and I$_0^{*\text{s}}$  are denoted by I$^*_{0,3}$,  I$^*_{0,2}$, and I$^*_{0}$, respectively.}

 If the discriminant of  $P(T)$ does not have a $\kappa$-rational root, then the fiber is of type I$^{*\text{ns}}_{0}$ with dual graph $\widetilde{G}_2^t$. 
 If  $P(T)$ has a unique $\kappa$-root, then the fiber is of type I$^{*\text{ss}}_0$ with dual graph $\widetilde{\text{B}}_3^t$. 
 If  $P(T)$ has three  $\kappa$-roots, then the fiber is of type I$^{*\text{s}}_0$ with dual graph $\widetilde{\text{D}}_4$.

There are two cases of fibers I$^{*\text{ns}}_0$, depending on whether the Galois group of the field extension $\kappa'/\kappa$ is either  $\mathbb{Z}/3\mathbb{Z}$ or $S_3$. 
The two cases differ  by the behavior of   the discriminant of $P(T)$ in view of the following well-known theorem. 
\begin{lem}[Galois group of a cubic polynomial]\label{lem:Gal.Cubic}The Galois group of the splitting field of a separable cubic polynomial $P(T)$, defined over a  field $\kappa$ of characteristic different from $2$ and $3$,  is 
\begin{enumerate}
\item $S_3$ if and only if the $P(T)$ is $\kappa$-irreducible  and its  discriminant  is not a perfect square. 
\item $\mathbb{Z}/3\mathbb{Z}$ if and only if $P(T)$ is  $\kappa$-irreducible  and its  discriminant  is a perfect square. 
\item   $\mathbb{Z}/2\mathbb{Z}$ if and only if  $P(T)$ factorizes into a linear factor and a $\kappa$-irreducible quadric. 
\item the trivial group  if and only if $P(T)$ factorizes into three linear factors over $\kappa$. 
\end{enumerate}
\end{lem}
\begin{proof}  See \cite[Chap. VI \S 2]{Lang}. \end{proof}

Lemma \ref{lem:Gal.Cubic} provides a direct route to the classification of  fibers of type I$_0^{*\text{s}}$, I$_0^{*\text{ss}}$, and I$_0^{*\text{ns}}$ using the Galois group of the splitting field of $P(T)$. 
A more refined classification will also take into account the values of the $j$-invariant. 
  In contrast to other singular Kodaira fibers, the fiber of type I$^*_0$ can take any finite value of the $j$-invariant.
However, for the arithmetic fibers, there are some restrictions. For example, a fiber of type I$_0^{*\text{ss}}$  cannot have $j=0$. 

\section{Summary of results}\label{Sec:summary}
 In this section, we  summarize the main mathematical results of the paper.
  In section \ref{Sec:Phys}, we discuss the physical implications of our results for five- and six-dimensional gauge theories obtained from of F-theory and M-theory compactifications.

\subsection{Canonical forms and crepant resolutions for G$_2$, Spin($7$), and Spin($8$)-models}\label{Sec:Canonical}

We summarize the canonical forms for G$_2$, Spin($7$), and Spin($8$)-models. 
 These forms are derived in Theorems   \ref{Thm:G2Can},  \ref{Thm:Spin7Can}, and  \ref{Thm:Spin8Normal}. 
We assume that the Mordell-Weil group of the elliptic fibrations is trivial. 
Using Step 6 of Tate's algorithm, we obtain the Weierstrass model  describing a  Kodaira fiber of type I$^*_0$ over the generic point of $S=V(s)$  where $s$ is a section of the line bundle $\mathscr{S}=\mathscr{O}_B(S)$. This model takes one of the following forms below, which are organized by the fiber type when seen as a scheme over the residue field of $S$, as well as  the value of the $j$-invariant: 
\begin{align}
\bullet\quad \textnormal{G}_2^{S_3}:& \quad y^2z-x^3-s^{2} f xz^2 -s^3 gz^3=0, \quad fg\neq 0, 
\quad  j\neq 0, 1728,\\
\intertext{where $4f^3+27 g^2$ is not a perfect square modulo $s$. }
\bullet\quad \textnormal{G}_2^{\mathbb{Z}/3\mathbb{Z}}:& \quad y^2z-x^3-s^{2} f xz^2 -s^3 gz^3=0, \quad fg\neq 0, \quad  j\neq 0, 1728,\\
\intertext{where $4f^3+27 g^2$ is a perfect square modulo $s$. In particular, the following G$_2^{\mathbb{Z}/3\mathbb{Z}}$-model is uniquely specified by the valuations of the coefficients:}
\bullet\quad \textnormal{G}_2^{\mathbb{Z}/3\mathbb{Z}}:& \quad y^2z-x^3-s^{3+\alpha} f x z^2 -s^3 gz^3=0, \quad \alpha\in \mathbb{Z}_{\geq 0}\quad  j=0,\\
\intertext{where $g$ is not a cube modulo $s$; }
\bullet\quad \textnormal{Spin($7$)}: &\quad y^2 z -( x^3 + a_{2,1} s x^2  z+ a_{4,2} s^2 x z^2+ a_{6,4} s^{4+\beta} z^3)=0,\quad \beta\in \mathbb{Z}_{\geq 0}, \quad j\neq 0, 1728,
\intertext{where $a_{2,1}$ is not proportional to $s$;}
\bullet\quad \textnormal{Spin($7$)}: & \quad y^2 z-x^3-s^{2} f x z^2-s^{4+\beta} g z^3=0,\quad\beta\in \mathbb{Z}_{\geq 0},\quad \quad j=1728,\\
\intertext{where $f$ is not a square modulo $s$; }
\bullet\quad \textnormal{Spin($8$)}: &
\quad
zy^2=(x-s x_1 z)(x-s x_2 z)(x+s x_3z)- s^{2+\alpha} Qz, \quad Q= r  x^2 +q s x z-s^2 t z^2,\quad \alpha\in \mathbb{Z}_{\geq 0},
\end{align}
where $ \quad v_S((x_1-x_2)(x_1-x_3)(x_2-x_3))=0$ and  $(r,q,t)\neq (0,0,0)$.
The $j$-invariant of a Spin($8$)-model is $j= 1728\frac{A^3}{A^3-B^2}= 1728 +1728\frac{B^2}{A^3-B^2}$, where
$A =-16(-x_1^2-x_2^2-x_3^2 +x_1x_2+x_1 x_3+x_2x_3)$ and 
$B=32(-2 x_1+x_2+x_3) (x_1-2 x_2+x_3) (x_1+x_2-2x_3)$.
 In particular,   $j=0$ when $A=0$ and $j=1728$ when $B=0$.

In contrast to the usual forms seen in the F-theory literature, we do not restrict to the  minimal values for the valuations of the coefficients. 
Allowing more general valuations enables  a richer set of  behaviors for the degeneration and  the $j$-invariant. In the case of Spin($7$), depending on the dimension of the base, the valuations can become an obstruction to the existence of a crepant resolution when they are too big. 

\subsection{Crepant resolutions}\label{Sec:Crepant}

\
 The following sequences of blowups provide crepant resolutions for the  Weierstrass models defined in section  \ref{Sec:Canonical}. These are, however, valid only  under some conditions which will be discussed in Theorems    \ref{Thm:G2Res},  \ref{Thm:Spin7Res1}, and \ref{Thm:Spin8Res}. 
In some cases, the non-minimal valuations obstruct the existence of a crepant resolution. 
We assume that the coefficients of the Weierstrass models are general except for their valuations  with respect to $S$. 
To prove smoothness, we also have to impose some light conditions on the coefficients that are usually left unspecified in the F-theory literature.

\begin{equation}
\label{eq:BlowUps}
\begin{array}{rl}
{\textnormal{ G$_2$}}  & \quad 
{ \begin{tikzcd}[column sep=huge] 
  X_0  \arrow[leftarrow]{r} {(x,y,s|e_1)} & \arrow[leftarrow]{r}{(y,e_1|e_2)} X_1 &X_2  
  \end{tikzcd}}\\
{\textnormal{Spin($7$)}} & \quad  \begin{tikzcd}[column sep=huge] 
  & & & X_3^+ \arrow[dashed,leftrightarrow]{dd} \\
  X_0  \arrow[leftarrow]{r} {(x,y,s|e_1)} & \arrow[leftarrow]{r}{(y,e_1|e_2)} X_1 & \arrow[leftarrow]{ru}{(e_2, x|e_3)} \arrow[leftarrow]{rd} [left,midway]{(e_2, Q|e_3)} X_2  &\\
  & & & X_3^-\vspace{.3cm}
  \end{tikzcd}\\
{\textnormal{Spin($8$)}}  & \quad  \begin{tikzcd}[column sep=huge] 
  X_0  \arrow[leftarrow]{r} {(x,y,s|e_1)} & \arrow[leftarrow]{r}{(y,e_1|e_2)} X_1 & X_2 \arrow[leftarrow]{r}{(x-x_i s z,e_2|e_3)} & X_3 \arrow[leftarrow]{r}{(x-x_j s z,e_2|e_4)} & X_4
  \end{tikzcd}
\end{array}
\end{equation}

For \textnormal{Spin($8$)}, $i,j$ are two distinct elements of $\{1,2,3\}$ and hence define six distinct crepant resolutions.   
We have two distinct crepant resolutions for \textnormal{Spin($7$)}, and a unique one for \textnormal{G}$_2$.

\subsection{Hyperplane arrangement $\mathrm{I}(\mathfrak{g},  \mathbf{R})$}

Let  $\mathfrak{g}$ be a semi-simple Lie algebra and $ \mathbf{R}$ a representation of $\mathfrak{g}$.
The kernel of each  weight $\varpi$ of $\mathbf{R}$ defines a  hyperplane $\varpi^\perp$ through the origin of the Cartan sub-algebra of $\mathfrak{g}$.

\begin{defn}
The hyperplane arrangement I($\mathfrak{g},\mathbf{R}$) is defined inside the dual fundamental Weyl chamber of $\mathfrak{g}$, i.e. the  dual cone of the fundamental Weyl chamber of $\mathfrak{g}$, and its hyperplanes are the set of kernels of  the weights of $\mathbf{R}$. 
\end{defn}
For each $G$-model, we associate the hyperplane arrangement $\mathrm{I}(\mathfrak{g},  \mathbf{R})$ using the representation $\mathbf{R}$ induced by the weights of vertical rational curves produced by degenerations of the generic fiber over codimension-two points of the base. We then study the incidence structure of the hyperplane arrangement I$(\mathfrak{g}, \mathbf{R})$.

\begin{prop}
The weights of the vertical curves
 over codimension-two points are

\def\arraystretch{1}
\begin{center}
\scalebox{.95}{
\begin{tabular}{|c|c|c|c|c|c|}
\hline 
& G$_2$ & \multicolumn{2}{c}{\textnormal{Spin($7$)}} &\multicolumn{2}{|c|}{\textnormal{Spin($8$)}} \\
\hline 
Locus & $V(s,4f^3+ 2g^2)$  & $V(s,a_{4,2})$  & $V(s,a_{2,1}^2-4a_{4,2})$ &\multicolumn{2}{|c|}{$V(s,d)$}   \\
\hline 
\vrule width 0pt height 3ex 
Weights &$\boxed{2 \  -1}$   & $\boxed{0 \  1\    -2}$  & $\pm\boxed{1\  0  \  -1}$ &\multicolumn{2}{|c|}{$\boxed{0 \  0  \  -1  \  1}$, $ \boxed{-1 \ 0  \  0 \  1}$, $\boxed{-1\  0\  1\  0}$}   \\
\hline 
$\mathbf{R}$ & $\mathbf{7}$  & $\mathbf{7}$  & $\mathbf{8}$  &\multicolumn{2}{|c|}{  $\mathbf{8}_v\oplus\mathbf{8}_s\oplus \mathbf{8}_c$} \\
\hline
$\mathrm{I}(\mathfrak{g},\mathbf{R}) $&$\mathrm{I}(\text{G$_2$},\mathbf{7}) $   & \multicolumn{2}{|c|}{  $\mathrm{I}(B_3,\mathbf{7}\oplus \mathbf{8}$)} &\multicolumn{2}{|c|}{$ \mathrm{I}(D_4,\mathbf{8}_v\oplus\mathbf{8}_s\oplus \mathbf{8}_c)$}   \\
\hline 
 \# of Chambers &$1$   & \multicolumn{2}{|c|}{  $2$} &\multicolumn{2}{|c|}{$6$}   \\
\hline 

\end{tabular}
}
\end{center}
\end{prop}
The chambers are illustrated in Figures \ref{fig:G2Phases}, \ref{fig:Spin7Phases}, and \ref{fig:Spin8Phases} on page \pageref{fig:Spin8Phases} and described in the following theorem.

\begin{thm}
The hyperplane arrangement $\mathrm{I}(\text{G$_2$}, \mathbf{7})$ has a unique chamber. 
The hyperplane arrangement $\mathrm{I}(B_3, \mathbf{7}\oplus \mathbf{8})$ has two chambers whose incidence graph is the Dynkin diagram A$_2$. 
The hyperplane arrangement  $\mathrm{I}(D_4, \mathbf{8}_v\oplus \mathbf{8}_c \oplus \mathbf{8}_s)$ has six chambers whose incidence graph is a hexagon (the affine Dynkin diagram $\widetilde{\text{A}}_6$). 
\end{thm}
\begin{proof}
A hyperplane $\varpi^\perp$ (the kernel of  a weight $\varpi$)  intersects the interior of the dual fundamental Weyl chamber of $\mathfrak{g}$ if and only if when written in the basis of positive simple roots, at least two of its coefficients have  different signs. 
The hyperplane arrangement $\mathrm{I}(\text{G$_2$}, \mathbf{7})$ has a unique chamber as none of its weights have coefficients of different signs in the basis of positive simple roots. 
The hyperplane arrangement $\mathrm{I}(B_3, \mathbf{7}\oplus \mathbf{8})$ has  two chambers separated by the kernel of the unique (up to a sign) weight of the representation 
$\mathbf{7}\oplus \mathbf{8}$; this representation is not in the cone generated by the positive simple roots, namely  $\boxed{1\ 0\   -1}$. 
For Spin($8$), each of the representations  $\mathbf{8}_v$, $\mathbf{8}_c$, and $\mathbf{8}_s$ has a unique weight (up to a sign)  not in the cone generated by the positive simple roots. 
These are the following weights, written respectively for $\mathbf{8}_v$, $\mathbf{8}_c$, and $\mathbf{8}_s$ in the basis of fundamental weights: 
$$
\varpi_4^v=\boxed{0 \  0 \ -1 \ 1 }, \quad \varpi_4^c=\boxed{-1 \  0 \  0 \ 1 }, \quad \varpi_4^s=\boxed{-1 \  0 \ 1 \ 0 }.
$$
Let $\phi=(\phi_1, \phi_2, \phi_3,\phi_4)\in\mathfrak{h}$ be a vector of  the dual fundamental Weyl chamber written in the basis of simple coroots.   The linear forms corresponding to the  weights $\varpi_4^v$, $\varpi_4^c$, and $\varpi_r^s$ are
$$
L_1=(\varpi_4^v,\phi)=-\phi_3+\phi_4, \quad L_2=(\varpi_4^c,\phi)=-\phi_1 + \phi_4 , \quad L_3=(\varpi_4^s,\phi)=-\phi_1+ \phi_3.
$$
The set of chambers of $\mathrm{I}(D_4, \mathbf{8}_v\oplus \mathbf{8}_c \oplus \mathbf{8}_s)$ is in bijection with the set of all possible signs of $(L_1, L_2, L_3)$. Two chambers are adjacent to each other when the signs of $(L_1, L_2, L_3)$ differ by one entry only.  Since $L_1+L_3=L_2$, we can neither get the sign vector $(- + -)$ nor $(+ - +)$. 
It is easy to see that all the six remaining arrangements of signs are possible. 
In total, there are six chambers determined by the signs of $L_1$, $L_2$, and $L_3$ :
\begin{align}\label{Eq:Phases}
\begin{split}
& 1.\, (---)\  \phi_1>\phi_3>\phi_4 , \quad  2. \,(--+) \  \phi_3>\phi_1>\phi_4 , \quad 3.\, (-++) \  \phi_3>\phi_4>\phi_1, \\
& 4.\, (+++) \  \phi_4>\phi_3>\phi_1 , \quad  5. \,(++-) \  \phi_4>\phi_1>\phi_3 , \quad  6.\, (+--) \  \phi_1>\phi_4>\phi_3.
\end{split}
\end{align}
Two chambers are said to be adjacent when they differ  by the sign of  only one  $L_i$. 
It follows that there are six chambers organized as in Figure \ref{fig:Spin8Phases} on page \pageref{fig:Spin8Phases}. 
\end{proof}

\subsection{Matching the crepant resolutions and 
the chambers of
 the  hyperplane arrangement}

For each crepant resolution, we study the fiber structure and compute geometrically the weights of vertical curves. These weights uniquely determine a chamber. 
G$_2$-models only have one crepant resolution.  
For Spin($7$)-models, the crepant resolution $Y^\pm$ corresponds to $\pm \boxed{1,0,-1}$. 
For Spin($8$)-models, the vertical curves give the weights $\pm L_1$, $\pm L_2$, and $\pm L_3$. The explicit matching of the chambers and the crepant resolutions is given in Table \ref{Table:Spin8Match}.

\begin{table}[htb]
\begin{center}
\scalebox{1.3}{
\begin{tabular}{|c|c|c|}
\hline
& Resolutions & Chambers  \\
\hline 
1& 		$Y^{(2,3)}$   & $(- - -) \quad   (\phi_1>\phi_3>\phi_4) $ \\
		\hline 
2& $Y^{(3,2)}$ &    $ (- - +)\quad  (\phi_3>\phi_1>\phi_4) $ \\
\hline
3& $Y^{(3,4)}$&    $(- + +)\quad (\phi_3>\phi_4>\phi_1)$ \\
\hline
4& $Y^{(4,3)}$   & $(+++)\quad  (\phi_4>\phi_3>\phi_1)$\\
\hline
5& $Y^{(4,2)}$   & $(++-)\quad   (\phi_4>\phi_1>\phi_3)$\\
\hline 
 6& $Y^{(2,4)}$   &  $(+--)\quad  (\phi_1>\phi_4>\phi_3)$ \\ 
 \hline 
 \end{tabular}
 }
 \end{center}
\caption{Matching the crepant resolutions of the Spin($8$)-model with the chambers of the hyperplane arrangement I($\text{D}_4,\mathbf{8}_v\oplus\mathbf{8}_s\oplus\mathbf{8}_c$). \label{Table:Spin8Match} }
\end{table}

\begin{figure}
\begin{center}
 \begin{tikzcd}[column sep=.1 , row sep=.1cm, scale=.5]
&    \raisebox{-1.5cm}{\scalebox{.7}{\includegraphics[scale=.8]{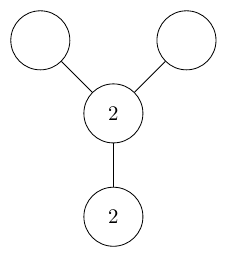}}}
    & \raisebox{-1.5cm}{\scalebox{.7}{\includegraphics[scale=.8]{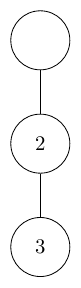}}}
   & \raisebox{-1.5cm}{\scalebox{.6}{\includegraphics[scale=.8]{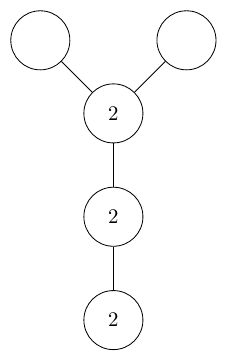}}}
&    \raisebox{-1.5cm}{\scalebox{.6}{\includegraphics[scale=.8]{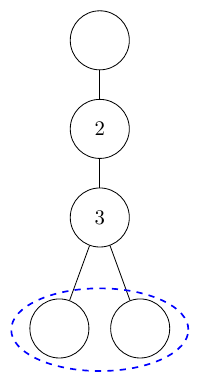}}} 
& \raisebox{-1.5cm}{\scalebox{.6}{\includegraphics[scale=.8]{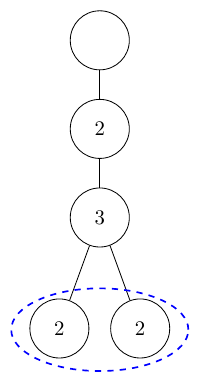}}}  
& \quad \raisebox{-1.5cm}{\scalebox{.7}{\includegraphics[scale=.8]{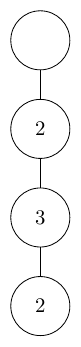}}} 
&\quad  \raisebox{-1.5cm}{\scalebox{.7}{\includegraphics[scale=.8]{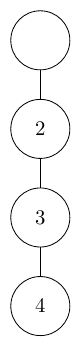}}} 
& \quad  \raisebox{-1.5cm}{\scalebox{.8}{\includegraphics[scale=.6]{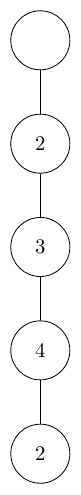}}} \\
\text{G}_2 & \checkmark & \checkmark & & & & & & \\
\text{Spin}(7) & \checkmark &  &\checkmark  &\checkmark  &\checkmark  &\checkmark  & \checkmark & \\
\text{Spin}(8) &  & &  &  && &\checkmark  & \checkmark  \\
 \end{tikzcd} 
 \end{center}
 \caption{Non-Kodaira fibers appearing in the fiber structures of  G$_2$, Spin($7$), and Spin($8$)-models \label{Fig:NonKodaira}}
 \end{figure}

\begin{figure}
{

\centering
\begin{tikzpicture}[scale=.5]

		\draw (0,0)--(75:5);	
			\draw (0,0)--(75+30:5);

			\end{tikzpicture}
\caption{Chambers of the hyperplane arrangement I$(\text{G$_2$},\mathbf{7})$ or equivalently, the Coulomb phases of a G$_2$ gauge theory with matter in the representation $\mathbf{7}$.
There is a unique chamber since the non-zero weights of $\mathbf{7}$ are the short roots of G$_2$. \label{fig:G2Phases}
				}		
				
				\vspace{1cm}

				\centering
\begin{tikzpicture}[scale=.3]
				
			\draw (90:7)--(210:7)--(330:7)--(90:7);
			\draw (90:7)--(90:-4.5);
									\node at (-87:5.7) {\scalebox{.8}{$\varpi_1-\varpi_3$}};	
					\node at (200:-2) {\scalebox{2}{$+$}};
					\node at (160:2) {\scalebox{2}{$-$}};													
			\end{tikzpicture}
\caption{Chambers of the hyperplane arrangement I$(B_3,\mathbf{8})$ or equivalently, the Coulomb phases of a Spin($7$) gauge theory with matter in the representation $\mathbf{8}$.
The only weight defining an interior wall is the weight $\boxed{1\ 0\  -1}$ of the representation $\mathbf{8}$. \label{fig:Spin7Phases}
				}

	 \vspace{1cm}

}

\begin{center}
\begin{tabular}{cc}
\begin{tikzpicture}[scale=.6]
				
				\node[draw,circle,thick,scale=1] (1) at (0:3cm){$+--$};
								\node[draw,circle,thick,scale=1] (2) at (60:3cm){$++-$};
																\node[draw,circle,thick,scale=1] (3) at (120:3cm){$+++$};
				\node[draw,circle,thick,scale=1] (4) at (180:3cm){$-++$};
				\node[draw,circle,thick,scale=1] (6) at (-60:3cm){$---$};
				\node[draw,circle,thick,scale=1] (5) at (-120:3cm){$--+$};
				\draw[thick] (1)--(2)--(3)--(4)--(5)--(6)--(1);
							\draw[dashed] (90:3.6)--(90:-3.1);
			\draw[dashed] (210:3.2)--(210:-3.2);
			\draw[dashed] (330:3.2)--(330:-3.2);	
			\node at (330:-3.8) {\scalebox{1.4}{$L_1$}};		
						\node at (210:-3.8) {\scalebox{1.4}{$L_2$}};	
									\node at (90:-3.4) {\scalebox{1.4}{$L_3$}};

			\end{tikzpicture}
			& 
			\begin{tikzpicture}[scale=.4]
				
			\draw (90:7)--(210:7)--(330:7)--(90:7);
			\draw (90:7)--(90:-4.1);
			\draw (210:7)--(210:-4.1);
			\draw (330:7)--(330:-4.1);	
			\node at (330:-4.8) {\scalebox{1.4}{$L_1$}};		
						\node at (210:-4.8) {\scalebox{1.4}{$L_2$}};	
									\node at (90:-4.8) {\scalebox{1.4}{$L_3$}};	
						\node at (240:-3) {\scalebox{1}{$++-$}};	
						\node at (240+60:-3) {\scalebox{1}{$+++$}};
						\node at (240+2*60:-3) {\scalebox{1}{$-++$}};	
						\node at (240+3*60:-3) {\scalebox{1}{$--+$}};	
				         \node at (240+4*60:-3) {\scalebox{1}{$---$}};	
				          \node at (240+5*60:-3) {\scalebox{1}{$+--$}};

			\end{tikzpicture}

			\end{tabular}
			\end{center}
			\caption{Chambers of the hyperplane arrangement I$(D_4,\mathbf{8}_v\oplus\mathbf{8}_s\oplus\mathbf{8}_c)$ or equivalently, the Coulomb phases of a Spin($8$) gauge theory with matter in the representation $\mathbf{8}_v\oplus\mathbf{8}_s\oplus\mathbf{8}_c$. 
			The signs in the figure are those of linear forms induced by the weights $L_1=\boxed{0,0,-1,1}$, $L_2=\boxed{-1,0,0,1}$, and $L_3=\boxed{-1,0,1,0}$.
			\label{fig:Spin8Phases}
				}

\end{figure}

\clearpage

\subsection{Fiber degenerations}

 The G$_2^{S_3}$-model is the generic case of a fiber of type I$_0^*$ and has the  simple fiber structure presented in Figure \ref{Fig:G2S3}. In codimension two, the generic fiber with dual graph $\widetilde{\text{G}}^t_2$ degenerates  to an incomplete affine Dynkin diagram 
of type $\widetilde{\text{D}}_5$, whose dual graph is a Dynkin diagram of type D$_4$. The fiber degenerates further  in codimension three to a fiber of type $1-2-3$, which can be understood as an incomplete fiber of type  $\widetilde{\text{E}}_6$. 
 The G$_2^{\mathbb{Z}/3\mathbb{Z}}$-model can be realized in various ways. 
 The model is defined in equation \eqref{Eq:G2Z31}, and has a fiber structure with an enhancement in codimension two 
 from   $\widetilde{\text{G}}_2^t$ 
 to an incomplete $\widetilde{\text{E}}_6$. This can be understood as a specialization of G$_2^{S_3}$ in which the codimension two fiber $\widetilde{\text{D}}_5$ does not appear and the fiber degenerates  directly to  an incomplete 
  $\widetilde{\text{E}}_6$.
The fiber structure of the G$_2^{\mathbb{Z}/3\mathbb{Z}}$-model defined in equation \eqref{Eq:G2Z32} is presented in Figure \ref{Fig:G2Z3term} and has two distinct fibers in codimension-two, in contrast to the other G$_2$-model, which has only  one type of specialization in codimension two.

The generic fiber structure of   a Spin($7$)-model is presented in Figure \ref{Fig:Spin7.1}.  The fiber structure  
depends on the  valuations ($v_S(a_2)\geq 1$ and $v_S(a_6)\geq 4$) and the choice of a crepant resolution. The valuations can be  organized into  four different cases, and we have two possible choice of crepant resolutions. Thus,  there are  eight distinct types of fiber structures. Two resolutions related by a flop 
have distinct fibers in codimension two, three, or four,  depending on the valuations. We organize this information by grouping the flops together. 
In one of the two possible resolutions, both the $C_2$ and $C_3$ components of the generic $\widetilde{\text{B}}^t_3$  fiber degenerate, while  in the flop, only $C_3$ degenerates. 

The first case is in Figure \ref{Fig:Spin7.1}, and corresponds to the lowest possible values for the valuations of $a_2$ and $a_6$, namely $v_S(a_2)=1$ and $v_S(a_6)=4$. All the other cases are specialization of this case,  obtained by skipping  some of the intermediate steps in the degenerations as forced by the valuations. The fiber structure  in 
Figure \ref{Fig:Spin7.2} is for the case ($v_S(a_2)=1, v_S(a_6)\geq 5$), Figure \ref{Fig:Spin7.3} for the case ($v_S(a_2)=2, v_S(a_6)= 4$), and finally Figure \ref{Fig:Spin7.4} for the case ($v_S(a_2)\geq 2, v_S(a_6)\geq 5$). 
Each of these Weierstrass models has two possible resolutions related by a flop, and each possible resolution has a different fiber in codimension two or three.

For the Spin($8$)-models, the fiber structure is the split case of a fiber of type I$_0^*$ and is presented in Figures \ref{Fig:Spin8AlphaZero} and \ref{Fig.Spin8AlphaPos}. The most generic case for the Spin($8$)-models is when $\alpha=0$. 
In such a model, the fiber structure degenerates into  $\widetilde{\text{D}}_5$ in codimension two, $\widetilde{\text{E}}_6$ or an incomplete $\widetilde{\text{D}}_6$ in codimension three, and an incomplete $\widetilde{\text{E}}_7$ in codimension four.
The fiber structure for  $\alpha> 0$ is presented in Figure \ref{Fig.Spin8AlphaPos}, in which the fiber structure skips $\widetilde{\text{D}}_5$ and degenerates directly into
  either an incomplete $\widetilde{\text{D}}_6$ in codimension two or an incomplete $\widetilde{\text{E}}_7$ in codimension three. In contrast to Spin($7$)-models, the fiber type does not depend on the choice of the crepant resolution.
Even though different crepant resolutions are differentiated by the way the curves split, they give the same fiber types.

\begin{figure}[h]
\centering

\begin{tikzcd}[column sep= normal , row sep=.4cm, scale=.5]
  \raisebox{-1.5cm}{\scalebox{1}{\includegraphics[scale=.6]{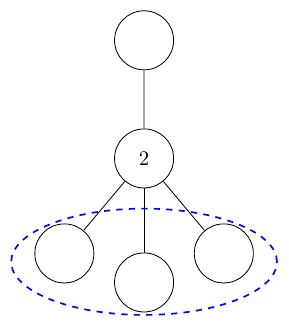}}}      \arrow[rightarrow]{r}{4 f^3+27 g^2=0}  
  &   \raisebox{-1.5cm}{\scalebox{1}{\includegraphics[scale=.7]{VT22-12-12}}} \arrow[rightarrow]{r}{f=g=0}  & 
   \raisebox{-1.5cm}{\scalebox{1}{\includegraphics[scale=.7]{VT123}}}
   \\
  \widetilde{\text{G}}_2^t & \text{Incomplete $\widetilde{\text{D}}_5$}  &  \text{Incomplete $\widetilde{\text{E}}_6$}   \\
\end{tikzcd}  
\caption{Geometric fiber degeneration of $\widetilde{G}_2^{S_3}$-model with valuation $(2,3,6)$. \label{Fig:G2S3}
}

\begin{tikzcd}[column sep= normal , row sep=.4cm, scale=.5]
  \raisebox{-1.5cm}{\scalebox{1}{\includegraphics[scale=.6]{G2dash}}} \quad 
   \arrow[rightarrow]{r}{\displaystyle g=0}  
  &   \quad \quad\raisebox{-1.5cm}{\scalebox{1}{\includegraphics[scale=.6]{VT123}}}  \\
  \widetilde{G}_2^t & \text{Incomplete $\widetilde{\text{E}}_6$}   
\end{tikzcd}  
\caption{Geometric fiber degeneration of $\widetilde{G}_2^{\mathbb{Z}/3\mathbb{Z}}$ with valuation $(\geq 3,3,6)$. \label{Fig:G2Z3}
 }
  \vspace*{\floatsep}
\centering
\begin{tikzcd}[column sep= normal , row sep=.4cm, scale=.5]
&  \raisebox{-1.5cm}{\scalebox{1}{\includegraphics[scale=.6]{VT123}     }}    \arrow[leftarrow]{d}{V(a,r)} &  \\
  \raisebox{-1.5cm}{\scalebox{1}{\includegraphics[scale=.6]{G2dash}  }}      \arrow[rightarrow]{r}{V(a-r)}  \arrow[rightarrow]{ru}{ V(ar)}  
  &   \raisebox{-1.5cm}{\scalebox{1}{\includegraphics[scale=.6]{VT22-12-12}}}\\
  \widetilde{\text{G}}^t_2 &
    \text{Incomplete $\widetilde{\text{D}}_5$}   
\end{tikzcd}  
\caption{Geometric fiber degeneration for a  $\widetilde{G}_2^{\mathbb{Z}/\mathbb{Z}_3}$-model with valuation $(2,3,6)$ and special configuration for $f$ and $g$ such that $\Delta'=4f^3+27 g^2$ is a perfect square in the residue field $\kappa(\eta)$. 
 We assume that the base $B$ is a surface to avoid  $\mathbb{Q}$-factorial terminal singularities. \label{Fig:G2Z3term}
}
\end{figure}

\clearpage

\begin{figure}[ht]
\centering

\begin{tikzcd}[column sep= normal , row sep=.2cm, scale=.4]
& \raisebox{-1.5cm}{\scalebox{1}{\includegraphics[scale=.8]{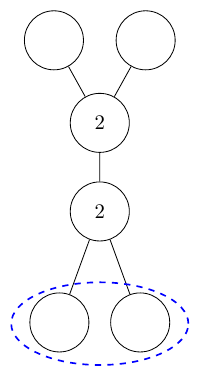}}} \arrow[rightarrow]{r}{\displaystyle  a_{6,4}=0}  \arrow[rightarrow]{rd}{\displaystyle  a_{2,1}=0}&\raisebox{-1.5cm}{\scalebox{1}{\includegraphics[scale=.8]{VT11222}}}  \arrow[rightarrow]{rd}{\displaystyle a_{2,1}=0}
&\\
\raisebox{-2.6cm}{\scalebox{.8}{\includegraphics[scale=.8]{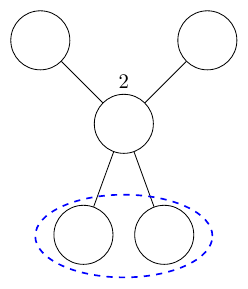}}} \arrow[rightarrow]{ru}{\displaystyle a_{4,2}=0}  \arrow[rightarrow]{rd}{\displaystyle  a_{2,1}^2-4a_{4,2}=0}& 
&    \raisebox{-2.6cm}{\scalebox{.8}{\includegraphics[scale=.8]{Vdash123-11}}} or \raisebox{-2.6cm}{\scalebox{.8}{\includegraphics[scale=.8]{Vdash123-22}}}  \arrow[rightarrow]{r}{\displaystyle  a_{6,4}=0} & \raisebox{-2.6cm}{\scalebox{.8}{
\includegraphics[scale=.8]{V1232}}} 
or 
\raisebox{-2.6cm}{\scalebox{.8}{
\includegraphics[scale=.8]{V1234}}} 
 \\
&  \raisebox{-1.5cm}{\scalebox{.8}{\includegraphics[scale=.8]{VT22-12-12}}}  \arrow[rightarrow]{ru}{\displaystyle  a_{4,2}=0}& 
 \end{tikzcd} 
 \label{Fig:Spin7.1}
\caption{Geometric fiber degeneration of a Spin($7$)-model with $v_{S}(a_2)=1$ and $v_{S}(a_6)=4$. When $v_{S}(a_2)>1$, the degeneration graph contracts to the middle row. 
When there are multiple fibers, the one on the left corresponds to the monomial resolution and the one on the right to its flop.  
 }
\end{figure}

\clearpage

\begin{figure}[ht]
\centering
\begin{tikzcd}[column sep= 3.5cm , row sep=-.2cm, scale=.7]
& 
\raisebox{-3cm}{\scalebox{1}{\includegraphics[scale=.5]{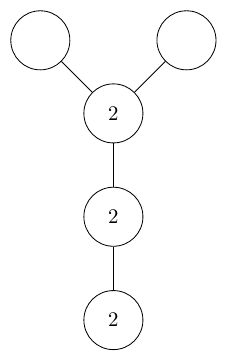}}}  \arrow[rightarrow]{rd}{\displaystyle a_{2,1}=0}
  \arrow[rightarrow]{rd}{\displaystyle  a_{2,1}=0}&
&\\
\raisebox{-2cm}{\scalebox{1}{\includegraphics[scale=.5]{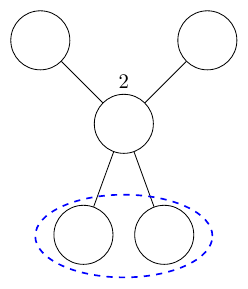}}} \arrow[rightarrow]{ru}{\displaystyle a_{4,2}=0}  \arrow[rightarrow]{rd}{\displaystyle  a_{2,1}^2-4a_{4,2}=0}& 
&  \quad \raisebox{-2cm}{\scalebox{.7}{
\includegraphics[scale=.7]{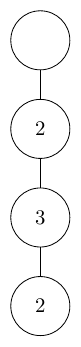}}}  or 
\raisebox{-2cm}{\scalebox{.7}{
\includegraphics[scale=.7]{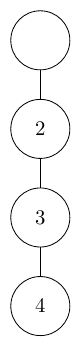}}} &
 \\
&  \raisebox{-.5cm}{\scalebox{1}{\includegraphics[scale=.5]{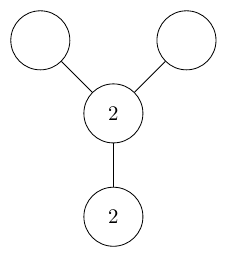}}}  \arrow[rightarrow]{ru}{\displaystyle  a_{4,2}=0}& 
 \end{tikzcd} 
\caption{Geometric fiber degeneration of a Spin($7$)-model with $v_{S}(a_2)=1$ and $v_{S}(a_6)\geq 5$. 
\label{Fig:Spin7.2}}
 \vspace*{\floatsep}
\begin{center}
  \begin{tikzcd}[column sep=normal, scale=.6] 
   \raisebox{-2.6cm}{\scalebox{.8}{
\includegraphics[scale=.8]{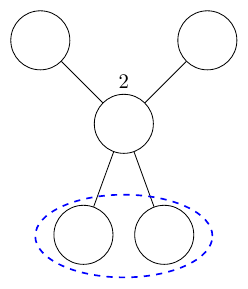}}}
 \arrow[rightarrow]{r}{\displaystyle{a_{4,2}=0}}
&
    \raisebox{-2.6cm}{\scalebox{.8}{
\includegraphics[scale=.8]{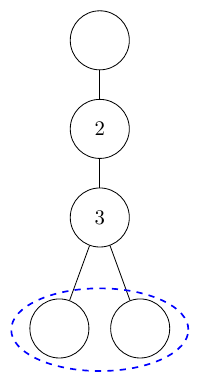}} \raisebox{2.6cm}{or} 
{\scalebox{.8}{
\includegraphics[scale=.8]{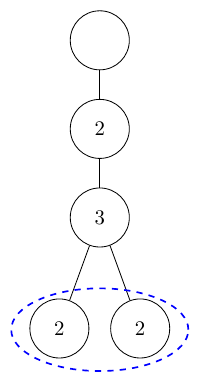}}}}
 \arrow[rightarrow]{r}{\displaystyle{a_{6,4}=0}}
&
\quad  \raisebox{-2.6cm}{\scalebox{.8}{
\includegraphics[scale=.8]{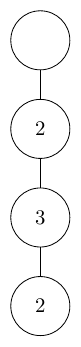}}}
or 
\raisebox{-2.6cm}{\scalebox{.8}{
\includegraphics[scale=.8]{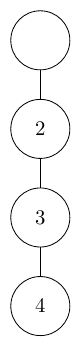}}} 

\end{tikzcd}
\end{center}
\caption{
Geometric fiber degeneration of a Spin($7$)-model with $v_{S}(a_2)\geq 2$ and $v_{S}(a_6)= 4$. 
\label{Fig:Spin7.3}
}
 \vspace*{\floatsep}
\centering
\begin{tikzcd}[column sep= normal , row sep=.2cm, scale=.4]
\raisebox{-2.6cm}{\scalebox{1}{\includegraphics[scale=.6]{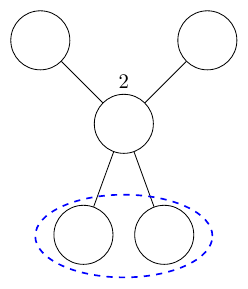}}} \arrow[rightarrow]{r}{\displaystyle a_{4,2}=0}  
&  \quad \raisebox{-2.6cm}{\scalebox{.6}{
\includegraphics[scale=.9]{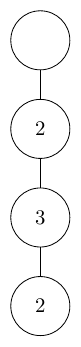}}}  or 
\raisebox{-2.6cm}{\scalebox{.6}{
\includegraphics[scale=.9]{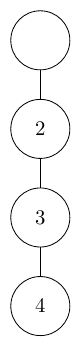}}} 
 \end{tikzcd} 
\caption{
Geometric fiber degeneration of a Spin($7$)-model with $v_{S}(a_2)\geq 2$ and $v_{S}(a_6)\geq 5$.
\label{Fig:Spin7.4}
  }
\end{figure}

\clearpage

\begin{figure}[htb]
\begin{center}
\scalebox{.9}{
\begin{tikzpicture}[scale=1]
	\node(0) at (-2.5,0){\scalebox{1}{\includegraphics[scale=.8]{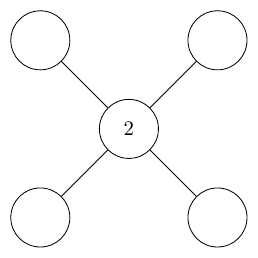}}};
	\node(1) at (2,0){\scalebox{1}{\includegraphics[scale=.8]{B4dash}}};
\node(2p) at (7,5){\scalebox{1}{\includegraphics[scale=.9]{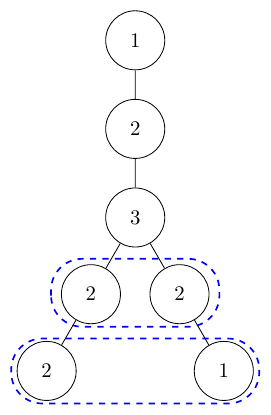}}};
	\node(2m) at (7,-5){\scalebox{1}{\includegraphics[scale=.8]{VT11222}}};
	\node(3) at (11.5,0){\scalebox{1}{\includegraphics[scale=.8]{VT12342}}};
		\draw[big arrow] (0) -- node[above=.1]{$V(d)$} (1);
		\draw[big arrow] (3,0) -- node[ below right] {$V(x_1-x_2,x_2-x_3)$} (2p);
		\draw[big arrow] (3,-1) -- node[above right]{$V(Q)$} (2m);	
		\draw[big arrow] (2m) -- node[above left]{$V(x_1-x_2,x_2-x_3)$} (3);
		\draw[big arrow] (2p) -- node[below left ]{$V(Q)$} (3);      
	\end{tikzpicture}}
\caption{Fiber degeneration of a Spin($8$)-model with $\alpha=0$. \label{Fig:Spin8AlphaZero}}
\begin{tikzcd}[column sep= large , scale=.6]
  \raisebox{-1.5cm}{\scalebox{1}{\includegraphics[scale=.7]{Istar0}}}   \arrow[rightarrow]{r}{V(d)}  &
  \raisebox{-1.5cm}{\scalebox{1}{\includegraphics[scale=.7]{VT11222}}} \arrow[rightarrow]{r}{V(x_1-x_2=x_2-x_3}  &\quad
    \raisebox{-1.5cm}{\scalebox{1}{\includegraphics[scale=.6]{VT12342}}} \\
    \widetilde{\text{D}}_4 & \text{Degenerate  $\widetilde{\text{D}}_5$ } & \text{Degenerated $\widetilde{\text{E}}_7$}
\end{tikzcd} 
\caption{Fiber degeneration of  a Spin($8$)-model with $\alpha>0$.  \label{Fig.Spin8AlphaPos}}
\end{center}
\end{figure}
\clearpage

\subsection{Triple intersection numbers}\label{Sec:Triple}

Two crepant resolutions of the same Weierstrass model have the same Euler characteristic \cite{Batyrev.Betti}. 
 In this subsection, we compute a topological invariant that is dependent on the choice of the crepant resolution. 
Let $D_0, D_1, \cdots, D_r$ be the  fibral divisors of  an elliptic fibration  over a base $B$ of dimension $d$. 
By definition, they are the irreducible components of 
$$
\varphi^* S= m_0 D_0 +m_1 D_1 + \cdots + m_r D_r,
$$
where $m_i$ is the multiplicity of $D_i$ and $D_0$ is the divisor touching the section of the elliptic fibration. 
It is useful to introduce a polynomial ring 
$A_*(Y)[\phi_0, \cdots, \phi_r]$ over the Chow ring $A_*(Y)$ of $Y$.

We define the polynomial $\mathscr{F}$  in $A_*(B)[\varphi_0, \cdots, \varphi_r]$ via a  pushforward as 
$$
\mathscr{F}:=\varphi_* (D_0 \phi_0 + D_1 \phi_1 +\cdots D_r \phi_r)^3.
$$
Hence, if $M$ is an element of $A_{d-3} (B)$, we have 
$$
\int_Y (D_0 \phi_0 + D_1 \phi_1 +\cdots D_r \phi_r)^3 \cdot \varphi^* M =\int_B \varphi_* \Big[(D_0 \phi_0 + D_1 \phi_1 +\cdots D_r \phi_r)^3 \Big]\cdot M=\int_B  \mathscr{F} \cdot M.
$$
The polynomial $ \mathscr{F}$ is called the triple intersection polynomial of the elliptic fibration. When the base is a surface, its coefficients are numbers.

A Weierstrass model of a G$_2$-model has a unique crepant resolution and therefore a unique possible triple intersection polynomial. 
There are two distinct crepant resolutions   $Y^\pm$ for the Weierstrass model corresponding to  a Spin($7$)-model. These two crepant resolutions also have different triple intersection polynomials $\mathscr{F}^\pm_{\text{Spin($7$)}}$. The  Weierstrass model of a Spin($8$)-model has  six distinct crepant resolutions $Y^{(a,b)}$ with $\{a, b\neq a\}\subset\{2,3,4\}$,  and we will compute all  six different triple intersetion polynomials $\mathscr{F}^{(a,b)}_{\text{Spin($8$)}}$.
\begin{thm}\label{Them:triple}
The triple intersection numbers of a \textup{G$_2$}, \textup{Spin($7$)}, or \textup{Spin($8$)}-model defined by the sequence of blowups listed in section \ref{Sec:Crepant} are 
\begin{enumerate}
\item for G$_2$-models,
\begin{align}
& \begin{aligned}
\mathscr{F}_{\textup{G}_2} &=-4 S (S-L)\phi _0^3 +3 S (S-2 L) \phi _0^2 \phi _1+3 L S \phi _0 \phi _1^2 \\
&-4 S (S-L) \phi _1^3 +12 S  (S-3 L)\phi _2^3+9 S  (2 S-3 L) \phi _1^2 \phi _2-27 S (S-2 L)  \phi _1 \phi _2^2 ,
\end{aligned}
\end{align}
\item for the two chamers of \textup{Spin($7$)}-models identified as
\begin{align}
\mathscr{F}^+_{\textup{Spin($7$)}}=\mathscr{F}^{(1,3)}_{\textup{Spin($7$)}}\ \text{and} \quad \mathscr{F}^-_{\textup{Spin($7$)}}=\mathscr{F}^{(3,1)}_{\textup{Spin($7$)}},
\end{align}
\begin{align}
& \begin{aligned}
\mathscr{F}^{(m,n)}_{\textup{Spin($7$)}} &=-4 S(S-L) \phi _0^3 +3 S (S-2 L) \phi _0^2\phi _2+3 L S \phi _0\phi _2^2-4 S (S-L) \phi _2^3  \\
& -3 S (S-2L) \phi _2 (\phi _m^2+4 \phi _1 \phi _3 +4  \phi _n^2 ) + 3 S  (2 S-3 L) (\phi_m+2\phi _n)\phi _2^2\\
&-4 L S \phi _m^3-8 L S \phi _n^3+12 S (S-2 L)  \phi _m\phi _n^2 -4S (S-2 L) (\phi _m-\phi _1)^3,
\end{aligned}
\end{align}
\item for the six chambers $Y^{(a,b)}$ of \textup{Spin($8$)}-models 
\begin{align}
\begin{aligned}
\mathscr{F}^{(a,b)}_{\text{Spin($8$)}} &=-4 S(S-L) \phi _0^3  +3 S  (S-2 L) \phi _0^2 \phi _1+3 L S \phi _0 \phi _1^2-4 S (S-L)\phi _1^3 \\
&+3 S (2 S-3 L) \phi _1^2 (\phi _2+\phi _3+\phi _4) -3 S (S-2L) \phi _1 (\phi _2+ \phi _3+\phi _4)^2 \\
& -4 L S \phi _a^3-2 S^2 \phi _b^3-4 S (S-L) \phi _c^3+6 S  (S-2 L) (\phi _b \phi _c^2+ \phi _a \phi _b^2 +\phi _a \phi _c^2) ,
\end{aligned}
\end{align}
where $(a,b,c)$ is a permutation of $(2,3,4)$.
\end{enumerate}
\end{thm}

\begin{lem}\label{Lem:Triple}
If  a G-model is a Calabi--Yau threefold, $c_1=L=-K$. Furthermore, denoting by $g$ the genus of $S$, the triple intersection numbers are
\begin{enumerate}
\item for G$_2$-models,
\begin{align}
& \begin{aligned}
\   \   \mathscr{F}_{\text{G}_2} &= -8 (g-1) \phi _0^3 +3 \phi _1 \phi _0^2 \left(4 g-4-S^2\right)-3 \phi _1^2 \phi _0 \left(2 g-2-S^2\right) \\
& -8 (g-1) \phi_1^3+24 \left(3 g-3-S^2\right)\phi _2^3  -27 \left(4 g-4-S^2\right)\phi_1\phi_2^2 +9\left(6 g-6-S^2\right)\phi _1^2 \phi _2 
\end{aligned}
\end{align}
\item for the two chamers of \textup{Spin($7$)}-models identified as
\begin{align}
\mathscr{F}^+_{\textup{Spin($7$)}}=\mathscr{F}^{(1,3)}_{\textup{Spin($7$)}}\ \text{and} \quad \mathscr{F}^-_{\textup{Spin($7$)}}=\mathscr{F}^{(3,1)}_{\textup{Spin($7$)}},
\end{align}
\begin{align}
&\begin{aligned}
\mathscr{F}^{(m,n)}_{\text{Spin($7$)}} &=-8(g-1)\phi _0^3 +3(4 g-4-S^2)\phi _0^2\phi _2  -3(2 g-2-S^2)\phi _0\phi _2^2 -8(g-1)\phi _2^3 \\
& -3(4 g-4-S^2) \phi _2 (\phi _m^2 + 4 \phi _1 \phi _3 +4  \phi _n^2)  + 3(6g-6-S^2)(\phi_m+2\phi _n)\phi _2^2\\
&+4 (2 g-S^2-2)\phi _m^3  +8 (2 g-S^2-2)\phi _n^3 +12(4g-4-S^2)\phi _m\phi _n^2 -4S (S-2 L) (\phi _m-\phi _1)^3
\end{aligned}
\end{align}
\item for the six chambers $Y^{(a,b)}$ of \textup{Spin($8$)}-models 
\begin{align}
& \begin{aligned}
\mathscr{F}^{(a,b)}_{\text{Spin($8$)}} &=
-8 (g-1) \phi _0^3+3 \left(4 g-4-S^2\right)\phi _0^2 \phi _1 -3 \left(2 g-2+S^2\right)\phi _0 \phi _1^2 -8 (g-1) \phi _1^3 \\
&+3 \left(6 g-6-S^2\right) \phi _1^2 \left(\phi _2+\phi _3+\phi _4\right) -3 \left(4 g-4-S^2\right) \phi _1 \left(\phi _2+\phi _3+\phi _4\right)^2 \\
& +4 \left(2 g-2-S^2\right) \phi _a^3-2 S^2 \phi _b^3-8(g-1) \phi _c^3+6(4 g-4-S^2)(\phi _b \phi _c^2+ \phi _a \phi _b^2 +\phi _a \phi _c^2),
\end{aligned}
\end{align}
where $(a,b,c)$ is a permutation of $(2,3,4)$.
\end{enumerate}
\end{lem}

\begin{rem}[Comparison with Diaconescu--Entin]
The formula we get for the triple intersection polynomial of a G$_2$-model should reproduce the result of  Diaconescu and Entin when the base curve is a smooth rational curve (which implies that $g=0$) \cite{Diaconescu:1998cn}:
\begin{equation}
6\mathcal{F}'_{\text{G$_2$}}= 8\phi_1^3 + 8(1-\tilde{g}) \phi^3_2+ 9(\tilde{g}+2) \phi_1\phi_2^2 -3 (\tilde{g}+8) \phi_1^2 \phi_2,
\end{equation} 
where $\tilde{g}$ is the genus of the curve $\gamma$ defined by the intersection of the fibral divisors D$_1$ and D$_2$. 
Using adjunction and the Calabi--Yau condition, we see that the Euler characteristic of $D_1\cap D_2$ is 
\begin{equation}
2-2g'= -\int_Y (D_1+D_2) D_1 D_2 = -\int_B \varphi_* (D_1^2 D_2 +D_1 D_2^2)=-18 + 18 g - 6 S^2,
\end{equation} 
which gives 
\begin{equation}
g'=10-9g+3 S^2.
\end{equation} 
  We get a perfect matching for  
\begin{equation}
\phi_0=g=0, \quad \tilde{g}= 10+3S^2. 
\end{equation}
\end{rem}

 Lemma \ref{Lem:Triple}  is a direct specialization of 
Theorem  \ref{Them:triple}. Theorem \ref{Them:triple} is  proven using the following pushforward theorems (see \cite{Euler} for examples of such computations). 
\begin{thm}[Esole--Jefferson--Kang,  see  {\cite{Euler}}] \label{Thm:Push}
    Let the nonsingular variety $Z\subset X$ be a complete intersection of $d$ nonsingular hypersurfaces $Z_1$, \ldots, $Z_d$ meeting transversally in $X$. Let $E$ be the class of the exceptional divisor of the blowup $f:\widetilde{X}\longrightarrow X$ centered 
at $Z$.
 Let $\widetilde{Q}(t)=\sum_a f^* Q_a t^a$ be a formal power series with $Q_a\in A_*(X)$.
 We define the associated formal power series  ${Q}(t)=\sum_a Q_a t^a$ whose coefficients pullback to the coefficients of $\widetilde{Q}(t)$. 
 Then the pushforward $f_*\widetilde{Q}(E)$ is:
 $$
  f_*  \widetilde{Q}(E) =  \sum_{\ell=1}^d {Q}(Z_\ell) M_\ell, \quad \text{where} \quad  M_\ell=\prod_{\substack{m=1\\
 m\neq \ell}}^d  \frac{Z_m}{ Z_m-Z_\ell }.
 $$ 
\end{thm}

The second pushforward theorem deals with  the projection from the ambient projective bundle to the base $B$ over which the Weierstrass model is defined.  
Let $\mathscr{V}$ be a vector bundle of rank $r$ over a nonsingular variety $B$. The Chow ring of a projective bundle $\pi: \mathbb{P}(\mathscr{V})\longrightarrow B$ is  isomorphic to the module $ A_*(B)[\zeta]$ modded out by the relation \cite[Remark 3.2.4, p. 55]{Fulton}
 $$\zeta^r+ c_1(\pi^* \mathscr{V}) \zeta^{r-1}+\cdots+ c_i(\pi^* \mathscr{V}) \zeta^{r-i}+\cdots + c_r (\pi^* \mathscr{V})=0, \quad \zeta=c_1 \Big( \mathscr{O}_{\mathbb{P}(\mathscr{V})}(1)\Big).$$

\begin{thm}[{See  \cite{Euler} and \cite{AE1,AE2,Fullwood:SVW,EKY}}]\label{Thm:PushH}
Let $\mathscr{L}$ be a line bundle over a variety $B$ and $\pi: X_0=\mathbb{P}[\mathscr{O}_B\oplus\mathscr{L}^{\otimes 2} \oplus \mathscr{L}^{\otimes 3}]\longrightarrow B$ a projective bundle over $B$. 
 Let $\widetilde{Q}(t)=\sum_a \pi^* Q_a t^a$ be a formal power series in  $t$ such that $Q_a\in A_*(B)$. Define the auxiliary power series $Q(t)=\sum_a Q_a t^a$. 
Then 
$$
\pi_* \widetilde{Q}(H)=-2\left. \frac{{Q}(H)}{H^2}\right|_{H=-2L}+3\left. \frac{{Q}(H)}{H^2}\right|_{H=-3L}  +\frac{Q(0)}{6 L^2},
$$
 where  $L=c_1(\mathscr{L})$ and $H=c_1(\mathscr{O}_{X_0}(1))$ is the first Chern class of the dual of the tautological line bundle of  $ \pi:X_0=\mathbb{P}(\mathscr{O}_B \oplus\mathscr{L}^{\otimes 2} \oplus\mathscr{L}^{\otimes 3})\rightarrow B$.
\end{thm}

 Lemma \ref{Lem:Triple} is a specialization to the case of Calabi--Yau threefolds and  will play an important role in section 
\ref{Sec:IMS}. For each model and each irreducible representation $\mathbf{R}_i$ induced by its fiber degeneration, we compute the  number of hypermultiplets transforming in the representation $\mathbf{R}_i$. 
 The results are given by Proposition \ref{prop:ChargeNumber} 
 by comparing the triple intersection numbers given in Lemma \ref{Lem:Triple} with the one-loop prepotential computed in section \ref{Sec:IMS}:
 \begin{align}
\textnormal{G}_2: &\quad n_{\mathbf{7}}=-10(g-1)+3S^2, & n_{\mathbf{14}}=g, \nonumber \\
\textnormal{Spin($7$)}: &\quad n_{\mathbf{7}}= S^2-3 (g-1), \quad n_{\mathbf{8}} =2 S^2-8 (g-1), \quad & n_{\mathbf{21}}=g,\nonumber \\
\textnormal{Spin($8$)}: &\quad n_{\mathbf{8_v}}=n_{\mathbf{8_s}}=n_{\mathbf{8_c}}=S^2-4(g-1), \quad & n_{\mathbf{28}}=g. \nonumber
\end{align}

\subsection{Euler characteristics and Hodge numbers}\label{sec:Euler}

The results of this subsection are proven in \cite{Euler}  and  used in section \ref{Sec:Phys} to check the cancellation of anomalies in the  six-dimensional  supersymmetric  theory.

\begin{thm}[Euler characteristics]
Smooth elliptic fibrations $Y\rightarrow B$ defined as   crepant resolutions  of the Weierstrass models given in section \ref{Sec:Crepant}  have the following Euler characteristics: 
$$
\begin{aligned}
\text{G$_2$}   :& \quad \chi(Y)=12\int_B \frac{ L +2 S L-S^2}{(1 + S) (1  + 6 L- 3 S)}\  c(TB) \\
\text{Spin($7$)} :&\quad \chi(Y)=4 \int_B \frac{
3 L  + (12 L^2 + L S - 5 S^2)  + 5 (3 L - 2 S) (2 L - S) S 
}
{(1 + S) (-1 - 6 L+ 4 S ) (-1  - 4 L+ 2 S)}\  c(TB) \\
\text{Spin($8$)} :&\quad \chi(Y)=12\ \int_B \frac{L+3S L-2 S^2 }{(1 + S) (1+6L-4 S )}\  c(TB),
\end{aligned}
$$
where  $L=c_1(\mathscr{L})$, $S$ is the class of the divisor $S$, $c(TB)$ is the total Chern class of the tangent bundle of the base $B$ of the Weierstrass model, and $\int_B A=\int A \cap B$ is the degree in the Chow ring of  the base $B$.
\end{thm}
 We determine  the Euler characteristic for a $d$-dimensional base $B$  via the coefficient of $t^d$,  after substituting  $L\to t L$,  $S\to t S$, and replacing  $c(TB)$  by the Chern polynomial $c_t(TB)=1+c_1 t + c_2 t^2 + \cdots+ c_d t^d$ with $c_i=c_i(TB)$. 
We give the results for threefolds and fourfolds below. 
\begin{lem}
If the base is a surface,  $Y$ is a threefold, and the Euler characteristic for each model is 
\begin{align}
\begin{array}{|c|c|c|}
\hline
\textnormal{G}_2  &  12 (c_1 L-6 L^2+4 L S-S^2) \\\hline
\textnormal{Spin($7$)} &4 (3 c_1 L-18 L^2+16 L S-5 S^2)  \\\hline
\textnormal{Spin($8$)} & 12 (c_1 L-6 L^2+6 L S-2 S^2) \\\hline
\end{array}\nonumber
\end{align}
\end{lem}

We specialize to the Calabi--Yau case by requiring  $L=-K_B$. In the case of a Calabi--Yau threefold, we have the following result. 
\begin{lem}
For Calabi--Yau threefolds, let $S$ be the divisor over which we have the fiber I$_0^*$. If $g$ denotes the genus of $S$ and $K$ denotes the canonical class of the base, then the Euler characteristic for each model is
$$
\begin{array}{|c|c|c|}
\hline
\textnormal{G}_2  &  -60 K^2+96(1-g)+36 S^2 \\\hline
\textnormal{Spin($7$)}  & -60 K^2 + 128 (1-g) +44 S^2  \\\hline
\textnormal{Spin($8$)} & -60K^2+ 144(1-g) +48 S^2\\\hline
\end{array}
$$
\end{lem}

\begin{thm}\label{lem:Hoddge}
For Calabi--Yau threefolds, let $S$ be the divisor over which we have the fiber I$_0^*$. If $g$ denotes the genus of $S$ and $K$ denotes the  canonical class of the base, then the Hodge numbers $h^{1,1}$ and $h^{1,2}$ for each model are
$$
\begin{array}{|c|c|c|c|}
\hline
&  h^{1,1} &h^{1,2} \\ \hline
\textnormal{G}_2   &13-K^2 & 13+29 K^2-18 S^2+48 (g-1)  \\\hline
\textnormal{Spin($7$) }& 14-K^2& 14+29 K^2-22 S^2+64 (g-1)\\\hline
\textnormal{Spin($8$)} & 15-K^2& 15+29 K^2-24 S^2+72 (g-1) \\\hline
\end{array}
$$
\end{thm}

\begin{lem}
If  $Y$ is a fourfold, then the Euler characteristic  for each model is 
$$
\begin{aligned}
\begin{array}{|c|c|c|}
\hline
\textnormal{G}_2  & 12 (-6 c_1 L^2+4 c_1 L S-c_1 S^2+c_2 L+36 L^3-42 L^2 S+17 L S^2-2 S^3)
  \\\hline
\textnormal{Spin($7$)} & 4 (-18 c_1 L^2+16 c_1 L S-5 c_1 S^2+3 c2 L+108 L^3-166 L^2 S+89 L S^2-15 S^3)
  \\\hline
\textnormal{Spin($8$)} & 12 (-6 c_1 L^2+6 c_1 L S-2 c_1 S^2+c_2 L+36 L^3-60 L^2 S+34 L S^2-6 S^3)
  \\\hline
\end{array}
\end{aligned}
$$
\end{lem}

\begin{lem}
If the base is a threefold and $Y$ is a Calabi--Yau fourfold,  the Euler characteristic  for each model is
\begin{align}
\begin{array}{|c|c|c|}
\hline
\textnormal{G}_2  & -12 \left(c_2 K+30 K^3+38 K^2 S+16 K S^2+2 S^3\right)
 \\\hline
\textnormal{Spin($7$)} & -12 \left(c_2 K+30 K^3+50 K^2 S+28 K S^2+5 S^3\right)
  \\\hline
\textnormal{Spin($8$)} &  -12 \left(c_2 K+30 K^3+54 K^2 S+32 K S^2+6 S^3\right)
  \\\hline
\end{array}\nonumber
\end{align}
\end{lem}

\section{G$_2$-models}\label{Sec:G2}
A G$_2$-model is a Weierstrass model with a  geometric fiber of type I$^*_0$  over the generic point of a divisor $S$ of the base such that the auxiliary polynomial $P(T)$ is $\kappa$-irreducible. 
We distinguish two  types of G$_2$-models, depending on the Galois group $\text{Gal}(\kappa'/\kappa)$.
If  the discriminant $\Delta(P)$ of the associated cubic polynomial $P(T)$ is a perfect square modulo $s$, then the Galois group is $\mathbb{Z}/3\mathbb{Z}$,  and we 
call such a model a G$_2^{\mathbb{Z}/3\mathbb{Z}}$-model. Otherwise the Galois group is  the symmetric group $S_3$, and the model is called a  G$_2^{S_3}$-model. 
Geometrically, these two types have different  fiber structures  from  codimension-two.

\begin{thm}[Canonical form for G$_2$-models]\label{Thm:G2Can}
 \quad 
 A G$_2^{S_3}$-model  can always be  written in the following canonical form 
\begin{equation}\nonumber
\text{G}_2^{S_3} \quad  y^2z=x^3 + s^{2} f xz^2 +s^{3} gz^3, \quad v(f)\geq 0, \quad v(g)=0.
\end{equation}
The polynomial 
$P(T)=T^3+f T +g$ is an irreducible cubic in $\kappa[T]$ and $\Delta(P)=4 f^3+ 27 g^2$ is its discriminant. 
If $\Delta(P)$ is not a perfect square in  $\kappa$,  the Galois group $\text{Gal}(\kappa'/\kappa)$ is $S_3$. Otherwise, it is  $\mathbb{Z}/3\mathbb{Z}$. 
The $j$-invariant is $j=1728 \frac{4 f^3}{4f^3+27 g^2}$, which  varies over $S$.
\end{thm}
\begin{proof}
Follows directly from Step 6 of Tate's algorithm and Lemma \ref{lem:Gal.Cubic}.
\end{proof}
The generic case of a fiber of type I$_0^*$ is when the Galois group  $\text{Gal}(\kappa'/\kappa)$ is the symmetric group $S_3$ and requires $v(f)=0$, i.e.  
\begin{equation}\label{Eq:G2S3}
\text{G}_2^{S_3} \quad  y^2z=x^3 + s^{2} f xz^2 +s^{3} gz^3, \quad v(f)=v(g)=0.
\end{equation}
If we increase the valuation of $f$, the Galois group is automatically $\mathbb{Z}/3\mathbb{Z}$ since $\Delta(P)$ will be a perfect square modulo $s$. 
Note that $g$ cannot be a perfect cube modulo $s$ because  otherwise, $P(T)=T^3+g$ will have three $\kappa$-roots and the model would be a Spin($8$)-model instead of   a G$_2$-model. Hence, we have in this case 
\begin{equation}\label{Eq:G2Z31}
\text{G}_2^{\mathbb{Z}/3 \mathbb{Z}} \quad  y^2z=x^3 + s^{3+\alpha} f xz^2 +s^{3} gz^3,\quad \alpha\geq 0, \quad v(f)=v(g)=0,\quad 
\text{$g$ is not a cube modulo $s$}.
\end{equation}
 
In section \ref{Sec:G2Z3}, we prove that this   G$_2^{\mathbb{Z}/3\mathbb{Z}}$-model is compatible with a crepant resolution after some mild assumptions. Moreover, we allow $\Delta(P)$ to be a perfect square modulo $s$ using more complicated coefficients. For example, a Weierstrass model inspired by well-known examples in number theory is 
\begin{equation}\label{Eq:G2Z32}
\text{G}_2^{\mathbb{Z}/3 \mathbb{Z}} \quad    y^2=x^3 + s^{2}  (-3  a  r+s q)  x +s^{3} ( a^2 r +ar ^2+s t).
\end{equation}
This model suffers from $\mathbb{Q}$-factorial terminal singularities obstructing the existence of a  crepant resolution if the base is of dimension three or higher. 
When the base is a surface, this model does have a crepant resolution, and its fiber structure (see Figure \ref{Fig:G2Z3term}) is much richer than that of equation \eqref{Eq:G2Z31}  (see Figure \ref{Fig:G2Z3}). 
The fiber structure of the crepant resolution of a $\text{G}_2^{S_3}$-model  described by equation \eqref{Eq:G2S3} is presented in Figure \ref{Fig:G2S3}.

Finally, we  point out that all the different G$_2$-models discussed have  (at best) a unique  crepant resolution. In other words, a smooth 
G$_2$-model does not have flops. This can be explained numerically by studying the vertical curves produced by the codimension-two degenerations. 
They have intersection numbers with the fibral divisors corresponding to  weights of the representation $\mathbf{7}$ of $\mathfrak{g}_2$; thus, the  corresponding hyperplane arrangement I$(\mathfrak{g}_2, \mathbf{7})$  has only one chamber.

\begin{thm}[Crepant resolutions for G$_2$]\label{Thm:G2Res}
Assuming that $V(s)$ and $V(g)$ intersect transversally, the  sequence of  blowups that define a crepant resolution of the Weierstrass model 
$\mathscr{E}_0:\quad  y^2-(x^3 + s^{2} f  x+s^3 g)=0$ is given to be
\begin{subequations}
\begin{equation}\nonumber
  \begin{tikzcd}[column sep=huge] 
  X_0  \arrow[leftarrow]{r} {(x,y,s|e_1)} & \arrow[leftarrow]{r}{(y,e_1|e_2)} X_1 &X_2.
  \end{tikzcd}
\end{equation}
The proper transform of the Weierstrass model  is the vanishing locus of
\begin{equation}\nonumber
 F= e_2 y^2-e_1(x^3 + s^{2+\alpha} e_1^\alpha e_2^\alpha f  x+s^3 g), \quad \alpha\geq 0.
\end{equation}
The relative projective coordinates are   $[e_1 e_2  x:e_1 e_2^2  y: z][x:e_2y:s][y:e_1]$.
\end{subequations}
\end{thm}
\begin{proof}
We check smoothness in charts. By our assumptions, $(x,y,s,g)$ forms a regular sequence and can be extended to a local set of coordinates. 
This will allow us to take derivatives with respect to $s$ and $g$. 
The first blowup is done in three charts.  
\begin{enumerate}
\item 
$
(x,y,s)\to (xy, y, ys)
$
The defining equation is 
\begin{equation}\nonumber
F_{(1)}=  1-y(x^3 + s^{2+\alpha} f  x+s^3 g). 
\end{equation}
There are no singularities left in this chart  since $F$ and $\partial_y F$  cannot both vanish at the same time. 
The center of the second blowup is not visible in this chart since $y\neq 0$. 
\item 
$
(x,y,s)\to (x, xy, xs)
$
\begin{equation}\nonumber
 F_{(2)}=y^2-x(1 + s^{2+\alpha} f +s^3 g). 
\end{equation}
There is a singularity at $(y,x,1 + s^{2+\alpha} f +s^3 g)$.
The exceptional divisor is $V(x)$. Hence, in this chart, the second blowup has center $(x,y)$ and requires two charts. 
\begin{enumerate}
\item $(x,y)\to (x, yx)$
 $$
 F_{(2,1)}= x y^2-(1 + s^{2+\alpha} f +s^3 g). 
$$
$V(F_{(2,1)})$ is smooth in this chart since  $\partial_x F_{(2,1)}$, $\partial_f F_{(2,1)}$ (or $\partial_g F_{(2,1)}$), and $F_{(2,1)}$ cannot all vanish at the same time. 
\item  $(x,y)\to (xy, x)$
$$
 F_{(2,2)}= y-x(1 + s^{2+\alpha} f +s^3 g). 
$$
$V(F_{(2,2)})$ is smooth since $\partial_y F_{(2,2)}$ is a unit. 
\end{enumerate}

\item 
$
(x,y,s)\to (sx, sy, s)
$
\begin{equation}\nonumber
 F_{(3)}=   y^2-s(x^3 + s^\alpha f  x +g). 
\end{equation}
We have  double point singularities at $(y,s,x^3 + f  x+g)$. The exceptional divisor is $V(s)$;  hence, the second blowup has center $(y,s)$, which requires two charts. 

\begin{enumerate}
\item 
$
(y,s)\to (y , ys)
$
\begin{equation}\nonumber
 F_{(3,1)}= y-s(x^3 + s^\alpha f  x +g). 
\end{equation}
This is smooth, as can be demonstrated   by taking the derivative with respect to $y$.
\item 
$
(y,s)\to (ys, s)
$
\begin{equation}\nonumber
F_{(3,2)}=sy^2-(x^3 + s^\alpha f  x +g). 
\end{equation}
Since $\partial_g F_{(3,2)}$ is a unit, there are no singularities. 
\end{enumerate}
\end{enumerate}
\end{proof}

\subsection{$\text{G$_2$}^{S_3}$-model}

We recall from Theorem  \ref{Thm:G2Res} that, after the blowup, the elliptic fibration is cut out by
\begin{equation}\nonumber
 F= e_2 y^2-e_1(x^3 + s^{2+\alpha} e_1^\alpha e_2^\alpha f  x+s^3 g), \quad \alpha\geq 0.
\end{equation}
The relative projective coordinates are   $[e_1 e_2  x:e_1 e_2^2  y: z][x:e_2y:s][y:e_1]$. 
The irreducible components of the generic fiber over $S$ are $C_0$, $C_1$, and $C_2$. 
The curve $C_a$ is also the generic fiber of the fibral  divisor $D_a$ ($a=0,1,2$), which are given by 
\begin{equation}
\begin{cases}
D_0&\quad :  \quad s= ze_2 y^2-e_1  x^3=0 \\
D_1 &\quad:  \quad  e_1=e_2=0 \\
D_2 &\quad:  \quad e_2=  x^3 +s^2 f x z^2 + s^3 g z^3=0.\\
\end{cases}
\end{equation}
In the Chow ring $A(Y)$, the divisors $D_a$ are of classes\footnote{
The fibral divisor $D_1$ is the Cartier divisor $V(e_1)$ while the Cartier divisor $V(e_2)$ is $D_1+D_2$.  Hence to the class of $D_2$ is $[V(e_2)-V(e_1)]=E_2-(E_1-E_2)$.}
\begin{align}\nonumber
[D_0] &= [S]-[E_1],\quad 
[D_1] =[E_1]-[E_2],\quad 
[D_2] =2[E_2]-[E_1].
\end{align}
The curve $C_0$ is the only one that touches the section of the Weierstrass model. 
The curves $C_0$ and $C_1$ are smooth geometrically irreducible rational curves.  
Hence, the divisors $D_0$ and $D_1$ are $\mathbb{P}^1$-bundles over $S$:
\begin{equation}\nonumber
D_0=\mathbb{P}_{S} [ \mathscr{L}\oplus \mathscr{O}_{S}], \quad D_1=\mathbb{P}_{S} [ \mathscr{L}^{\otimes 2} \oplus\mathscr{S}],
\end{equation}
where $\mathscr{L}$ is the fundamental line bundle of the Weierstrass model and $\mathscr{S}$ is the line bundle $\mathscr{O}_B(S)$. 
The curve  $C_2$ splits into three geometrically irreducible rational curves  in the splitting field of the polynomial $P(T)$. 
The divisor $D_2$ is a triple cover of the $\mathbb{P}^1$-bundle $\mathbb{P}_{S} [ \mathscr{L}^{\otimes 3} \oplus\mathscr{S}^{\otimes 2}]$ ramified over  $V(s, 4 f^3+2 g^2)$.
Over the generic point of $V(s, 4 f^3+2 g^2)$, the curve $C_2$ factorizes into a line and a double line. 
The full geometric generic fiber over $V(s, 4 f^3+2 g^2)$ is an incomplete $\widetilde{D}_5$, while the  full generic fiber over $V(s, 4 f^3+2 g^2)$ is an incomplete $\widetilde{B}_4^t$ if the dimension of the base is three or higher. 
Over the generic point of $V(s, f, g)$, $C_2$  degenerates further into a triple line $3 C'_2$, where
\begin{equation}\nonumber
C_2':  e_2=x=f=g=0.
\end{equation}
It follows that  the full fiber structure is an  incomplete $\widetilde{\text{E}}_6$, formed by three rational curves of multiplicity $1$, $2$, and $3$, 
namely  $C_0+2 C_{12}+3 C'_2$. 
The proper morphism $f:D_2\to S$ has a Stein factorization 
$$f:D_2\xrightarrow{f'} S'\xrightarrow{\pi} S,$$
where  $\pi: S'\xrightarrow{} S$ is a finite map, as well as   a triple cover  of $S$ with ramification divisor $4f^3+27g^2$. The proper morphism $f':D_2\to S'$ has connected fibers, and endows $D_2$ with  the structure of a $\mathbb{P}^1$-bundle over $S'$ such that
  $$D_2\cong \mathbb{P}_{S'}[\pi^*  (\mathscr{L}^{\otimes 3} \oplus\mathscr{S}^{\otimes 2})]\to S'.$$

The node $C'_2$ has the  quasi-minuscule weight $\boxed{1,-2}$, whose Weyl orbit corresponds to the non-zero weights of the fundamental representation $\mathbf{7}$ of G$_2$. 
 It follows that the representation associated with the G$_2$-model is the direct sum of the adjoint representation (which is always present) and the fundamental representation $\mathbf{7}$.

\subsection{$G^{\mathbb{Z}/3\mathbb{Z}}_2$-model}\label{Sec:G2Z3}

A G$_2$-model over $S=V(s)$  with Galois group $\mathbb{Z}/3\mathbb{Z}$  is given by the Weierstrass model
\begin{equation}\nonumber
\mathscr{E}_0:\quad  y^2=x^3 + s^{3+\alpha} f  x +s^3 g , \quad \alpha\geq 0,
\end{equation}
where $g$ is not a perfect square modulo $s$.
This $G^{\mathbb{Z}/3\mathbb{Z}}_2$ is a specialization of  $G^{S_3}_2$, obtained by increasing the valuation of $c_4$. The $G^{\mathbb{Z}/3\mathbb{Z}}_2$-model distinguishes itself by the behavior of its $j$-invariant over $S$ and its fiber degeneration. 

\begin{lem}
The value of the $j$-invariant for the generic fiber over  $S$  is $0$. 
\end{lem}
\begin{proof}
An elliptic fibration $y^2= x^3 + F x + G$ has a $j$-invariant $j=1728 \frac{4 F^2}{ (4 F^3 + 27 G^2)}$. In this case, since $v_S(F^3)>v_{S}(G^2)$, $j=0$ over the generic point of $S$.  
\end{proof}
The crepant resolution of this elliptic fibration follows the same blowup as in the general case. But the proper transform is now 
\begin{equation}\nonumber
Y: \quad  e_2 y^2=e_1(x^3 +e_2^{1+\alpha} e_1^{1+\alpha} s^{3+\alpha} f  x +s^3 g). 
\end{equation}
In $X_2$, the projective coordinates are
\begin{equation}\nonumber
[e_1 e_2  x:e_1 e_2^2  y: z=1][x:e_2y:s][y:e_1].
\end{equation}
The divisor  $D_2$ is now simply 
\begin{equation}\nonumber
D_2 \quad:  \quad e_2=  x^3 +s^3 g=0.
\end{equation}
The divisor $D_2$ is still a triple cover of the $\mathbb{P}^1$-bundle $\mathbb{P}_{S} [ \mathscr{L}^{\otimes 3} \oplus\mathscr{S}^{\otimes 2}]$, except  now branched at $V(s,g)$. 
The node $C_2$ splits into three irreducible components  in the splitting field of the polynomial $T^3- g$. Hence, the Galois group of the splitting field is $\mathbb{Z}/3 \mathbb{Z}$.  

The arithmetic degeneration only exits  to a $D_4$, while the geometric degeneration to an incomplete $\widetilde{\text{E}}_6$ over $V(s, g)$. 
The node $C'_2$ has the  quasi-minuscule weight $\boxed{2\ -1}$, whose  Weyl orbit corresponds to the non-zero weights of the  fundamental representation $\mathbf{7}$ of G$_2$. Thus, the  representation $\mathbf{7}$ is quasi-minuscule.

If the base is  of dimension three or higher, there is an {\em arithmetic degeneration} 
\begin{equation}\nonumber
\widetilde{\text{G}}_2\longrightarrow \widetilde{\text{D}}_4,
\end{equation}
as the $\widetilde{\text{G}}_2$ becomes a fiber $\widetilde{\text{D}}_4$ over the loci where the polynomial $x^3+s^3 g$ splits completely.  
An example of such a loci is over the intersection of $s$ with the locus 
$$V(\xi^3+g+s r),$$
 where $\xi$ is a section of the line bundle $\mathscr{L}^2\otimes \mathscr{S}^{-1}$.
 
 \begin{exmp}
 If $B=\mathbb{P}^3$, $\mathscr{L}=\mathscr{O}_{\mathbb{P}^3}(4)$, and $\mathscr{S}=\mathscr{O}_{\mathbb{P}^3}(2)$, then $g$ is a section of $\mathscr{O}(18)$ and $\xi$ is a section of $\mathscr{O}(6)$. 
Such a curve depends on twenty parameters, and one can easily write a family of them passing through any arbitrary point of $S$. 
  \end{exmp}
The same weight  $\boxed{2\  -1}$ of the representation $\mathbf{7}$ appears for the arithmetic degeneration, i.e. 
\begin{equation}\nonumber
(\widetilde{G}_2\longrightarrow \widetilde{\text{D}}_4)\Longrightarrow \quad \text{weight} \quad \boxed{2\  -1} \quad \text{of the  representation}\quad \mathbf{7}.
\end{equation}
We remark that this is one third of the weight corresponding to a generic fiber of $D_{1}$.

\subsection{ $\text{G$_2$}^{\mathbb{Z}/ 3\mathbb{Z}}$ and terminal singularities }\label{Sec:G2Z3Terminal}

The $\text{G$_2$}^{\mathbb{Z}/ 3\mathbb{Z}}$-model we have considered  in the previous section has a crepant resolution for  a base of arbitrary dimension. 
It was uniquely defined by the valuation of the Weierstrass coefficients ($v_S(c_4)\geq 3$, $v_S(c_6)=3$). 
In this section, we explore  an example of a  $\text{G$_2$}^{\mathbb{Z}/ 3\mathbb{Z}}$-model that has the same valuations as the $\text{G$_2$}^{S_3}$  ($v_S(c_4)= 2$, $v_S(c_6)=3$) model. 
The Galois group $\mathbb{Z}/ 3\mathbb{Z}$  is enforced by requiring that the  discriminant $\Delta(P)$ is a perfect square in the residue field $\kappa(\eta)$. However, we will encounter non-trivial $\mathbb{Q}$-factorial terminal singularities  when the base is of dimension three or higher. 

Consider the Weierstrass model given by\footnote{This model is inspired by the family of cubics $F=x^3- b x + b c$, which   has Galois group $\mathbb{Z}/ 3\mathbb{Z}$  when $4b-27 c^2$ is a perfect square as it ensures that its discriminant  $\Delta(F)=b^2(4 b-27c^2)$ is a perfect square. A famous example of this type is the  family of cubics $F_m=x^3+ m x^2 + (m-3) x + 1$  introduced by Shank in the definition of the 
{\em simplest cubic fields}   \cite{Shanks}:  after completing the cube in $x$, $F_m$ takes the form $F=x^3- b x + b c$ with $b= (m^2 + 3 m + 9)/3$ and $c=(2m+3)/9$.
}
\begin{equation}\label{z3W}
\mathscr{E}_0: \quad  y^2=x^3 + s^{2}  (-3  a  r+s q)  x +s^{3} ( a^2 r +ar ^2+s t).
\end{equation}
The auxiliary polynomial $P(T)$ and its discriminant $\Delta(P)$ are given by 
\begin{equation}\nonumber
P(T)=T^3-3 a r T+  a r(a+r), \quad \text{with } \quad \Delta(P)=27 a^2 r^2 (a-r)^2.\nonumber
\end{equation}
Note that $P(T)$ is irreducible and $\Delta(P)$ is a perfect square. Hence, the splitting field of $P(T)$ has Galois group $\mathbb{Z}/3\mathbb{Z}$, which means  we have a  $\text{G$_2$}^{\mathbb{Z}/ 3\mathbb{Z}}$-model. 
 
Using the same sequence of blowups as in section  \ref{Sec:Crepant}, we have 
\begin{equation}\nonumber
\mathscr{E}_2: \quad  e_2y^2=e_1\Big({x^3 + s^{2}  (-3  a  r+se_1 e_2 q)  x +s^{3} ( a^2 r +ar ^2+se_1 e_2 t)}\Big).
\end{equation}
The irreducible components of the generic fiber over $S$ are 
\begin{equation}\nonumber
\begin{cases}
D_0&:  \quad s= ze_2 y^2-e_1  x^3=0\\
D_1 &:  \quad  e_1=e_2=0\\
D_2 &:  \quad \frac{e_2}{e_1}=  x^3 -3 a rs^2  x + s^3 a r (a+r)=0.
\end{cases}
\end{equation}
 The curve $C_2$  (i.e. the generic fiber of $D_2$)  is irreducible over $\kappa$. The degeneration of the fibers are characterized by the irreducible components of $\Delta(P)=27 a^2 r^2 (a-r)^2$. 
Along $ar=0$, $P(T)$ specializes to a perfect cube $P(T)=T^3$. This means all  three geometric components of $C_2$ coincide, yielding the 
 fiber of type $1-2-3$, which is an incomplete $\widetilde{\text{E}}_6$. 
Along $a-r=0$, $P(T)=(T-a)^2(T+2a)$. This means  two of the three geometric components of $C_2$ coincide,   giving an incomplete fiber of type $\widetilde{\text{D}}_5$. 
Finally, along  $a+r=0$, $P(T)=T(T^2+3 a^2)$. This means $P(T)$ factorizes into a linear term and an irreducible quadratic term, giving a   fiber of type I$_0^{*\text{ss}}$ ($\widetilde{\text{B}}^t_3$).  
We have three codimension two loci over which the fiber $C_2$ splits into  a double curve and a curve. These are geometric degenerations  to an incomplete $\widetilde{\text{E}}_6$  along $V(s,a) $, $V(s,r)$, and $V(s,a+r)$.
 If the base  is of dimension three or higher, the loci $V(s,a)$, $V(s,r)$,  $V(s,a-r)$, and $V(s,a+r)$  intersect on a  codimension three loci at $V(s,a,r)$. 
 Over $V(s,a+r)$, the generic fiber is of type $1-2-3$.
 If  the base $B$ is a surface, $\mathscr{E}_2\to\mathscr{E}_0$ is a crepant resolution. 
 If the base has dimension three or higher, we can rewrite the equation as 
 $$
 \mathscr{E}_2: \quad  e_2(y^2-s^4e_1^2 t)=e_1\Big({x^3 + s^{2}  (-3  a  r+se_1 e_2 q)  x +s^{3} ar( a +r)}\Big),
 $$
where the singularity is in the patch $e_1 s \neq 0$. 
 Analytically, this is a binomial hypersurface of type $V( u_1 u_2 u_3-w_1 w_2 w_3)$.

 There are terminal singularities  along the codimension four loci $(x, y^2-e_1^2 s^4 t,e_2,a,r)$. 
 By the Grothendick-Samuel's theorem,  $\mathscr{E}_2$ does not have  a crepant resolution if the base  is of dimension three or higher. 
 This is because $\mathscr{E}_2$ is  locally a complete intersection nonsingular up to codimension three with terminal singularities in codimension four.

\section{Spin($7$)-models}\label{Sec:Spin7}

We recall that a  fiber is of type I$_{0}^{*\text{ss}}$  if its dual graph is the Dynkin diagram of type $\widetilde{\text{B}}^t_3$ and its geometric fiber is the Kodaira fiber of type I$_0^*$ (with dual graph $\widetilde{\text{D}}_4$).
Tate's algorithm is performed with respect to the valuation ring associated to  a smooth divisor $S=V(s)$ with generic point $\eta$ and residue field $\kappa$. The fiber over $\eta$ is of Kodaira type I$_0^*$\, and the  auxiliary polynomial of Step 6 is $P(T)=T^3+a_{2,1} T^2+ a_{4,2} T+ a_{6,3}$. The arithmetic fiber is of type I$_0^{*\text{ss}}$ when the splitting field of $P(T)$ defines a quadratic field extension $\kappa'$ of the residue field $\kappa$. This means that $P(T)$ has a unique $\kappa$-rational solution and that the Galois group Gal($\kappa'/\kappa$) is the cyclic group $\mathbb{Z}/2\mathbb{Z}$.  

We give convenient canonical forms for Spin($7$)-models in Theorem \ref{Thm:Spin7Can}. 
All the canonical forms we deal with  have the $\kappa$-rational point of $P(T)$ at the origin $T=0$. We distinguish between two cases by the value of the $j$-invariant at the generic point $\eta$ in $S$.

\begin{thm}[Canonical forms for Spin($7$)-models]\label{Thm:Spin7Can}
 \quad 
\begin{itemize}
\item  A Spin($7$)-model such that $j(\eta)\neq 0, 1728$ always has a canonical form that can be written as 
\begin{equation}\nonumber
\text{$Spin($7$):\quad \mathscr{E}_0=$} V(zy^2-x^3- a_{2,1} s x^2 z -s^2 a_{4,2}x z^2- s^{4+\beta} a_{6,4+\beta} z^3),  \quad j\neq 0, 1728, \quad \beta\in \mathbb{Z}_{\geq 0},
\end{equation}
with
\begin{equation}\nonumber
P(T)=T(T^2 +a_{2,1} T +a_{4,2} ), \quad \Delta(P)=a_{4,2}^2 (a_{2,1}^2-4 a_{4,2}), 
\end{equation}
 where $a_{2,1}$ is not zero modulo $s$ and  $a_{2,1}^2-4 a_{4,2}$ is  not  a perfect square modulo $s$.
\item A Spin($7$)-model with  $j(\eta)=1728$ can always be written in such a way that 
$v_S(a_1)\geq 1$, $v_{S}(a_2)\geq 2$, $v_{S}(a_3)\geq 2$, $v_S(a_4)=2$, and $v_{S}(a_6)\geq 4$. Thus, it can be put in the canonical form 
 \begin{equation}\nonumber
\text{Spin($7$)}:\quad \mathscr{E}_0=V(zy^2-x^3- s^2 a_{4,2} x z^2- s^{4+\beta} a_{6,4+\beta} z^3), \quad j=1728, \quad \beta\in \mathbb{Z}_{\geq 0},
\end{equation}
with 
\begin{equation}\nonumber
P(T)=T(T^2+a_{4,2} ), \quad \Delta(P)=-4 a_{4,2}^3,
\end{equation}
 where $a_{4,2}$ is not a perfect square modulo $s$.
\end{itemize}
\end{thm}

\begin{proof}

Without loss of generality, we can solve the arithmetic condition requiring $P(T)$ to have a $\kappa$-rational point,  by requiring  $v_S(a_6)\geq 4$. 
This is essentially the same as performing a translation that puts the unique $\kappa(\eta)$-rational root of $P(T)$ at $T=0$. We then have to restrict the valuation of  $a_4$ to be exactly two; otherwise, $v(a_4)\geq 3$, the discriminant of $P(T)$ will be zero in $\kappa(\eta)$, and $P(T)$ will have a double root. 

The discriminant $\Delta(P)$ depends only on $a_{4,2}$ and $a_{2,1}$. Since the valuation of $a_4$ is fixed, it is interesting to explore how the geometry depends on the valuation of 
 $a_2$. The valuation of $a_2$ characterizes two distinct types of Spin($7$)-models  as  the  $j$-invariant and the residual discriminant $\Delta(P)$  have different behavior when $v_S(a_2)=1$  or $v_S(a_2)\geq 2$. 
 The $j$-invariant takes the generic value $j=1728$ over the generic point of $S$ when $v_S(a_2)\geq 2$ and varies over $S$.  
 Moreover,  the residual discriminant $\Delta(P)$ is composed of two distinct components when  $v_S(a_2)=1$,  and the two components coincide when  $v_S(a_2)\geq 2$.  
To avoid a trivial Galois group, we assume that the discriminant of the quadratic part of $P(T)$ is not a perfect square. 
When $v_S(a_2)\geq 2$, we can also complete the cube in $x$, and take the canonical form to be in a short Weierstrass form without changing the conditions $v_S(a_4)=2$ and $v_S(a_6)\geq 4$. 

\end{proof}

\subsection{First crepant resolution of the  Spin($7$)-models}\label{Sec:Spin7First}

We first consider a resolution of the Spin($7$)-model obtained by blowing up reduced monomial ideals.  
In the following section,  we consider a flop of that resolution. 
We assume the  Theorem \ref{Thm:Spin7Can} hold. 
\begin{thm}\label{Thm:Spin7Res1}
Let $\mathscr{E}_0\longrightarrow B$ be a  Spin($7$)-model given  in Theorem \ref{Thm:Spin7Can}.
Let $Y^+$  be the proper transform of  $\mathscr{E}_0$ after the following  sequence of blowups starting with the ambient space 
  $X_0=\mathbb{P}_B[\mathscr{L}^{\otimes 2}\oplus \mathscr{L}^{\otimes 3}\oplus \mathscr{O}_B]$:
\begin{equation}\nonumber
  \begin{tikzcd}[column sep=huge] 
  X_0  \arrow[leftarrow]{r} {(x,y,s|e_1)} & \arrow[leftarrow]{r}{(y,e_1|e_2)} X_1 & \arrow[leftarrow]{r}{(x,e_2|e_3)} X_2  & X_3^+
  \end{tikzcd} .
\end{equation}
 If $V(a_{4,2})$ and  $V(a_{6,4})$ are smooth hypersurfaces  in $B$  that meet $S=V(s)$ transversally, then   $Y^+\to \mathscr{E}_0$ is either a crepant resolution or a crepant partial resolution with  terminal $\mathbb{Q}$-factorial singularities, as indicated in  the  \ table below. 
\begin{center}
\begin{tabular}{|c|c|c|c|}
\hline 
&  $v_S(a_6)=4$ &$v_S(a_6)=5$ & $v_S(a_6)\geq 6$ \\
\hline 
$dim B=2$& Crepant resolution& Crepant resolution &   Term. $\mathbb{Q}$-factor. Sing.   \\
\hline 
$dim B\geq 3$& Crepant resolution& Term. $\mathbb{Q}$-factor. Sing.  & Term. $\mathbb{Q}$-factor. Sing.  \\
\hline 
\end{tabular}
\end{center}
\end{thm}

\begin{proof}[Proof of the theorem]

Assume that $V(s)$ and $V(a_{4,2})$  are smooth varieties intersecting transversally.  
For  notational simplicity, we take $b=a_{4,2}$ and $c=a_{6,4+\beta}$. When $\alpha=0$, we take  $a=a_{2,1}$; however,  if $\alpha>0$, we complete the cube in $x$ such that $a_2=0$. 
To examine  singularities, it is enough to work in the patch $z=1$ since  the section  $(z=x=0)$ is always smooth. To check smoothness, we work in local patches. The defining equation is 
\begin{equation}\nonumber
F=y^2 -(x^3 + a s  x^2 + b s^{2}  x + c s^{4+\beta}), \quad \beta\geq0,
\end{equation}
where $a=0$ if $v(a_2)\geq 2$. The variable $b$ cannot be identically zero for Spin($7$), as otherwise, the polynomial $P(T)$ will have a double root modulo $s$. 
The first blowup is centered at $(x,y,s)$ and requires the following three charts:
\begin{enumerate}
\item $(x,y,s)\to (xy,y, sy)$
\begin{equation}\nonumber
F_{(1)}=1 -y(x^3 + a s  x^2 + b s^{2}  x + cy^{1+\beta} s^{4+\beta}).
\end{equation}
$V(F_{(1)})$  is smooth as there is no solution for $F_{(1)}=\partial_b F_{(1)}=\partial_y F_{(1)}=\partial_x F_{(1)}=0$. 

\item $(x,y,s)\to (x,yx, sx)$
\begin{equation}\nonumber
F_{(2)}=y^2 -x (1 + a s   + b s^{2}   + c s^{4+\beta} x^{1+\beta}).
\end{equation}
$V(F_{(2)})$  has a double point  singularity at $(y,x,1 + a s   + b s^{2})$. The second blowup is centered at $(x,y)$ and is implemented in two charts. 
\begin{enumerate}
\item $
(x,y)\to (x, yx)
$
\begin{equation}\nonumber
F_{(2,1)}=xy^2 - (1 + a s   + b s^{2}   + c s^{4+\beta} x^{1+\beta}).
\end{equation}
$V(F_{(2,1)})$ is smooth since $F_{(2,1)}=\partial_y F_{(2,1)}=\partial_b F_{(2,1)}=0$ does not have a solution. 
\item 
$
(x,y)\to (xy, y)
$
\begin{equation}\nonumber
F_{(2,2)}=y -x (1 + a s   + b s^{2}   + c s^{4+\beta} x^{1+\beta} y^{1+\beta}).
\end{equation}
This is smooth as $\partial_b F_{(2,2)}$ and $\partial_y F_{(2,2)}$ cannot vanish simultaneously. 
\end{enumerate}

\item
$(x,y,s)\to (sx,sy,s)$
\begin{equation}\nonumber
F_{(3)}=y^2 -s(x^3 + a   x^2 + b   x + c s^{1+\beta}), \quad \beta\geq 0.
\end{equation}
We have a singularity at $(y, s, x(x^2+ a x + b))$. 
The second blowup is centered at $(y,s)$.
We have to perform the next blowup in two patches. 
\begin{enumerate}
\item  $(y,s)\to (y,y s)$
\begin{equation}\nonumber
F_{(3,1)}=y -s(x^3 + a   x^2 + b   x + c y^{1+\beta}s^{1+\beta}), \quad \beta\geq 0.
\end{equation}
When $\beta>0$, $V(F_{(3,1)})$ is smooth since $F_{(3,1)}=\partial_y F_{(3,1)}=\partial_b F_{(3,1)}=0$   cannot vanish simultaneously. 
 When $\beta=0$, there is a singularity at 
$x=y=1-cs^2=b=0$. This singularity has a crepant resolution  by blowing up $(x,y)$. 
\begin{enumerate}
\item $(x,y)\to (x,xy)$
\begin{equation}\nonumber
F_{(3,1,1)}=y -s(x^2 + a   x + b    + c x^{\beta} y^{1+\beta}s^{1+\beta}), \quad \beta\geq 0.
\end{equation}\nonumber
$V(F_{(3,1,1)})$ is smooth since $\partial_b F_{(3,1,1)}$ and $\partial_y F_{(3,1,1)}$ cannot vanish at the same time. 
\item $(x,y)\to (xy,y)$
\begin{equation}\nonumber
F_{(3,1,2)}=1-s(y^2x^3 + a y  x^2 + b   x + c y^{\beta}s^{1+\beta}), \quad \beta\geq 0.
\end{equation}
$V(F_{(3,1,2)})$ is smooth since $F_{(3,1,2)}$, $\partial_b F_{(3,1,2)}$, and  $\partial_s F_{(3,1,2)}$ cannot vanish at the same time. 
\end{enumerate}
\item $(y,s)\to (sy,s)$
\begin{equation}\nonumber
F_{(3,2)}= sy^2  -( x^3 + a   x^2 + b  x +c s^{1+\beta}). 
\end{equation}
We still have a singularity at $(y,s,x, b)$. 
The next blowup is centered at $(s,x)$. 
We again have two patches to consider. 
\begin{enumerate}
\item  $(x,s)\to (x,sx)$. 
\begin{equation}\nonumber
F_{(3,2,1)}=sy^2  -( x^2 + a   x + b   +c x^{\beta}s^{1+\beta}). 
\end{equation}
This is smooth as it is linear in $b$,  which  can be taken as a local parameter in the base. 
\item
  $(x,s)\to (xs,s)$.
 
 The proper transform is 
 \begin{equation}\nonumber
Y^+:\quad F_{(3,2,2)}=y^2  -(s^2x^3 + a   sx^2 + b  x +c s^{\beta}). 
\end{equation}
When $\beta=0$, $Y^+$ is smooth as it is linear in $c$, which can be used as a local parameter of the base, assuming  $V(c)$ is smooth. 
When $\beta>0$,  we have $F|_{s=0}=y^2-bx$, which is irreducible.  If we localize at $V(s)$, $F=0$ is just a redefinition of $c$ and we trivially have a UFD. It follows by Nagata's criterion of factoriality that 
$Y^+$ is factorial. 

When $\beta=1$, $Y^+$ has  terminal singularities at $(x,y,b,c,s)$, which is in codimension $4$. 
When $\beta>1$, $Y^+$ has terminal singularities at $(x,y,b,s)$, which is in codimension $3$. 
\end{enumerate}
\end{enumerate}
\end{enumerate}
 We thus have the following conclusions:\footnote{We recall that  factoriality is an obstruction for crepant resolutions in presence of terminal singularities.}
\begin{enumerate}
\item If  $\beta=0$, $Y^+$  is smooth if $V(c)$ is smooth. 
\item If $\beta=1$ and $\dim\ B =2$, $Y^+$ is smooth. 
\item If $\beta=1$ and $\dim\ B \geq 3$, $Y^+$ is factorial with terminal singularities at $(x,y,b,c,s)$. Hence, $Y^+$ does not have a crepant resolution. 
\item If $\beta>1$ and $\dim\  B\geq 2$, $Y^+$ is factorial with terminal  singularities at  $(y,x,b,s)$. Hence, $Y^+$ does not have a crepant resolution. 
\end{enumerate}

\end{proof}

The relative projective coordinates of $X_0$ over $B$, $X_{i+1}$ over $X_i$ ($i=0,1$), and $X_3^+$ over $X_2$ are respectively
\begin{equation}\nonumber
[e_1 e_2 e_3^2 x:e_1 e_2^2 e_3^2 y: z=1], \quad [e_3x:e_2e_3y:s],    \quad    [y:e_1],   \quad        [x:e_2].
\end{equation}

The proper transform of $\mathscr{E}_0$ is 
\begin{equation}\nonumber
Y^+:\quad e_2 y^2 =e_1 (e_3^2 x^3 + a_{2,1+\alpha} s^{1+\alpha} e_1^{\alpha} e_2^{\alpha} e_3^{1+\alpha} x^2 + a_{4,2} s^{2}  x + a_{6,4+\beta} e_3^{\beta}e_2^{1+\beta} e_1^{1+\beta} s^{4+\beta}). 
\end{equation}

 We now explore the fiber structure  of the smooth elliptic fibration $\varphi: Y^+\to B$ obtained by the crepant resolution. 
Denoting $C_a$ as the irreducible components of the fiber over the generic point $\eta$ of $S$, we have
\begin{equation}\nonumber
\varphi^*(\eta)=C_0+C_1+2C_2+C_3.
\end{equation}
This curve is a scheme with  respect to the residue field $\kappa(\eta)$. 
The curve $C_0$ is the one touching the section of the elliptic fibration.
These curves are generic fibers for the fibral divisors $D_a$, which are defined as the irreducible components of 
\begin{equation}\nonumber
\varphi^*(S)=D_0+D_1+2D_2+D_3.
\end{equation}
The fibral divisors are
\begin{subequations}
$$
\begin{cases}
D_0:& \quad  s= e_2 y^2 -e_1 e_3^2 x^3=0\\
D_1:& \quad  e_3= e_2(y^2 -a_{6,4} e_1^2 s^4) - a_{4,2} s^2 e_1 x =0\\
D_2:& \quad  e_1=e_2=0\\
D_3: &\quad \frac{e_2}{e_1}=e_3^2 x^2 + a_{2,1} se_3 x + a_{4,2} s^2=0 .
\end{cases}
$$
\end{subequations}

We observe that D$_0$ and D$_2$   are   respectively isomorphic to the projective bundles $\mathbb{P}_S [\mathscr{O}_S\oplus\mathscr{L}]$ and  $\mathbb{P}_S [\mathscr{L}^2\oplus\mathscr{S}]$, while 
the fibers of the divisors D$_1$ and D$_3$ can degenerate  over higher codimension loci. 
The generic fiber (over $S$) of the divisor D$_3$  is not  geometrically irreducible; after a field extension, it splits into two rational curves. 
The divisor $D_3$ is  a double cover of $\mathbb{P}_S[\mathscr{L}^3 \oplus \mathscr{S}^2]$ branched along $V(s,a_{2,1}^2-4 a_{4,2})$.  
We note that $a_{2,1}^2-4 a_{4,2}$  is one of the components of the discriminant $\Delta(P)$ of the associated polynomial $P(T)$.

The only fiber components that can degenerate are $C_1$ and $C_3$. 
The degeneration of the generic fiber are on the components of $\Delta(P)$, and are given by
\begin{subequations}
\begin{equation}\nonumber
V(a_4)
\begin{cases}
C_3\to C_{13}+C_3'\\
C_1 \to C_{13}+C_1'
\end{cases},
\quad
V(a_2^2-4a_4)
\begin{cases}
C_3\to  2C_3''\\
C_1 \to C_1
\end{cases},
\end{equation}

\begin{equation}\nonumber
V(a_2,a_4) 
\begin{cases}
C_3\to 2C_{13}\\
C_1 \to C_{13}+C_1'
\end{cases},\quad 
V(a_4, a_6)
\begin{cases}
& C_3 \to C_{13}+C_3'\\
& C_1\to  C_{13}+2C_{1}''
\end{cases}, \quad
V(a_2,a_4, a_6)
\begin{cases}
& C_3 \to 2C_{13}\\
& C_1\to  C_{13}+2C_{1}''
\end{cases},
\end{equation}

where 
\begin{equation}\nonumber
\begin{cases}
& C_{13}:e_2=e_3=0,\\
& C_3':e_2=e_3 x+a_2 s =0,\\
& C_3'': e_2=e_3 x + \frac{1}{2}a_2 s=0,\\
& C_1':e_3=y^2 -a_6 e_1^2 s^4=0\quad C_1'': e_3= y =0 .
\end{cases}
\end{equation}
\end{subequations}
Since the divisor $D_0$, $D_i$ ($i=1,2,3$) satisfy the linear relation 
$D_0 + D_1 +2 D_2  + D_3\cong 0$,
it is enough to consider the weights with respect to $D_i$   ($i=1,2,3$). 
Using the fiber structure, it is easy to evaluate the weights for each curve by solving the  linear relations below at the level of intersection numbers. 
$$
C_{13}\cong C_3'' \cong C_3'\cong  \frac{1}{2}C_3, \quad  
C_1'\cong C_1-\frac{1}{2}C_{3}, \quad  C_1''\cong \frac{1}{2} C_1-\frac{1}{4}C_{3}. 
$$
\begin{equation}\nonumber
\begin{matrix}
\\
D_0\\
D_1\\
D_2\\
D_3
\end{matrix}
\begin{pmatrix}
C_0 & C_{1} & C_{2} & C_{3} \\
\hline
2  & 0 & -1 & 0 \\
0 & 2 & -1 & 0 \\
-1 & -1& 2 &  -2\\
0 & 0 & -2 & 4
\end{pmatrix}
\quad
\begin{pmatrix}
C_{13}\\
\hline
0 \\
0\\
-1\\
2\\
\end{pmatrix}
\quad
\begin{pmatrix}
C_{3}'\\
\hline
0 \\
0\\
-1\\
2\\
\end{pmatrix}
\quad 
\begin{pmatrix}
C_{3}''\\
\hline
0 \\
0\\
-1\\
2\\
\end{pmatrix}
\quad 
\begin{pmatrix}
C_{1}'\\
\hline
0 \\
2\\
0\\
-2\\
\end{pmatrix}
\quad
\begin{pmatrix}
C_{1}''\\
\hline
0 \\
1\\
0\\
-1\\
\end{pmatrix}.
\end{equation}
The new weight is $\boxed{1\  0  \   -1}$.

\subsection{Second crepant resolution of the Spin($7$)-models}
 In this section, we construct a flop of the crepant resolution obtained in the previous subsection. 
The flop appears after the second blowup. 
The proper transform of $\mathscr{E}_0$ after the second blowup is $\mathscr{E}_2$, which can be suggestively rewritten as 
\begin{equation}\nonumber
Y^- 
\begin{cases}
e_2 (y^2-a_6  e_1 s^{4}) -e_1 x Q &=0\\
Q-( x^2 + a_2 s  x + a_4 s^{2}) &=0.
\end{cases} 
\end{equation}
The first equation emphasizes that $\mathscr{E}_2$ has double point singularities, while the  second equation defines $Q$, which is used in  our next blowup so that it formally resembles  a blowup of a monomial ideal.  
The singular scheme of $\mathscr{E}_2$ is supported along $(e_2,e_1 x, Q, y^2-a_6 e_1 s^4)$, which is in the patch $s e_1\neq 0$. 
Recall that for the first resolution, the center of the blowup is the ideal $(e_2,x)$. In the spirit of Atiyah's flop, this time we blowup the  non-Cartier Weil divisor $(e_2, Q)$, i.e. $D_3$: 
\begin{equation}\nonumber
  \begin{tikzcd}[column sep=huge] 
  & & & X_3^+ \arrow[dashed,leftrightarrow]{dd} \\
  X_0  \arrow[leftarrow]{r} {(x,y,s|e_1)} & \arrow[leftarrow]{r}{(y,e_1|e_2)} X_1 & \arrow[leftarrow]{ru}{(e_2, x|e_3)} \arrow[leftarrow]{rd} [left,midway]{(e_2, Q|e_3)} X_2  &\\
  & & & X_3^-
  \end{tikzcd}
  \end{equation}
The proper transform is 
\begin{equation}\nonumber
Y^- 
\begin{cases}
e_2 (y^2 -a_6 s^{4})-e_1 x Q &=0\\
Q e_3 -(x^2 + a_2 s x + a_4 s^2) &=0;
\end{cases}
\end{equation}
the relative projective coordinates of  $X_{i}\rightarrow X_{i-1}$ ($i=1,2,3$) are
\begin{equation}\nonumber
[e_1 e_2 e_3 x:e_1 e_2^2e_3^2  y: z=1] \quad [x:e_2e_3y:s]    \quad    [y:e_1]  \quad [Q:e_2]; 
\end{equation}
and the  irreducible components of the generic fibers are
\begin{subequations}
\begin{equation}\nonumber
\begin{aligned}
&C_0 : \quad s=e_2 y^2 -e_1 x Q=Q e_3 -x^2=0, \quad  C_1:\quad e_1=e_2=0, \\
& C_2:\quad\frac{e_2}{e_1}=x=Q e_3 - a_4 s^2=0, \quad C_3: \quad e_3=x^2+a_2 s x + a_4 s^2= e_2(y^2 -a_6 e_1 s^4 )-e_1 x Q=0.
\end{aligned}
\end{equation}
They degenerate as follows 
\begin{equation}\nonumber
\begin{aligned}
 V(a_4) &   
\begin{cases}
C_2\to C_{23}\\
C_3 \to C_{23} +C^{(1)}_{3}+C^{(2)}_3
\end{cases},
\quad 
V(a_2^2-4a_4)   
\begin{cases}
C_2\to C_{2}\\
C_3 \to 2C_{3}'
\end{cases},
\quad
V(a_2,a_4) 
\begin{cases}
C_2\to C_{23}\\
C_3 \to 2C_{23}+2C^{(1)}_{3}
\end{cases},
\\
 V(a_4, a_6)   &
\begin{cases}
C_2\to C_{23}\\
C_3 \to C_{23} +2C^{(1')}_{3}+C^{(2)}_3
\end{cases},
\quad
V(a_2,a_4,a_6)   
\begin{cases}
C_2\to C_{23}\\
C_3 \to 2C_{23}+4C^{(1')}_3
\end{cases},
\end{aligned}
\end{equation}
where
\begin{equation} \nonumber
\begin{aligned}
& C_{23}: \quad e_2=e_3=x=0, \quad 
C^{(1)}_{3}: \quad  e_3=x=y^2-a_6 e_1 s^4=0,  
 \quad 
 C^{(1')}_3: \quad  e_3=x=y=0,
\\
& C^{(2)}_3:\quad  e_3=x+a_2 s=e_2 (y^2 -a_6 e_1 s^4)+a_2 se_1 Q=0, \\
& C'_3: e_3=x+\frac{1}{2}a_2 s=e_2(y^2 -a_6 e_1 s^4 )-e_1 x Q=0.
\end{aligned}
\end{equation}

The weights of each of these curves with respect to the divisors $D_i$ $(i=0,1,2,3)$ are
\def\arraystretch{1}
\begin{equation}\nonumber
\begin{matrix}
\\
D_0\\
D_1\\
D_2\\
D_3
\end{matrix}
\begin{pmatrix}
C_0 & C_{2} & C_{1} & C_{3} \\
\hline
2  & 0 & -1 & 0 \\
0 & 2 & -1 & 0 \\
-1 & -1& 2 &  -2\\
0 & 0 & -2 & 4
\end{pmatrix}
\quad
\begin{pmatrix}
C_{23}\\
\hline
0 \\
2\\
-1\\
0\\
\end{pmatrix}
\quad
\begin{pmatrix}
C'_{3}\\
\hline
0 \\
0\\
-1\\
2\\
\end{pmatrix}
\quad 
\begin{pmatrix}
C^{(2)}_{3}\\
\hline
0 \\
0\\
-1\\
2\\
\end{pmatrix}
\quad
\begin{pmatrix}
C^{(1)}_{3}\\
\hline
0 \\
-2\\
0\\
2\\
\end{pmatrix}
\quad 
\begin{pmatrix}
C_{3}^{(1')}\\
\hline
0 \\
1\\
0\\
-1\\
\end{pmatrix}
\end{equation}
\end{subequations}

\subsection{Weights and representations}
Letting $(a,b,c)$ denote a weight expressed in the basis of simple roots and $\boxed{a\   b\    c}$ denote a weight expressed in the basis of fundamental weights, the weights of the representations $\mathbf{7}$ and $\mathbf{8}$ are 
given below. 
\begin{center}
\begin{tabular}{c|c}
Representation $\mathbf{7}$ of B$_3$ & Representation $\mathbf{8}$  of B$_3$ \\
\hline 
{\begin{tabular}{lr}
\boxed{\  \  1\  \  \   0\  \  0}   & (1,1,1)\\ \boxed{-1 \  \   \  1\  \ \,  0}   & (0,1,1) \\ \boxed{\  \   0 -1\  \  2} & (0,0,1)\\ \boxed{\   \ 0\   \  \   0\   \    0} & (0,0,0) \\  \boxed{\  \   0\  \  1- 2}&    (1,-1,0) \\ \boxed{\  \   1    -1\  \  0} & (0,-1,-1) \\ \boxed{-1\  \  \  0\   \   0} & (0,0,-1)
\end{tabular}}
&
{\begin{tabular}{c c}
\boxed{\ \ 0\ \ 0\ \ 1} & (1/2, 1, 3/2) \\
\boxed{\ \ 0\ \ 1-1} &(1/2, 1, 1/2) \\
\boxed{\ \ 1-1\ \ 1}&(1/2, 0, 1/2) \\
\boxed{-1\ \ 0\ \ 1}\boxed{\ \ 1\ \ 0-1} & (-1/2, 0, 1/2)(1/2, 0, -1/2) \\
\boxed{-1\ \ 1-1}& (-1/2, 0, -1/2)\\
\boxed{\ \ 0-1\ \ 1} &(-1/2, -1, -1/2) \\
\boxed{\ \ 0\ \ 0-1}& (-1/2, -1, -3/2)
\end{tabular}}
\end{tabular}
\end{center}

For the first crepant resolution, we obtained the following two weights from the fiber degenerations: 
\begin{equation}\nonumber
\boxed{0\  -1  \ 2}\quad \text{ and}\quad \boxed{1\  0   \ -1}.
\end{equation} 
The weight \boxed{0\  -1  \ 2} is a  weight of the  representation $\mathbf{7}$ (the vector representation of $B_3$). 
The  weight \boxed{1\  0   \ -1}  is weight of the  representation $\mathbf{8}$ (the spin representation of $B_3$). 
The vector representation of $B_3$ is quasi-miniscule while the spin representation is minuscule; 
as their name indicates, they  are of  dimensions seven and eight, respectively.

\begin{center}
\begin{tabular}{|c|c|c|c|c|}
\hline 
& \multicolumn{2}{|c|}{Spin($7$)  with $v_S(a_2)=1$} &\multicolumn{2}{|c|}{ Spin($7$)  with $v_S(a_2)\geq 2$} \\
\hline 
Locus &  $V(a_{4,2})$  & $V(a_{2,1}^2-4a_{4,2})$ &\multicolumn{2}{|c|}{$V(a_{4,2})$}   \\
\hline 
Weights &  $\boxed{0\  1\  -2}$  & $\pm\boxed{1\  0\  -1}$ &\multicolumn{2}{|c|}{$\boxed{0\ 1\  -2}$  and $\boxed{-1\  0\  1}$}   \\
\hline 
Representations &  $\mathbf{7}$  & $\mathbf{8}$  &\multicolumn{2}{|c|}{  $\mathbf{7}$  and $\mathbf{8}$} \\
\hline
\end{tabular}
\end{center}

\section{Spin($8$)-models}\label{Sec:Spin8}

The Weierstrass model of a  Spin($8$)-model is described by Step $6$ of Tate's algorithm with the additional arithmetic condition that $P(T)$ has three distinct $\kappa$-rational solutions, where $\kappa$ is the residue field of the generic point $\eta$ of $S=V(s)$ and $s$ is a section of $\mathscr{O}_B(S)=\mathscr{S}$.
\begin{thm}[Canonical form for Spin($8$)-models] \label{Thm:Spin8Normal}
The Weierstrass model for a \textnormal{Spin($8$)}-model can be written as
\begin{equation}\label{Eq:Spin8Normal}
\mathscr{E}_0:\quad 
zy^2=(x-s x_1 z)(x-s x_2 z)(x+s x_3z)- s^{2+\alpha} Qz, \quad Q= r  x^2 +q s x z-s^2 t z^2,\quad \alpha\in \mathbb{Z}_{\geq 0},
\end{equation}
where $(r,q,t)\neq (0,0,0)$ on the divisor $S$ and the  coefficients $s$, $x_i$, $r$, $q$, and $t$ are sections of the line bundles given below.
$$
\begin{tabular}{|c|c|c|c|c|}
\hline 
$s$& $x_i$ &$ r$ & $q$ & $t$ \\
\hline
$\mathscr{S}$ & $\mathscr{L}^{\otimes 2} \otimes \mathscr{S}^{-1}$ &$ \mathscr{L}^{\otimes 2} \otimes \mathscr{S}^{\otimes(-2-\alpha)} $& $\mathscr{L}^{\otimes 4} \otimes \mathscr{S}^{\otimes(-3-\alpha)}$ & $\mathscr{L}^{\otimes 6} \otimes \mathscr{S}^{\otimes(-2\alpha-4)}$\\
\hline 
\end{tabular}
$$
 $Q$ cannot be identically zero since otherwise, the Mordell-Weil group will contain at least a torsion subgroup  $\mathbb{Z}/2\mathbb{Z}\oplus \mathbb{Z}/2\mathbb{Z}$.
 Moreover,  $d=(x_1-x_2) (x_2-x_3) (x_3-x_1)$ cannot be identically zero on $S$. 
\end{thm}

\begin{proof} 
By Step $6$ of Tate's algorithm, the Weierstrass coefficients have  following valuations along $S$: 
 \begin{equation}\nonumber
 v_{S}(a_1)\geq 1, \quad v_{S}(a_2)\geq 1 , \quad v_{S}(a_3)\geq 2 , \quad v_{S}(a_4)\geq 2, \quad v_{S}(a_6)\geq 3.
 \end{equation}
By definition,  a Spin($8$)-model is such that the cubic polynomial $P(T)=T^3+ a_{2,1} T^2+a_{4,2} T + a_{6,3}$ factorizes in   $\kappa$. That is,
\begin{equation}\nonumber
P(T)=(T-x_1)(T-x_2)(T-x_3), \quad x_i\in \kappa.
\end{equation}
The discriminant of $P(T)$ is  a perfect square $\Delta(P)=d^2$ with 
\begin{equation}\nonumber
d=(x_1-x_2) (x_2-x_3) (x_3-x_1).
\end{equation}
The polynomial $P(T)$  has three distinct roots in $\kappa$ if and only if  $d$ is nonzero modulo $s$. 
Working backwards, from $P(T)$, we can then compute the Weierstrass coefficients $a_2$, $a_4$, and $a_6$ to be \\
$$a_2=-(x_1+x_2+x_3) s + s^{2} r',\quad  a_4= s^2 (x_1 x_2 + x_1 x_3  + x_2 x_3)+ s^3 q',\quad  a_6= s^3 x_1 x_2 x_3 + s^{4} t'.$$
We  complete the square in $T$ and obtain  $a_1=a_3=0$; this modifies  $r'$, $q'$, and $t'$ accordingly. 
We then define $\alpha$ to be  the highest power of $s$ that we can factor out of $r'$, $q'$, and  $t'$, i.e. $r'= s^\alpha r$, $q'= s^\alpha q$, and $t'=t^\alpha q$.  
This explains the  canonical  form of the equation. 
We require $Q$ to be nonzero, as otherwise, the Mordell-Weil group is non-trivial. 
\end{proof}

\subsection{Crepant resolution of singularities}
\begin{thm}[Crepant resolutions for Spin($8$)-models]\label{Thm:Spin8Res}
Assuming that $V(x_i)$ are smooth varieties intersecting two by two transversally,  the following sequence of  blowups defines a   crepant resolution of the normal form of a Spin($8$)-model given by Theorem 
\ref{Thm:Spin8Normal}: 
\begin{equation}\nonumber
   \begin{tikzcd}[column sep=huge] 
  X_0  \arrow[leftarrow]{r} {(x,y,s|e_1)} & \arrow[leftarrow]{r}{(y,e_1|e_2)} X_1 & X_2 \arrow[leftarrow]{r}{(x-x_i s z,e_2|e_3)} & X_3 \arrow[leftarrow]{r}{(x-x_j s z,e_2|e_4)} & X_4,
  \end{tikzcd}
\end{equation}
where $ X_0 =\mathbb{P}_B[\mathscr{O}_B\oplus\mathscr{L}^{\otimes 2}\oplus \mathscr{L}^{\otimes 3}]$.  
 The proper transform of $\mathscr{E}_0$ is
\begin{equation}\nonumber
Y^{(i+1,j+1)}: \quad 
\begin{cases}
e_2\big( y^2  -e_4^\alpha e_3^\alpha e_2^\alpha e_1^{2+\alpha s^{2+\alpha}} Q\big)=e_1 u_i u_j  (x-x_k s)\\   e_3u_i= (x-x_i s)\\
  e_4u_j= (x-x_j s),
  \end{cases}
\end{equation}
where the relative projective coordinates are
\begin{equation}\nonumber
[e_1 e_2  e_3^2e_4 x:e_1 e_2^2 e_3^2 e_4^2 y: z=1][e_3x:e_3e_2e_4y:s][y:e_1][u_i:e_2 e_4] [u_j:e_2].
\end{equation}
\end{thm}

\begin{proof}
The Kodaira fiber I$^{*\text{s}}_0$ over the generic point of $S$ is seen after the first two blowups:
\begin{subequations}
\begin{equation}\nonumber
  \begin{tikzcd}[column sep=huge] 
  X_0 =\mathbb{P}_B[\mathscr{O}_B\oplus\mathscr{L}^{\otimes 2}\oplus \mathscr{L}^{\otimes 3}] \arrow[leftarrow]{r} {(x,y,s|e_1)} & \arrow[leftarrow]{r}{(y,e_1|e_2)} X_1 & X_2 \ .
  \end{tikzcd}
\end{equation}
The proper transform of $\mathscr{E}_0$ is
 \begin{equation}\nonumber
\mathscr{E}_2: \quad
e_2 (y^2- e_2^\alpha e_1^{2+\alpha} s^{2+\alpha} Q )=e_1  (x-x_1s)(x-x_2s)(x-x_3s).
\end{equation}
In $X_2$, the projective coordinates of the fiber of $X_0$ and the successive blowup maps are 
\begin{equation}\nonumber
[e_1 e_2  x:e_1 e_2^2  y: z=1][x:e_2y:s][y:e_1].
\end{equation}
\end{subequations}

After the second blowup, the variety is smooth up to codimension two. The generic fibers of the fibral divisors of this partial resolution are
\begin{equation}\nonumber
\begin{aligned}
C_0:& \quad s=e_2 y^2 -e_1 x^3=0 \\
C_{1}:&\quad  e_1=e_2=0 \\
C_2^{(i)}:& \quad {e_2}=(x-x_is)=0, \quad i=1,2,3.
\end{aligned}
\end{equation}
Their dual graph is the affine Dynkin diagram  $\widetilde{\text{D}}_4$. The component $C_1$ has multiplicity two, where one comes from  the exceptional divisor $e_1=0$, and  the other  from the exceptional divisor $e_2=0$.  The $C_0$ component 
 is the proper transform of the original elliptic fiber and is the only one touching the section $x=z=0$ of the elliptic fibration. 
 The $C_{1}$ component is the central node of the affine Dynkin diagram $\widetilde{\text{D}}_4$s, and $C_2^{(i)}$ are the remaining nodes. 
Over $V(x_i-x_j)$, the components of $C_2^{(i)}$ and $C_2^{(j)}$ coincide. These three subvarieties intersect along the subvariety  $V(x_1-x_2, x_2-x_3)$, over which  the three external nodes $C_2^{(i)}$ coincide. 
There are leftover families of  double point  singularities in codimension three. 
$\mathscr{E}_2$ has terminal singularities in codimension three along the three loci
\begin{align}\nonumber
S^{(3)}_i=V\Big({y^2- e_2^\alpha e_1^{2+\alpha} s^{2+\alpha} Q, e_2, x-x_j s,x_j-x_k }\Big)=V(y^2- e_2^\alpha e_1^{2+\alpha} s^{2+\alpha} Q)\cap C_2^{(j)}\cap C_2^{(k)},
\end{align}
where $(i,j,k)$ is a permutation of $(1,2,3)$. These singularities are located in the patch $ze_1 s\neq 0$. 

These singularities have crepant resolutions obtained by blowing up two of the three Weil divisors $C_2^{(i)}$. Thus, there are six possible choices. 
 Since the singularities are located in the patch $zs e_1\neq 0$, we note  that $\mathscr{E}_2$  has the structure of the binomial variety 
\begin{equation}\nonumber
V(v_1 v_2-w_1 w_2 w_3), 
\end{equation}
which was studied in \cite{EY}. This binomial variety  has six small resolutions whose flop diagram is a hexagon (a Dynkin diagram of type $\widetilde{\text{A}}_5$) \cite{EY}. 
\end{proof}

The following is a proof that we have a resolution by inspecting the singularities chart by chart. 
\begin{proof}[Proof in charts] 
\begin{equation}\nonumber
F=y^2 -(x-x_1 s)(x-x_2s)(x-x_3s) + s^{2+\alpha} (p x^2 + q s x + s^2 t).
\end{equation}
We assume that $V(x_1)$, $V(x_2)$, and $V(x_3)$ are smooth varieties intersecting two by two transversally. 
The idea of the proof is the following. Working in charts. 
The first blowup has center $(x,y,s)$ and requires three charts. If we call the exceptional divisor $E_1=V(e_1)$, the second blowup is centered at $V(y,e_1)$ and requires two charts.
We will show that after two blowups, the proper transform of $F$ describes a smooth variety or a binomial variety of the type 
$V(u_1 u_2-w_1 w_2 w_3)$, which will require two more blowups that can be done in six different ways.  

\begin{enumerate}
\item  $(x,y,s)\to (xy, y,sy)$
\begin{equation}\nonumber
F_{(1)}=1 -y (x-x_1 s)(x-x_2s)(x-x_3s) + s^{2+\alpha}y^{2+\alpha} (p x^2 + q s x + s^2 t).
\end{equation}
which is smooth since the system of equations $\partial_x F=\partial_y F=\partial_s F=F=0$ has no solutions. 

\item  $(x,y,s)\to (x, yx,sx)$
\begin{equation}\nonumber
F_{(2)}=y^2 -x(1-x_1 s)(1-x_2s)(1-x_3s) + s^{2+\alpha} x^{2+\alpha} (p  + q s  + s^2 t).
\end{equation}
The exceptional divisor is $V(x)$. Hence the second blowup is centered at $(x,y)$ and requires two charts. 
\begin{enumerate}
\item $(x,y)\to (x, yx)$
\begin{equation}\nonumber
F_{(2,1)}=x(y^2 +s^{2+\alpha} x^{\alpha} (p  + q s  + s^2 t))-(1-x_1 s)(1-x_2s)(1-x_3s).
\end{equation}
This has the singularities of the binomial variety $V(u_1 u_2-w_1 w_2 w_3)$.
\item $(x,y)\to (xy, y)$
\begin{equation}\nonumber
F_{(2,2)}=y\Big(1+ s^{2+\alpha} x^{2+\alpha}y^{\alpha} (p  + q s  + s^2 t)\Big)-x(1-x_1 s)(1-x_2s)(1-x_3s). 
\end{equation}
When $\alpha>0$, there are no singularities left. However, when $\alpha=0$, we still have double point singularities in the patch $x\neq 0$, and $F_{(2,2)}$ can be replaced by $V(u_1 u_2- w_1 w_2 w_3)$. 
\end{enumerate}
\item $(x,y,s)\to (xs, ys,s)$
\begin{equation}\nonumber
F_{(3)}=y^2 -s(x-x_1 )(x-x_2)(x-x_3) + s^{2+\alpha} (p x^2 + q  x + t).
\end{equation}
The exceptional divisor is $V(s)$. Hence, the second blowup is centered at $(y,s)$ in this chart.
We then blowup $(y,s)$, which requires two charts. 
\begin{enumerate}
\item $(y,s)\to (y, sy)$.
 \begin{equation}\nonumber
	F_{(3,1)}= y(1+ s^{2+\alpha} y^{\alpha}(p x^2 + q  x + t))-(x-x_1 )(x-x_2)(x-x_3).
\end{equation}
$F_{(3,1)}$ has the singularities of the binomial variety $V(u_1 u_2- w_1 w_2 w_3)$.
\item $(y,s)\to (ys, s)$.
 \begin{equation}\nonumber
F_{(3,2)}=s (y^2 + s^{\alpha} (p x^2 + q  x + t)) -(x-x_1 )(x-x_2)(x-x_3),
\end{equation}
which is again of the binomial variety $V(u_1 u_2- w_1 w_2 w_3)$.
\end{enumerate}
\end{enumerate}
After two blowups, if there are singularities left, they are those of the binomial variety
\begin{equation}\nonumber
V(u_1 u_2- w_1 w_2 w_3),
\end{equation}
whose toric description is a triangular prism. A crepant resolution of this binomial variety is given by  a sequence of two blowups corresponding to the subdivision of the triangular prism into two tetrahedrons \cite{EY}. 
Blowup of $(u_1,w_i)$ with $(u_1,w_i)\to (u_1 w_i,w_i)$ gives
\begin{equation}\nonumber
u_1 u_2  - w_j w_k=0.
\end{equation}
The other patch  $(u_1,w_i)\to (u_1,w_i u_1)$ is trivially smooth. 
Likely, blowing up $(u_1,w_j)$ with $(u_1,w_j)\to (u_1w_j,w_j)$ gives
\begin{equation}\nonumber
u_1 u_2 - w_k=0,
\end{equation}
which is smooth. The other patch is also trivially smooth. This resolution is $Y^{i+1,j+1}$. 
The graph of their flops is an affine Dynkin diagram $\widetilde{A}_5$ (a hexagon). 
\end{proof}

\subsection{Fiber structure and degenerations}

The fibral divisors of  the elliptic fibration $Y^{(i+1,j+1)}$  for $\alpha=0$ are 
\begin{equation}
\begin{cases}
D_0\quad C_0 : \quad s=ze_2 y^2 -e_1 x^3=0 \\
D_1\quad C_{1}:\quad  e_1=e_2=e_3u_1- (x-x_1 s)=e_4u_2- (x-x_2 s)=0 \\
D_{i+1}\quad C_2^{(i)}: \quad e_3=x-x_i s=e_4 u_j-(x_i-x_j) s= ze_2(y^2 -e_1^2 s^2 Q)-e_1 u_i u_j (x-x_k s)=0 \\
D_{j+1}\quad C_2^{(j)}: \quad e_4=x-x_js=e_3 u_i -(x_j-x_i) s=ze_2(y^2 -e_1^2 s^2 Q)-e_1 u_i u_j (x-x_k s)=0 \\
D_{k+1}\quad C_2^{(k)} \quad e_2=x-x_ks=e_3u_i-(x_k-x_i) s= e_4u_2- (x_k-x_j)s=0.
\label{spin8.div}
\end{cases}
\end{equation}
 If $\alpha>0$,  $Q$ is replaced by $e_4^\alpha e_3^\alpha e_2^\alpha e_1^{2+\alpha} s^{2+\alpha} Q$, which is zero for all the fibral divisors $D_{a}$. Note that this will carry through equations \eqref{spin8.div}, \eqref{Eq:lin}, \eqref{spin8.codim3} and \eqref{spin8.codim3.div}. Even though the generic fiber over $S$ has geometric components, this is not necessarily carried over to its degenerations.

 When $\alpha=0$,  the generic point of $V(x_i-x_j)\cap S$ (see equation \eqref{Eq:lin})  contains the irreducible component  
$C_2^{(i')}$, which is not geometrically irreducible but splits into  two geometrically irreducible curves after a quadratic field extension.  
 The fiber is of type I$^{*\text{ns}}_1$ with its dual graph of type $\widetilde{B}_4^t$.

Over codimension-two points (the three irreducible components of $V(d)$), we have 
\begin{equation}\label{Eq:lin}
\begin{array}{ccc}
 V(x_i-x_j)\cap S & V(x_i-x_k) \cap S& V(x_j-x_k)\cap S\\
 \begin{cases}
C_2^{(i)}\longrightarrow C_{2}^{(i,j)}+C^{(i')}_{2} \\ C_2^{(j)} \longrightarrow C_{2}^{(i,j)}
\end{cases}
& \begin{cases}
C_2^{(i)}\longrightarrow C_{2}^{(i,k)}+C^{(i'')}_{2} \\ C_2^{(k)} \longrightarrow C_{2}^{(i,k)}
\end{cases}
& \begin{cases}
C_2^{(j)}\longrightarrow C_{2}^{(j,k)}+C^{(j')}_{2} \\ C_2^{(k)} \longrightarrow C_{2}^{(j,k)}
\end{cases}.
\end{array}
\end{equation}

Over codimension-three points ( the common intersection of the three components of $V(d)$), we have 
\begin{equation}\label{spin8.codim3}
V(x_1-x_2, x_2-x_3)\cap S \begin{cases}
C_2^{(i)}\longrightarrow C_2^{(1,2,3)}+C_2^{(i,j)'}+C_2^{(i''')}\\
C_2^{(j)}\longrightarrow C_2^{(1,2,3)}+C_2^{(i,j)'}\\
C_2^{(k)}\longrightarrow C_2^{(1,2,3)}
\end{cases},
\end{equation}
with the components of the fiber defined as
\begin{equation}\label{spin8.codim3.div}
\begin{aligned}
\begin{cases}
C_2^{(i')}:& \quad e_3=x_i-x_j=x-x_i s= u_j= y^2 -e_1^2 s^2 Q=0 \\
C_2^{(i,j)}:& \quad e_3=e_4=x_i-x_j=x-x_i s= e_2(y^2 -e_1^2 s^2 Q)-e_1 u_i u_j (x_i-x_k)s=0 \\
C_2^{(i,j)'}:& \quad e_3=e_4=x_i-x_j=x_j-x_k=x-x_i s= y^2 -e_1^2 s^2 Q=0 \\
C_2^{(i'')}:& \quad e_3=x_1-x_k=x-x_1 s=e_4 u_j -(x_i-x_j) s= y^2 -e_1^2 s^2 Q=0 \\
C_2^{(i''')}:& \quad e_3=x_1-x_k=x_i-x_j=x-x_i s=u_2= y^2 -e_1^2 s^2 Q=0 \\
C_2^{(i,k)}:& \quad e_3=e_2=x_i-x_j=x-x_i s =
 e_4u_2- (x_1-x_2)s=0 \\
 C_2^{(j,k)}:& \quad e_4=e_2=x_j-x_k=x-x_js=e_3 u_i -(x_j-x_i) s=0 \\
 C_2^{(j')}:& \quad e_4=x_j-x_k=x-x_i s= y^2 -e_1^2 s^2 Q=0 \\
 C_2^{(1,2,3)} : & \quad e_2=e_3=e_4=x-x_k s  =x_i-x_j=x_j-x_k=0.
\end{cases}
\end{aligned}
\end{equation}

 Over $V(x_i-x_j)\cap S$, we have a fiber of type 
I$_1^{*\text{ns}}$ with dual graph of type $\widetilde{\text{B}}^t_4$. The non-geometrically irreducible node is 
$C_2^{(i')}$  whereas the  geometric fiber is a full $\widetilde{\text{D}}_5$. 
Over $V(x_1-x_2,x_2-x_3)\cap S$, the fiber is of type IV$^{*\text{ns}}$ with dual graph of type $\widetilde{\text{F}}^t_4$ and geometric dual graph of type 
 $\widetilde{\text{E}}_6$.
When $Q(x_i,s,z)=0$, the fibers  $I_1^{*\text{ns}}$ and $IV^{*\text{ns}}$ degenerate further along the codimension three locus  $V(s,x_i-x_j, x_i^2+ x_i r + t)$ in the base $B$ when the curve is $C_2^{(i')}$. 
The  degenerations are illustrated in Figure \ref{Fig:Spin8AlphaZero} and Figure  \ref{Fig.Spin8AlphaPos}, respectively, for $\alpha=0$ and $\alpha>0$.

\subsection{Flops and representations}
For $Y^{2,3}$, the curves obtained by analyzing the fiber structure have the following geometric weights.
\begin{equation}\label{Eq:Weights}
\begin{matrix}
\\
\alpha_0\\
\alpha_2\\
\alpha_1\\
\alpha_3\\
\alpha_4
\end{matrix}
\  \  
\begin{matrix}
\\
D_0\\
D_1\\
D_2\\
D_3\\
D_4
\end{matrix}
\begin{pmatrix}
C_0 & C_{1} & C_{2}^{(1)} & C_2^{(2)} & C_2^{(3)} \\
\hline
2  & -1 &  0 &  0 & 0 \\
-1& 2 &  -1 &  -1 & -1 \\
0  & -1 &  2 &  0 & 0 \\
0  & -1 &  0 &  2 & 0 \\
0  & -1 &  0 &  0 & 2
\end{pmatrix}
\quad
\begin{pmatrix}
 C_{2}^{(1,2)} & C_2^{(1')} & C_2^{(1,3)} &  C_2^{(1'')} &  C_2^{(2,3)}  &  C_2^{(2')} \\
\hline
  0 &  0 & 0 & 0&  0&0 \\
 -1 &  0 & -1 & 0& -1&0\\
  0 &  2 & 0 & 2& 0&0\\
  2 &  -2 & 0 & 0&0& 2\\
  0 &  0 & 2 & -2&2&-2
\end{pmatrix}.
\end{equation}
To express the intersection numbers with the fibral divisors, we introduce the  convention 
$$
\epsilon_0=(1,0,0,0,0), \quad \epsilon_1= (0, 1,0, 0,0),\quad  \epsilon_2= (0, 0, 1, 0,0), \quad \epsilon_3= (0, 0, 0,1,0), \quad \epsilon_4= (0, 0, 0,0,1).
$$
Since the central node of the D$_4$ diagram corresponds to the node $C_{1}$ of the resolution, in order to match the convention we use for weights, the intersection numbers  $(w_0, w_1, w_2,w_3)$ correspond to the weight $[w_2, w_1, w_3]$ of D$_4$: 
\begin{equation}\label{Eq:RulesD4}
(w_0, w_1, w_2,w_3, w_4)=\sum_{a=0}^4 w_a \epsilon_a \to \boxed{ w_2\  w_1\   w_3 \ w_4}.
\end{equation}
Since weights with their appropriate multiplicities sum to zero, i.e. $w_0+2w_1+w_2+w_3+w_4=0$, we have a bijection with the  inverse map
\begin{equation}\nonumber
 \boxed{ \varpi_1\  \varpi_2\   \varpi_3 \ \varpi_4}\to (-2\varpi_2-\varpi_1-\varpi_3-\varpi_4, \varpi_2, \varpi_1,\varpi_3,\varpi_4).
\end{equation}
For the resolution $Y^{(i+1,j+1)}$, by a direct generalization from equation \eqref{Eq:Weights}, we have the geometric weights 
$$
\begin{aligned}
& C_2^{(i,j)}\to -\epsilon_1+2\epsilon_{j+1}, \qquad C_2^{(i')}\to 2\epsilon_{i+1}-2\epsilon_{j+1}, \quad  C_2^{(i,k)}\to -\epsilon_1+2 \epsilon_{k+1},\\
& C_2^{(i'')}\to 2 \epsilon_{i+1}-2\epsilon_{k+1}, \quad C_2^{(j,k)}\to-\epsilon_1+2\epsilon_{k+1}, \qquad C_2^{(j')}\to 2 \epsilon_{j+1}-2\epsilon_{k+1} .
\end{aligned}
$$
  When $\alpha=0$,  these curves are not geometrically irreducible, as each curve splits into two irreducible curves, each having the same intersection numbers as  the fibral divisors corresponding to half of those of $ C_2^{(i')}$, $C_2^{(i'')}$, and $C_2^{(j')}$.  When $\alpha>0$, the curves  $ C_2^{(i')}$, $C_2^{(i'')}$, and $C_2^{(j')}$ are double curves and the intersection numbers of the corresponding reduced curves are also half of those of 
$ C_2^{(i')}$, $C_2^{(i'')}$, and $C_2^{(j')}$. Hence, for any  $\alpha$,  we end up with the following intersection numbers: 
\begin{align}\nonumber
 \epsilon_{i+1}-\epsilon_{j+1}, \quad  \epsilon_{i+1}-\epsilon_{k+1}, \quad  \epsilon_{j+1}-\epsilon_{k+1}.
\end{align}
These are, up to a sign, permutations of 
\begin{align}\nonumber
(0,0,0,1,-1) \quad  (0,0,1,0,-1), \quad (0,0,1,-1,0).
\end{align}
 Following the dictionary given by equation \eqref{Eq:RulesD4}, we get
 \begin{align}\nonumber
 \boxed{0\ 0\ 1\  -1},\quad \boxed{1\ 0 \  0\  -1}\ ,\quad  \boxed{1\ 0\  -1\  0}\ , 
\end{align}
which are the weights of  the minuscule  representations $\mathbf{8}_v$, $\mathbf{8}_c$, and $\mathbf{8}_s$, respectively. 
Hence, each resolution gives the representation $\mathbf{8}_v\oplus\mathbf{8}_c\oplus\mathbf{8}_s$.

The hexagon of crepant resolutions is isomorphic to the chamber structure 
of the hyperplane arrangement
\begin{equation}\nonumber
 \mathrm{I}(D_4, \mathbf{8}_v\oplus \mathbf{8}_c \oplus \mathbf{8}_s),
 \end{equation}
 where $\mathbf{8}_v$ is the vector representation, and $\mathbf{8}_c$ and $\mathbf{8}_s$ are the two irreducible spinor representations. Each of these three irreducible representations is minuscule of dimension eight, and
 their highest weights are respectively  $\boxed{1\ 0\ 0\  0}$, $\boxed{0\  0  \ 1  \  0}$, and $\boxed{0\   0\  0 \  1}$. Note that these  weights are related by simple involutions: $\mathbf{8}_v\leftrightarrow\mathbf{8}_c$ by the involution $(\varpi_1\leftrightarrow\varpi_3)$, and  $\mathbf{8}_v\leftrightarrow\mathbf{8}_c$ by the involution $(\varpi_3\leftrightarrow\varpi_4)$.

The weights of the representations $\mathbf{8}_v$, $\mathbf{8}_c$ and $\mathbf{8}_s$ are given below with the following conventions:  $(a,b,c,d)$ is a weight of D$_4$ expressed in the basis of simple roots while $\boxed{a\  b \ c \ d}$ is a weight of D$_4$ written in the basis of fundamental  weights. 

\begin{center}
\begin{tabular}{c|c|c}
Weight system of $\mathbf{8}_v$ of D$_4$&Weight system of $\mathbf{8}_c$ of D$_4$& Weight system of $\mathbf{8}_s$ of D$_4$\\
\hline
{\begin{tabular}{c}
\boxed{\   \  1\  \   0  \  \   0  \   \  0}\\
\boxed{-1\  \    1\     \  0\      \  0}\\
\boxed{\    \    0   -1  \  \  1 \   \  1}\\
\boxed{0\ 0\  -1\  1}  \boxed{0\ 0\  1  -1}  \\
\boxed{\    \    0   \   \   \   1   -1   - 1}\\
\boxed{ \  \    1     -1\     \  0\   \  0}\\
\boxed{ - 1\   \   \   0  \     \   0  \    \  0}
\end{tabular}}
&
{
\begin{tabular}{c}
\boxed{\   \  0\  \   0  \  \   1  \   \  0}\\
\boxed{\  \   0\  \    1-1  \      \  0}\\
\boxed{\    \    1   -1  \  \  0 \   \  1}\\
\boxed{-1\   0\   0 \  1}  \boxed{1\ 0\  0  -1}  \\
\boxed{ -1   \   \   \   1   \   \  0 - 1}\\
\boxed{ \  \    0     -1\     \  1\   \  0}\\
\boxed{ \    \    0\   \   \   0    -1  \    \  0}
\end{tabular}
}
&
{
\begin{tabular}{c}
\boxed{\   \  0\  \   0  \  \   0  \   \  1}\\
\boxed{\  \   0\  \    1\   \   0  -1}\\
\boxed{\    \    1   -1  \  \  1 \   \  0}\\
\boxed{-1\   0\   1 \  0}  \boxed{1\ 0 -1\   0}  \\
\boxed{ -1   \   \   \   1   -1 \   \  0 }\\
\boxed{ \  \    0     -1\     \  0\   \  1}\\
\boxed{ \    \    0\   \   \   0   \     \     0 -1}
\end{tabular}
}\\
& & \\
$
\begin{array}{c}
(1, 1, \frac{1}{2}, \frac{1}{2})\\
(0, 1, \frac{1}{2}, \frac{1}{2})\\
(0, 0, \frac{1}{2}, \frac{1}{2})\\
(0, 0, -\frac{1}{2}, \frac{1}{2})(0, 0, \frac{1}{2}, -\frac{1}{2})\\
(0, 0, -\frac{1}{2}, -\frac{1}{2})\\
(0, -1, -\frac{1}{2}, -\frac{1}{2})\\
(-1, -1, -\frac{1}{2}, -\frac{1}{2})
\end{array}
$
&
$
\begin{array}{c}
(\frac{1}{2}, 1, 1, \frac{1}{2})\\
(\frac{1}{2}, 1, 0, \frac{1}{2})\\
(\frac{1}{2}, 0, 0, \frac{1}{2})\\
(-\frac{1}{2}, 0, 0, \frac{1}{2})(\frac{1}{2}, 0, 0, -\frac{1}{2})\\
(-\frac{1}{2}, 0, 0, -\frac{1}{2})\\
(-\frac{1}{2}, -1, 0, -\frac{1}{2})\\
(-\frac{1}{2}, -1, -1, -\frac{1}{2})
\end{array}
$&
$
\begin{array}{c}
(\frac{1}{2}, 1, \frac{1}{2}, 1)\\
(\frac{1}{2}, 1, \frac{1}{2}, 0)\\
(\frac{1}{2}, 0, \frac{1}{2}, 0)\\
(-\frac{1}{2}, 0, \frac{1}{2}, 0)(\frac{1}{2}, 0, -\frac{1}{2}, 0)\\
(-\frac{1}{2}, 0, -\frac{1}{2}, 0)\\
(-\frac{1}{2}, -1, -\frac{1}{2}, 0)\\
(-\frac{1}{2}, -1, -\frac{1}{2}, -1)
\end{array}
$
\end{tabular}
\end{center}

\section{Application to five-dimensional supergravity theories }\label{Sec:Phys}
In this section, we use the information we gathered from the geometry of G$_2$, Spin($7$), and Spin($8$)-models to explore the corresponding gauge theories in M-theory and F-theory compactifications. 

We first consider  M-theory compactified on a Calabi--Yau threefold $Y$ elliptically fibered over a smooth rational surface $B$ of canonical class $K$. 
The divisor $S$ over which the generic fiber is of Kodaira type I$_0^*$ is now a smooth curve of genus $g$ and self-intersection $S^2$. 

The compactification of M-theory on a Calabi--Yau threefold $Y$ yields a  five-dimensional supergravity theory with eight supercharges coupled to  $h^{1,1}(Y)$  vector multiplets and $h^{2,1}(Y)+1$ neutral hypermultiplets \cite{Cadavid:1995bk}. 
The gravitation multiplet also contains a gauge field called the graviphoton. The dynamics of the vector multiplets and the graviphoton are derived from a real function called the prepotential.  
After integrating out massive charged vector and matter fields, the prepotential receives  a one-loop quantum correction protected from additional quantum corrections by supersymmetry. The vector multiplets transform in the adjoint representation of the gauge group while the  hypermultiplets transform in  representation $\mathbf{R}=\bigoplus_i \mathbf{R}_i$ of the gauge group, where $\mathbf{R}_i$ are irreducible representations. 

\begin{table}[htb]
\begin{center}
\begin{tabular}{|c|c|}
\hline
Multiplet & Fields \\
\hline 
Graviton &  $(g_{\mu\nu}, A_{\mu},\psi_\mu)$\\\hline
Vector  &  $(A_\mu,\varphi, \lambda)$\\\hline
Hyper &  $(q, \zeta)$  \\\hline
\end{tabular}
\end{center}
\caption{Matter content for five-dimensional $\mathcal{N}=1$ supergravity theories. The indices $\mu$ and $\nu$ refer to the five-dimensional spacetime coordinates. The tensor $g_{\mu\nu}$ is the metric of the five-dimensional spacetime.  The fields $\psi_\mu$, $\lambda$, $\zeta$  are symplectic Majorana spinors. The field $\psi_\mu$ is the gravitino and $A_\mu$ is the graviphoton.  The hyperscalar $q$ is a quaternion composed of four real scalar fields.  \label{Table:5DMatter}}
\end{table}

The Coulomb branches of the theory correspond to the chambers of the hyperplane arrangement I($\mathfrak{g},\mathbf{R}$). 
By matching the crepant resolutions with the chambers of I($\mathfrak{g},\mathbf{R}$), we determine which resolutions correspond to which phases of the Coulomb branch. 
 The triple intersection numbers of the fibral divisors correspond to the coefficient of the Chern--Simons couplings of the five-dimensional gauge theory and can be compared with the  Intrilligator--Morrison--Seiberg (IMS) prepotential, which is the  one-loop quantum contribution to the prepotential of the five-dimensional gauge theory. Since in field theory the Chern--Simons couplings are linear in the numbers $n_{\mathbf{R}_i}$ of hypermultiplets transforming in the  irreducible representation $\mathbf{R}_i$  such that $\mathbf{R}=\bigoplus_i \mathbf{R}_i$, computing the triple intersection numbers provides a way to determine the numbers  $n_{\mathbf{R}_i}$ from the topology of the elliptic fibration. 
We observe by direct computation in each chamber that the  numbers we find do not depend on the choice of the crepant resolution. 
This  idea of using the triple intersection numbers to determine the number of multiplets transforming in a given representation was used previously in \cite{ES} for SU($n)$-models  and most recently for F$_4$-models  in \cite{F4}. 
This technique has been advocated by Grimm and Hayashi in \cite{Grimm:2011fx}. 
The number of  representations we find using this procedure satisfy  the anomaly cancellation equations of a six-dimensional gauge theory with eight supercharges and the same matter content. 
 One can also determine them geometrically using either Witten's quantization formula for the G$_2$ and Spin($7$)-models, or  the usual intersecting brane methods for the Spin($8$)-model.

\subsection{Coulomb branches of G$_2$, Spin($7$), and Spin($8$) of 5d ${\cal N}=1$ gauge theories}
\label{Sec:IMS}

The Intrilligator--Morrison--Seiberg (IMS)  prepotential is the one-loop quantum contribution to the prepotential of a five-dimensional gauge theory with the matter fields in the representations $\mathbf{R}_i$ of the gauge group. Let   $\phi$ denote an element of the  Cartan subalgebra of the Lie algebra $\mathfrak{g}$, $\alpha$ the fundamental roots, $\varpi$ the weights of $\mathbf R_i$, and $\langle \varpi,\phi \rangle$ the evaluation of a  weight $\varpi$ on an element $\phi$ of the Cartan subalgebra. 
 The Intrilligator--Morrison--Seiberg (IMS) prepotential is \cite{IMS}
\begin{align}
6\mathscr{F}_{\text{IMS}} =&\frac{1}{2} \left(
\sum_{\alpha} |\langle \alpha, \phi \rangle|^3-\sum_{i} \sum_{\varpi\in \mathbf{R}_i} n_{\mathbf{R}_i} |\langle \varpi, \phi\rangle|^3 
\right).\nonumber
\end{align}
The full prepotential also contains a contribution proportional to the third Casimir invariant of the Lie algebra $\mathfrak{g}$; for simple groups, it is only nonzero for SU($N$) groups with $N\geq 3$. 

 For a given choice of a Lie algebra $\mathfrak{g}$, choosing a  dual  fundamental  Weyl chamber {resolves the  absolute values in the sum over the roots. 
 We then consider the arrangement of hyperplanes $\langle \varpi, \phi\rangle=0$, where $\varpi$ runs through all the weights of all the representations $\mathbf{R}_i$. 
 They define the hyperplane arrangement I($\mathfrak{g},\mathbf{R}=\bigoplus_i \mathbf{R}_i)$ restricted to the dual fundamental Weyl chamber. 
 If none of these hyperplanes intersect the interior of the dual fundamental Weyl chamber, we can safely remove the absolute values in the sum over the weights. 
 Otherwise, we have hyperplanes partitioning the fundamental Weyl chamber into subchambers. Each of these subchambers is defined by the signs of the linear forms $\langle \varpi, \phi\rangle$. 
 Two such subchambers are adjacent when they differ by the sign of a unique linear form.

Each of the subchambers is called a Coulomb phase of the gauge theory. 
 The transition from one chamber to an adjacent chamber is a phase transition that geometrically corresponds to a flop between different crepant resolutions of the same singular Weierstrass model.  
   The number of chambers of such a hyperplane arrangement is physically the number of phases of the Coulomb branch of the gauge theory.

For G$_2$ with the adjoint representation $\mathbf{14}$ and the fundamental representation $\mathbf{7}$, the one-loop prepotential is 
\begin{equation}\label{Eq:IMSG2}
6\mathscr{F}^{\text{IMS}}_{\text{G$_2$}}=-8 \phi _1^3 (n_{\mathbf{14}}+n_{\mathbf{7}}-1)+9 \phi _2 \phi _1^2 (-2 n_{\mathbf{14}}+n_{\mathbf{7}}+2)+3 \phi _2^2 \phi _1 (8n_{\mathbf{14}}-n_{\mathbf{7}}-8)-8 (n_{\mathbf{14}}-1) \phi _2^3.
\end{equation}
For Spin($7$) with the adjoint representation $\mathbf{21}$, the vector representation $\mathbf{7}$, and the spin representation $\mathbf{8}$, the prepotential depends on the choice of $sign(\phi_1-\phi_3)=\pm$, and is given by 
\begin{align}\label{Eq:IMSSpin7}
6\mathscr{F}^{\text{IMS}\pm}_{\text{Spin($7$)}} =&-\left(n_{\mathbf{8}}\pm n_{\bf 8}+8 n_{\mathbf{21}}-8\right) \phi _1^3-\left(8 n_{\mathbf{7}}+n_{\mathbf{8}}\mp n_{\mathbf 8}+8 n_{\mathbf{21}}-8 \right) \phi _3^3
 -3n_{\mathbf{8}}(1\mp 1) \phi _1^2 \phi _3-3n_{\mathbf{8}}(1\pm 1) \phi _1 \phi _3^2 \nonumber\\
& +3 \left(-n_{\mathbf{7}}+n_{\mathbf{8}}+n_{\mathbf{21}}-1\right) \phi _1^2 \phi _2 +3\left(n_{\mathbf{7}}-n_{\mathbf{8}}+n_{\mathbf{21}}-1\right) \phi _1\phi _2^2+6n_{\mathbf{8}} \phi _1 \phi _2\phi _3 \\
& -8 \left(n_{\mathbf{21}}-1\right) \phi _2^3+12 \left(n_{\mathbf{7}}-n_{\mathbf{21}}+1\right) \phi _2 \phi _3^2 -6 \left(n_{\mathbf{7}}-3 n_{\mathbf{21}}+3\right) \phi _2^2 \phi _3.\nonumber
\end{align}

Finally, for Spin($8$) with the adjoint representation $\mathbf{28}$, the vector representation $\mathbf{8}_v$,  and the two spin representations $\mathbf{8}_s$ and $\mathbf{8}_c$, we have six chambers. 
Each chamber is uniquely defined by the ordering of $(\phi_1, \phi_3,\phi_4)$.  The first chamber is defined by $\phi_1>\phi_3>\phi_4$, and the prepotential is 
\begin{align}\label{Eq:IMSSpin8Ch1}
\begin{split}
6\mathscr{F}^{\text{IMS}(2,3)}_{\text{Spin($8$)}}= & \ 2(4 (1-n_{\mathbf{28}})-n_{\mathbf{8_c}}- n_{\mathbf{8_s}})\phi _1^3 +8 (1-n_{\mathbf{28}}) \phi _2^3 +2 (4 (1-n_{\mathbf{28}})-n_{\mathbf{8_v}})\phi _3^3 +8 (1-n_{\mathbf{28}}) \phi _4^3 \\
& +3 \phi _2 \left( (n_{\mathbf{8_c}}+n_{\mathbf{8_s}}-n_{\mathbf{8_v}})\phi _1^2 +(-n_{\mathbf{8_c}}+n_{\mathbf{8_s}}+n_{\mathbf{8_v}})\phi _3^2 +(n_{\mathbf{8_c}}-n_{\mathbf{8_s}}+n_{\mathbf{8_v}})\phi _4^2 \right) \\
& +6 \phi _2 \left( n_{\mathbf{8_s}} \phi _1 \phi _3+n_{\mathbf{8_c}} \phi _1 \phi _4+n_{\mathbf{8_v}} \phi _3 \phi _4 \right) -6 \left( n_{\mathbf{8_s}} \phi _1 \phi _3^2+n_{\mathbf{8_c}} \phi _1 \phi _4^2+n_{\mathbf{8_v}} \phi _3 \phi _4^2 \right) \\
& +3 \phi _2^2 \left( (-2 (1-n_{\mathbf{28}})-n_{\mathbf{8_c}}-n_{\mathbf{8_s}}+n_{\mathbf{8_v}}) \phi _1  + (-2 (1-n_{\mathbf{28}})+n_{\mathbf{8_c}}-n_{\mathbf{8_s}}-n_{\mathbf{8_v}})\phi _3 \right. \\
&\left.+ (-2 (1-n_{\mathbf{28}})-n_{\mathbf{8_c}}+n_{\mathbf{8_s}}-n_{\mathbf{8_v}})\phi _4 \right) .
\end{split}
\end{align}
The second chamber is defined by $\phi_3>\phi_1>\phi_4$, and the prepotential is 
\begin{align}\label{Eq:IMSSpin8Ch2}
\begin{split}
6\mathscr{F}^{\text{IMS}(3,2)}_{\text{Spin($8$)}}= & \ 2 (4 (1-n_{\mathbf{28}}) -n_{\mathbf{8_c}})\phi _1^3 +8 (1-n_{\mathbf{28}}) \phi _2^3 +2(4 (1-n_{\mathbf{28}})-n_{\mathbf{8_v}}-n_{\mathbf{8_s}})\phi _3^3 ++8 (1-n_{\mathbf{28}}) \phi _4^3 \\
& +3 \phi _2 \left( (-n_{\mathbf{8_v}}+n_{\mathbf{8_s}}+n_{\mathbf{8_c}})\phi _1^2 +(n_{\mathbf{8_v}}+n_{\mathbf{8_s}}-n_{\mathbf{8_c}})\phi _3^2 +(n_{\mathbf{8_v}}-n_{\mathbf{8_s}}+n_{\mathbf{8_c}})\phi _4^2 \right) \\
& +6 \phi _2 \left( n_{\mathbf{8_s}} \phi _1 \phi _3+n_{\mathbf{8_c}} \phi _1 \phi _4+n_{\mathbf{8_v}} \phi _3 \phi _4 \right) -6 \left( n_{\mathbf{8_s}} \phi _1^2 \phi _3+n_{\mathbf{8_c}} \phi _1 \phi _4^2+n_{\mathbf{8_v}} \phi _3 \phi _4^2 \right) \\
& +3 \phi _2^2 \left( (-2 (1-n_{\mathbf{28}})+n_{\mathbf{8_v}}-n_{\mathbf{8_s}}-n_{\mathbf{8_c}})\phi _1 +(-2 (1-n_{\mathbf{28}})-n_{\mathbf{8_v}}-n_{\mathbf{8_s}}+n_{\mathbf{8_c}}) \phi _3 \right. \\
&\left.+ (-2 (1-n_{\mathbf{28}})-n_{\mathbf{8_v}}+n_{\mathbf{8_s}}-n_{\mathbf{8_c}})\phi _4 \right) .
\end{split}
\end{align}
The third chamber is defined by $\phi_3>\phi_4>\phi_1$, and the prepotential is 
\begin{align}\label{Eq:IMSSpin8Ch3}
\begin{split}
6\mathscr{F}^{\text{IMS}(3,4)}_{\text{Spin($8$)}}= & \ 8 (1-n_{\mathbf{28}}) \phi _1^3 +8 (1-n_{\mathbf{28}}) \phi _2^3 +2(4 (1-n_{\mathbf{28}})-n_{\mathbf{8_s}}-n_{\mathbf{8_v}})\phi _3^3 +2 (4 (1-n_{\mathbf{28}})-n_{\mathbf{8_c}})\phi _4^3  \\
& +3 \phi _2 \left( (n_{\mathbf{8_s}}-n_{\mathbf{8_v}}+n_{\mathbf{8_c}})\phi _1^2 + (n_{\mathbf{8_s}}+n_{\mathbf{8_v}}-n_{\mathbf{8_c}})\phi _3^2 +(-n_{\mathbf{8_s}}+n_{\mathbf{8_v}}+n_{\mathbf{8_c}})\phi _4^2 \right) \\
& +6 \phi _2 \left( n_{\mathbf{8_s}} \phi _1 \phi _3+n_{\mathbf{8_c}} \phi _1 \phi _4 +n_{\mathbf{8_v}} \phi _3 \phi _4 \right) -6 \left( n_{\mathbf{8_s}} \phi _1^2 \phi _3+n_{\mathbf{8_c}} \phi _1^2 \phi _4+n_{\mathbf{8_v}} \phi _3 \phi _4^2 \right) \\
& +3 \phi _2^2 \left( (-2 (1-n_{\mathbf{28}})-n_{\mathbf{8_s}}+n_{\mathbf{8_v}}-n_{\mathbf{8_c}})\phi _1 + (-2 (1-n_{\mathbf{28}})-n_{\mathbf{8_s}}-n_{\mathbf{8_v}}+n_{\mathbf{8_c}}) \phi _3  \right. \\
&\left.+(-2 (1-n_{\mathbf{28}})+n_{\mathbf{8_s}}-n_{\mathbf{8_v}}-n_{\mathbf{8_c}})\phi _4  \right) .
\end{split}
\end{align}
The fourth chamber is defined by $\phi_4>\phi_3>\phi_1$, and the prepotential is 
\begin{align}\label{Eq:IMSSpin8Ch4}
\begin{split}
6\mathscr{F}^{\text{IMS}(4,3)}_{\text{Spin($8$)}}= & \ 8 (1-n_{\mathbf{28}}) \phi _1^3 +8 (1-n_{\mathbf{28}}) \phi _2^3 +2 (4 (1-n_{\mathbf{28}})-n_{\mathbf{8_s}})\phi _3^3 +2(4 (1-n_{\mathbf{28}})-n_{\mathbf{8_c}}-n_{\mathbf{8_v}})\phi _4^3 \\
& +3 \phi _2 \left(  (n_{\mathbf{8_c}}-n_{\mathbf{8_v}}+n_{\mathbf{8_s}})\phi _1^2+(-n_{\mathbf{8_c}}+n_{\mathbf{8_v}}+n_{\mathbf{8_s}})\phi _3^2+(n_{\mathbf{8_c}}+n_{\mathbf{8_v}}-n_{\mathbf{8_s}})\phi _4^2  \right) \\
& +6 \phi _2 \left( n_{\mathbf{8_s}} \phi _1 \phi _3 + n_{\mathbf{8_c}} \phi _1\phi _4 +n_{\mathbf{8_v}} \phi _3 \phi _4 \right) -6 \left( n_{\mathbf{8_s}} \phi _1^2 \phi _3+n_{\mathbf{8_c}} \phi _1^2 \phi _4 +n_{\mathbf{8_v}} \phi _3^2 \phi _4 \right) \\
& +3 \phi _2^2 \left( (-2 (1-n_{\mathbf{28}})-n_{\mathbf{8_c}}+n_{\mathbf{8_v}}-n_{\mathbf{8_s}})\phi _1   + (-2 (1-n_{\mathbf{28}})+n_{\mathbf{8_c}}-n_{\mathbf{8_v}}-n_{\mathbf{8_s}})\phi _3\right. \\
&\left.+ (-2 (1-n_{\mathbf{28}})-n_{\mathbf{8_c}}-n_{\mathbf{8_v}}+n_{\mathbf{8_s}}) \phi _4 \right) .
\end{split}
\end{align}
The fifth chamber is defined by $\phi_4>\phi_1>\phi_3$, and the prepotential is 
\begin{align}\label{Eq:IMSSpin8Ch5}
\begin{split}
6\mathscr{F}^{\text{IMS}(4,2)}_{\text{Spin($8$)}}= & \ 2 (4 (1-n_{\mathbf{28}})-n_{\mathbf{8_s}})\phi _1^3 +8 (1-n_{\mathbf{28}}) \phi _2^3 +8 (1-n_{\mathbf{28}}) \phi _3^3 +2(4 (1-n_{\mathbf{28}})-n_{\mathbf{8_v}}-n_{\mathbf{8_c}})\phi _4^3 \\
& +3 \phi _2 \left( (-n_{\mathbf{8_v}}+n_{\mathbf{8_c}}+n_{\mathbf{8_s}})\phi _1^2 +(n_{\mathbf{8_v}}-n_{\mathbf{8_c}}+n_{\mathbf{8_s}})\phi _3^2 +(n_{\mathbf{8_v}}+n_{\mathbf{8_c}}-n_{\mathbf{8_s}})\phi _4^2 \right) \\
& +6 \phi _2 \left( n_{\mathbf{8_s}} \phi _1 \phi _3+n_{\mathbf{8_c}} \phi _1 \phi _4+n_{\mathbf{8_v}} \phi _3 \phi _4 \right) -6 \left( +n_{\mathbf{8_s}} \phi _1 \phi _3^2+n_{\mathbf{8_c}} \phi _1^2 \phi _4+n_{\mathbf{8_v}} \phi _3^2 \phi _4 \right) \\
& +3 \phi _2^2 \left( (-2 (1-n_{\mathbf{28}})+n_{\mathbf{8_v}}-n_{\mathbf{8_c}}-n_{\mathbf{8_s}})\phi _1 +(-2 (1-n_{\mathbf{28}})-n_{\mathbf{8_v}}+n_{\mathbf{8_c}}-n_{\mathbf{8_s}})\phi _3 \right. \\
&\left.+ (-2 (1-n_{\mathbf{28}})-n_{\mathbf{8_v}}-n_{\mathbf{8_c}}+n_{\mathbf{8_s}}) \phi _4  \right) .
\end{split}
\end{align}
The sisth chamber is defined by $\phi_1>\phi_4>\phi_3$, and the prepotential is 
\begin{align}\label{Eq:IMSSpin8Ch6}
\begin{split}
6\mathscr{F}^{\text{IMS}(2,4)}_{\text{Spin($8$)}}= & \ 2(4 (1-n_{\mathbf{28}})-n_{\mathbf{8_s}}-n_{\mathbf{8_c}})\phi _1^3 +8 (1-n_{\mathbf{28}}) \phi _2^3 +8 (1-n_{\mathbf{28}}) \phi _3^3 +2 (4 (1-n_{\mathbf{28}})-n_{\mathbf{8_v}})\phi _4^3 \\
& +3 \phi _2 \left( (n_{\mathbf{8_s}}+n_{\mathbf{8_c}}-n_{\mathbf{8_v}})\phi _1^2+(n_{\mathbf{8_s}}-n_{\mathbf{8_c}}+n_{\mathbf{8_v}})\phi _3^2 +(-n_{\mathbf{8_s}}+n_{\mathbf{8_c}}+n_{\mathbf{8_v}})\phi _4^2  \right) \\
& +6 \phi _2 \left( n_{\mathbf{8_s}} \phi _1 \phi _3+n_{\mathbf{8_c}} \phi _1 \phi _4+n_{\mathbf{8_v}} \phi _3 \phi _4 \right) -6 \left( n_{\mathbf{8_s}} \phi _1 \phi _3^2+n_{\mathbf{8_c}} \phi _1 \phi _4^2+n_{\mathbf{8_v}}  \phi _3^2 \phi _4 \right) \\
& +3 \phi _2^2 \left( (-2 (1-n_{\mathbf{28}})-n_{\mathbf{8_s}}-n_{\mathbf{8_c}}+n_{\mathbf{8_v}}) \phi _1  + (-2 (1-n_{\mathbf{28}})-n_{\mathbf{8_s}}+n_{\mathbf{8_c}}-n_{\mathbf{8_v}})\phi _3 \right. \\
&\left. + (-2 (1-n_{\mathbf{28}})+n_{\mathbf{8_s}}-n_{\mathbf{8_c}}-n_{\mathbf{8_v}})\phi _4 \right) .
\end{split}
\end{align}

\subsection{Counting hypermultiplets with triple intersection numbers}\label{Sec:CountingHypers}
\def\arraystretch{1}

\begin{prop}
\label{prop:ChargeNumber}
 In the case of a Calabi--Yau threefold, the number of each representations derived by matching the triple intersection numbers of a G$_2$, Spin($7$), or Spin($8$)-model and the one-loop prepotential does not depend on the choice of a crepant resolution and  are given by 
\begin{align}
\textnormal{G}_2: &\quad n_{\mathbf{7}}=3S^2-10(g-1), & n_{\mathbf{14}}=g, \nonumber \\
\textnormal{Spin($7$)}: &\quad n_{\mathbf{7}}= S^2-3 (g-1), \quad n_{\mathbf{8}} =2 S^2-8 (g-1), \quad & n_{\mathbf{21}}=g,  \nonumber \\
\textnormal{Spin($8$)}: &\quad n_{\mathbf{8_v}}=n_{\mathbf{8_s}}=n_{\mathbf{8_c}}=S^2-4(g-1), \quad & n_{\mathbf{28}}=g. \nonumber
\end{align}
\end{prop}
\begin{proof} 
Direct comparison of the prepotentials given by equations \eqref{Eq:IMSG2}, \eqref{Eq:IMSSpin7}, \eqref{Eq:IMSSpin8Ch1}, \eqref{Eq:IMSSpin8Ch2}, \eqref{Eq:IMSSpin8Ch3}, \eqref{Eq:IMSSpin8Ch4}, \eqref{Eq:IMSSpin8Ch5}, \eqref{Eq:IMSSpin8Ch6} with triple intersection polynomials provided in Lemma \ref{Lem:Triple} after imposing $\phi_0=0$. 
\end{proof}

We checked that the numbers computed in Proposition \ref{prop:ChargeNumber} satisfy the genus formula of Aspinwall--Katz--Morrison,  here they are derived
from the triple intersection numbers. The same numbers were computed by Grassi and Morrison using Witten's genus formula. 

The advantage of our method is that we do not use the degeneration loci of the curves and hence the computation is the same even if the  fiber degeneration is not generic.
This method also  provides the number of charged hypermultiplets from a purely fivedimensional
point of view, thereby avoiding a six-dimensional argument based on cancellations of
anomalies  and the subtleties of the Kaluza-Klein circle compactification \cite{Grimm:2015zea}. 
The same method was used in \cite{ES} for SU(N)-models and in \cite{F4} for F$_4$-models.

The representation induced by the weights of vertical curves over codimension-two points is not always physical as it is possible that no hypermultiplet is charged
under that representation. In such a case, the representation is said to be ``frozen'' \cite{F4}.  

In all cases, the adjoint representation is always frozen when the curve $S$ is a smooth rational curve ($g=0$) \cite{Witten}.
For a G$_2$-model, the fundamental representation is frozen if and only if $g=3k+1$ and $S^2=10k$ with $k$ in  $\mathbb{Z}_{\geq 0}$. 
For a Spin($7$)-model, the vector representation is frozen when $S^2=3(g-1)$, whereas the spin representation is frozen when $S^2=4(g-1)$, and both representations are simultaneously frozen when $S^2=g-1=0$. 
For a Spin($8$)-model, all representations are frozen if and only if $g=0$ and $S^2=-4$.  In this case, one can check that there are  no curves  carrying the weights of the representations: this corresponds to the well-known  non-Higgsable model with a rational curve of self-intersection $-4$ \cite{Morrison:2012np}.

\section{Anomaly cancellation in six-dimensional ${\mathcal N}=(1,0)$ theories} \label{sec:anomaly6d}
In this section, we consider six-dimensional ${\cal N}=(1,0)$ supergravity theories derived from the compactifications of F-theory on  elliptically fibered Calabi--Yau threefolds  that are G$_2$, Spin($7$), or Spin($8$)-models. 

\begin{table}[htb]
\begin{center}
\begin{tabular}{|l|l|}
\hline 
Multiplets & Field content \\
\hline
gravity multiplet & $(g_{\mu\nu}, \psi_\mu^-, B_{\mu,\nu}^+)$ \\
tensor multiplet & $(B_{\mu\nu}^-, \lambda^+,\phi)$\\
vector multiplet&$(A_\mu, \lambda^-)$ \\
hypermultiplet &  $(\zeta^+, 4 q)$ \\
\hline 
\end{tabular}
\end{center}
\caption{Matter content of six-dimensional ${\cal N}=(1,0)$ supergravity theories. The symmetric tensor  $g_{\mu\nu}$ is the graviton, $\psi_\mu^-$ is 
the negative chirality
gravitino,  $B^+_{\mu\nu}$ is an antisymmetric two-form  with self-dual curvature.  The tensor
multiplet has an antisymmetric two-form $B^-_{\mu\nu}$  with anti-self-dual curvature,  the dilatino $\lambda^+$ and a dilaton $\phi$.  The vector multiplet contains a gauge field $A_\mu$  and its supersymmetric partner (the gaugino) $\lambda^-$. 
The hypermultiplet contains four pseudo-real scalars and a hyperino $\zeta^+$.
All the spinors are symplectic Majorana--Weyl spinors. 
The spinors $\psi_\mu^-$, $\lambda^+$, and $\lambda^-$ are also doublet with respect to the SU($2$)  R-symmetry group. 
The chirality of spinors and tensors is indicated by a superscript $\pm$. 
The tensor multiplet scalars parametrize the homogeneous symmetric space SO(1,$ n_T$)/SO($n_T$). The hypermultiplet  scalars parametrize locally a non-compact quaternionic-K\"ahler manifold.}\label{Table:Fields6D}
\end{table}

 Six-dimensional ${\cal N}=(1,0)$ supergravity theories  always have a unique gravity multiplet and can be coupled to $n_T$ chiral tensor multiplets, $n_V^{(6)}$ vector multiplets, and $n_H$ hypermultiplets \cite{Schwarz:1995zw,Sadov:1996zm}. Such theories are characterized by a choice of a non-compact quaternionic-K\"ahler manifold, a gauge group $G$ and the choice of a  representation $\bf{R}$ of the gauge group under which charged hypermultiplets transform. 
The vector multiplets always transform in the  adjoint representation of the gauge group $G$, thus  $n^{(6)}_V = \rm{dim}\ G$. 
 Models with $n_T>1$ would not have a covariant Lagrangian formulator, but will stil have  covariant equations of motion. 

Since we have chiral tensors and chiral fermions, the theory can have gauge, gravitational, and mixed local anomalies that have to be cancelled by the Green--Schwarz--Sagnotti mechanism. 
The condition for the pure gravitational anomalies in presence of an arbitrary number of tensor multiplets was derived in 
 \cite{Salam} and local anomalies in F-theory was first studied by Sadov  \cite{Sadov:1996zm} who determined conditions in terms of the intersection numbers of the divisors supporting the simple factors of the gauge group and the topology of the base of the elliptic fibration. In particular, the number of tensor multiplets $n_T$,  depends only on the square of the canonical class of the base (which we assume is a rational surface). 
 Using anomaly cancellation in six-dimensional theories  to determine the number of moduli fields of a theory can be traced to Seiberg \cite{Seiberg.1988} and  Erler \cite{Erler}.
 We use anomaly cancellation conditions of the six-dimensional theory to constrain the number of tensor multiplets, vector multiplets, and charged hypermultiplets transforming in the representations we observed geometrically for each of the models.

In our case,  the delicate part is to compute $n_H$, which has two main contributions:  the number of neutral hypermultiplets and the number  charged hypermultiplets. The neutral hypermultiplets depends only on  the Hodge number $h^{2,1}(Y)$ of the Calabi--Yau threefold. 
Following the techniques developed in  \cite{Euler} and  further explored  for various gauge groups in \cite{SU2SU3,SU2G2,SO4,EKY2,F4}, we can compute the Hodge numbers of an elliptically fibered Calabi--Yau  threefold defined by a crepant resolution of a Weierstrass model (see section \ref{sec:Euler}). 
This allows us in section \ref{sec:6dourmodels} to  derive explicitly the number of neutral hypermultiplets $n_H^0$ without studying the structure of all singular fibers of the elliptic fibration in contract to what is done for example in \cite{GM1}. 
Another advantage of the method we use for computing the Euler characteristic is that it produces at once a  generating function for  the Euler characteristic of a $G$-model over a base of arbitrary dimension and without imposing the Calabi--Yau condition \cite{Euler}. 
We can also compute the number of charged hypermultiplets  $n_{\bf{R}_i}$ transforming in given representation $\mathbf{R}_i$ of the gauge group  without using ad-hoc additional conditions on hypermultiplets.

\subsection{CPT and counting charged hypermultiplets in 6D}\label{CPT}

The most general representation  $\bf{R}$ of the gauge group $G$ can be written as follows 
\begin{equation}
{\mathbf{R}}=\mathbf{\rho}_c\oplus\mathbf{\overline{\rho}}_{c}\oplus \mathbf{\rho}_{r}\oplus \mathbf{\rho}_{pr}, 
\end{equation}
where $\mathbf{\rho}_c$ is a direct sum of irreducible complex representation (so that $\mathbf{\rho}_c\oplus\mathbf{\overline{\rho}}_c$ is a  real representation),  $\rho_r$ is a direct sum of real representations, 
and $\mathbf{\rho}_{pr}$ is a direct sum of pseudo-real representations. 

The most familiar case is when
\begin{equation}
{\mathbf{R}}=\mathbf{\rho}_c\oplus\mathbf{\overline{\rho}}_c.
\end{equation}
In that case, the  CPT-invariance requires that for each hypermultiplet transforming in a complex representation $\rho_c$, there is another hypermultiplet (its CPT-image) transforming in the dual representation $\overline{\mathbf{\rho}}_c$. 
The first example of a case with hypermultiplets charged under pseudo-real representation  in anomaly-free six-dimensional supergravities were constructed in \cite{Salam} where the pseudo-real representation $\bf{972}$ of E$_7$ was gauged.
In presence of pseudo-real representations, the CPT invariance is respected only when the pseudo-real representations are carried by half-hypermultiplets.\footnote{ If a hypermultiplet transform under a  pseudo-real representation  of the gauge group, it is  possible to impose an additional reality condition on the hypermultiplet  leaving only half of the degree of freedom \cite{Witten.Small}. Such a hypermultiplet is called a   {\em half-hypermultiplet} and is the smallest CPT self-conjugate supersymmetric multiplet in a  ${\cal N}=(1,0)$ six-dimensional supergravity theory.  
Some important examples with pseudo-real presentations are the SU($2$)-model with matter in the fundamental representation $\bf{2}$, the E$_7$-model with matter in the representation $\bf{56}$ \cite{Esole:2018vnm}, the SU($2$)$\times$G$_2$ model with matter in the representation $(\bf{2},\bf{7})$ \cite{SU2G2}. 

In the case of G$_2$, Spin($7$), and Spin($8$)-models, all the irreducible representations are real. Thus, there is no need to have half-hypermultiplets. 
Using the six-dimensional anomaly cancellation conditions, we find that without imposing CPT nor triality, the number of hypermultiplets of the representation $\mathbf{8_s}$ and $\mathbf{8_c}$  of a Spin($8$)-model can only be identified as their sum  $n_{\bf{8}_s}+n_{\bf{8}_c}$, but not individually. This is a striking difference from what we observe in a compactification of M-theory to five-dimensional supergravity theory on the same manifold  where we achieved the complete counting of charged hypermultiplets  ($n_{\bf{28}}$, $n_{\bf{8}_v}$, $n_{\bf{8}_s}$, $n_{\bf{8}_c}$)  by  matching the triple intersection polynomials of a Spin($8$)-model  with the one-loop prepotential of a five-dimensional gauge theory coupled to  hypermultiplets transforming  in the  representations $\bf{8}_v$, $\bf{8}_s$, and $\bf{8}_c$. In particular, we were able to compute the number of charged hypermultiplets transforming in the representationss $\mathbf{8_s}$ and $\mathbf{8_c}$ individually without imposing any additional conditions. This shed lights on  the differences between the nature of five and six dimensional supergravity theories with eight supercharges. The numbers of charged hypermultiplets in each representation  match those found by Grassi and Morrison \cite{GM1} where additional conditions on hypermultiplets are considered and the Hodge numbers and Euler characteristic are computed by a painful analysis of the Euler characteristic of all the singular fibers. }

\subsection{Generality on anomalies in 6d ${\cal N}=(1,0)$ theories with a simple Lie group } \label{sec:general6d}
The action of the gauge group on a hypermultiplet is characterized by a weight vector determining the charge of the hypermultiplet. 
It follows that hypermultiplets are organized into representations of the Lie algebra $\mathfrak{g}$. 
In our convention, a neutral hypermultiplet  has a  zero weight. 
The number of zero weights of a representation $\mathbf{R}_i$ is denoted by $\rm{dim}\, \mathbf{R}_{i,0}$. The difference 
$$(\rm{dim}\,\mathbf R_i -\rm{dim}\,\mathbf R_{i,0})$$
 is called the charge dimension of the representation $\mathbf{R}$ \cite{GM1}. 
We denote by $n_{\mathbf{R}_i}$ the multiplicity of the representation  $\mathbf{R}_i$ and $n_H^0$ the number of neutral hypermultiplets.  
The total number of charged hypermultiplets is \cite{GM1}
\begin{equation}
n_H^{ch} =\sum_{i} n_{{\mathbf R}_i} \left( \rm{dim}\,{\mathbf R}_i -\rm{dim}\,{\mathbf R}_{i,0} \right),
\end{equation}
where $\rm{dim}\,R_{i,0}$ is the number of zero weights in the representation $\mathbf{R}_i$. The total number of hypermultiplets is then
 \begin{equation}
 n_H=n_H^0+n_H^{ch}.
 \end{equation}
We can compute  $n_H^0$ and $n_T$ from the Hodge numbers of the Calabi--Yau threefold and its base $B$ of canonical class $K$, which we assume is a rational surface \cite{Sadov:1996zm,GM1}:  
\begin{equation}
n_V^{(6)}=\dim\, G, \quad    n_H^0=h^{2,1}(Y)-1, \quad n_T=h^{1,1}(B)-1=9-K^2.\end{equation}
The pure gravitational anomaly $\mathrm{tr}\  R^4$ is canceled by the vanishing of its coefficient \cite[Footnote 3]{Salam}: 
\begin{equation}
n_H-n_V+29n_T-273=0.
\end{equation}
 In six-dimensional theories, local anomalies are due to quadrangle graphs and their contributions are summarized into an eight-form polynomial in the Riemann curvature $R$ and the field strength $F_a$ of the gauge fields.
If the gauge group is a simple group $G$, the remainder of the anomaly polynomial is \cite{Sadov:1996zm} 
 \begin{equation}\label{eqn.I8def}
 I_8= \frac{9-n_T}{8} (\mathrm{tr} \   R^2)^2+\frac{1}{6}\  X^{(2)} \mathrm{tr}\ R^2-\frac{2}{3} X^{(4)},
 \end{equation}
where 
\begin{align}
X^{(n)}=\mathrm{tr}_{\mathbf{adj}}\  F^n -\sum_{i}n_{{\mathbf R}_i} \mathrm{tr}_{{\mathbf R}_i}\  F^n.
\end{align}

Choosing a reference representation $\mathbf F$,  we have
\begin{align}
X^{(2)}=\Big(A_{\mathbf{ adj}}   -\sum_i n_{{\bf R}_i} A_{{\bf R}_i}\Big) \mathrm{tr}_{\mathbf F}  F^2,\quad
X^{(4)}=\Big(B_{\mathbf{adj}}   -\sum_{i}n_{{\bf R}_i} B_{{\bf R}_i}    \Big) \mathrm{tr}_{\mathbf F}  F^4 +
\Big(C_{\bf{ adj}}   -\sum_{i}n_{{\bf R}_i} C_{{\bf R}_i} \Big)( \mathrm{tr}_{\mathbf F} F^2)^2
,
\end{align}
where $F$ is the Young--Mills curvature two-form of the gauge group $G$,  the coefficients $A_{{\bf R}_i}$, $B_{{\bf R}_i}$, and $C_{{\bf R}_i}$ are defined by the trace identities with respect to a reference representation ${\bf F}$:
\begin{equation} 
\mathrm{tr}_{{\bf R}_i}\  F^2 = A_{{\bf R}_i} \mathrm{tr}_{\bf F}\  F^2 , \quad 
\mathrm{tr}_{{\bf R}_i}\  F^4 =  B_{{\bf R}_i} \mathrm{tr}_{\bf F}\   F^4  +  C_{{\bf R}_i} (\mathrm{tr}_{\bf F}\ F^2)^2.
\end{equation}
In a theory with at least two quartic Casimirs, to  satisfy the anomaly cancellation conditions, the coefficient of $\mathrm{tr}_{\mathbf F}  F^4$ must vanish,\footnote{
In the case of a Lie algebra  with only one quartic Casimir, we simply have $B_R=0$. 
That is the case for a Lie algebra of type A$_1$, A$_2$, G$_2$, F$_4$, E$_6$, E$_7$, or E$_8$. 
}
\begin{align}
 B_{\mathbf{ adj}}   -\sum_i n_{{\bf R}_i} B_{{\mathbf R}_i}=0,
 \end{align}
and the anomaly cancellation polynomial in equation \eqref{eqn.I8def} needs to factorize.

Local anomalies are canceled by the Green-Schwarz mechanism when the polynomial I$_8$ factorizes \cite{Green:1984bx,Sagnotti:1992qw,Schwarz:1995zw}. 
The modification of the field strength $H$ of  the antisymmetric tensor $B$ is 
\begin{equation}
H= dB + \frac{1}{2} K \omega_{3L}+ \frac{2}{\lambda}S\omega_{a,3Y}, 
\end{equation}
where  $\omega_{3L}$ and $\omega_{a,3Y}$ are respectively the gravitational and Yang-Mills Chern--Simons  terms. 
 If I$_8$ factors as 
 \begin{equation}
 \text{I}_8= X\cdot  X,
 \end{equation}
  then the anomaly is canceled by adding the following Green-Schwarz counter-term 
\begin{equation}
\Delta L_{GS}\propto \frac{1}{2} B \wedge X.
\end{equation}
Following Sadov \cite{Sadov:1996zm}, to cancel the local anomalies, we consider
\begin{equation}
X= \frac{1}{2} K  \mathrm{tr} R^2 + \frac{2}{\lambda} S \mathrm{tr}_{\bf F} F^2,
\label{eq:Xfactor}
\end{equation} where the traces involving $R$ are in the 
six-dimensional representation of the Lorentz group SO$(1,5)$, 
 the  coefficient $\lambda$ is a constant normalization factor chosen such that the  smallest
topological charge of an embedded SU($2$) instanton in the gauge group G is one \cite{Kumar:2010ru, Park, Bernard}. 
This forces $\lambda$ to be the Dynkin index of the fundamental representation of  $G$ as listed in Table \ref{tb:normalization} \cite{Park}. 
\begin{table}[h!]
\begin{center}
\begin{tabular}{|c|c|c|c|c|c|c|c|c|c|}
\hline
 $\mathfrak{g}$ & A$_n$ & B$_n$ & C$_n$ & D$_n$ & E$_8$ & E$_7$ & E$_6$&  F$_4$ & G$_2$ \\
 \hline
 $\lambda$ & $1$ & $2$  & $1$ & $2$ & $60$ & $12$ & $6$ & $6$ & $2$ \\
 \hline  
\end{tabular}
\caption{The normalization factors for each simple gauge algebra. See \cite{Kumar:2010ru}.}
\label{tb:normalization}
\end{center}
\end{table}

If the simple group $G$ is supported on a divisor $S$, the local anomaly cancellation conditions are
\begin{subequations}
\begin{align}
n_T&=9-K^2 , \\
n_H-n_V^{(6)}+29n_T-273 &=0,\\
\left(B_{\bf{adj}}-\sum_{i}n_{\bf{R}_{i}}B_{\bf{R}_{i}}\right)& = 0, \\
\lambda  \left(A_{\bf{adj}}-\sum_{i}n_{\bf{R}_{i}}A_{\bf{R}_{i}}\right) & =6  K\cdot S, \\
\lambda^2 \left(C_{\bf{adj}}-\sum_{i}n_{\bf{R}_{i}}C_{\bf{R}_{i}}\right) & =-3 S ^2.
\end{align}
\end{subequations}

\subsection{Local anomalies for G$_2$, Spin($7$), and Spin($8$)-models}\label{sec:6dourmodels}

In this paper, the reference representation will always be the vector representation: the $\bf{7}$ of G$_2$, the $\bf{7}$ of Spin($7$), and the $\bf{8}_v$ of Spin($8$). 
Trace identities for G$_2$, Spin($7$), and Spin($8$) are obtained from the tables of reference  \cite{Avramis:2005hc} and summarized in Table \ref{Table:Trace}. 
\begin{table}[htb]
\begin{center}
$
\begin{array}{|c|c|c|c|c|}
\hline
G& R& A_R & B_R & C_R\\
\hline
& \bf{7} & 1 &0 &1/4 \\
\cline{2-5}
\text{G}_2 & \bf{14} &4 & 0&5/2 \\
\hline
& \bf{7} &1 & 1&0 \\
\cline{2-5}
\text{Spin($7$)}& \bf{8} &1 &-1/2 &3/8\\
\cline{2-5}
& \bf{21} &5 & -1&3\\
\hline 
& \bf{8}_v & 1& 1&0 \\
\cline{2-5}
\text{Spin($8$)}& \bf{8}_s, \  \bf{8}_c &1 &-1/2 &3/8\\
\cline{2-5}
& \bf{28} &6 & 0 & 3 \\
\hline 
\end{array}
$
\end{center}
\caption{
Trace identities for the smallest representations of G$_2$, Spin($7$), and Spin($8$).  See   \cite{Avramis:2005hc}. 
Given a representation $\bf{R}$ of a Lie group $G$,  the coefficients A$_{\bf R}$,
 B$_{\bf R}$, and C$_{\bf R}$ are defined by the two trace identities 
$\mathrm{tr}_{{\bf R}}\  F^2 = A_{{\bf R}} \mathrm{tr}_{\bf F}\  F^2$  and $
\mathrm{tr}_{{\bf R}}\  F^4 =  B_{{\bf R}} \mathrm{tr}_{\bf F}\   F^4  +  C_{{\bf R}} (\mathrm{tr}_{\bf F}\ F^2)^2$, where  the reference representation ${\bf F}$ is the $\bf{7}$ for G$_2$ and Spin($7$), and the $\bf{8}_v$ for Spin($8$). Since G$_2$ has a unique quartic Casimir invarant, we choose $B_{\bf R}=0$. 
 }\label{Table:Trace}
\end{table}

The group G$_2$ has a unique quartic Casimir, so the equation involving $B_{\bf{R}_i}$ is automatically satisfied. 
The non-gravitational conditions for G$_2$ with matter in the representations $\bf{14}$ and ${\bf 7}$ are 
\begin{align}
4(1-n_{\bf{14}})-n_{\bf{7}}  =3  K\cdot S, \quad 
10(1-n_{\bf{14}})-n_{\bf{7}} =-3 S ^2. 
\end{align}
Keeping in mind that  $2-2g= -K\cdot S-S^2$, the unique solution to the previous two equations is 
\begin{equation}
n_{\bf{14}}=g, \quad n_{\bf{7}}=10(1-g)+3 S^2,
\end{equation}
which reproduces what we had from the triple intersection numbers.

For a Spin($7$)-model, we have matter in the representation $\bf{21}$, $\bf{8}$, and $\bf{7}$, the corresponding non-gravitational anomaly equations are 
 \begin{equation}
-(1-n_{\bf{21}}) +\frac{1}{2} n_{\bf{8}}-n_{\bf{7}}=0
, \quad 5(1-n_{\bf{21}}) -n_{\bf{8}}-n_{\bf{7}}=3 K\cdot S, \quad     4(1-n_{\bf{21}})-\frac{1}{2}n_{\bf{8}}=-S^2.
\end{equation}
Once again, we can completely solve the equations and obtain the unique solution: 
\begin{equation}
n_{\mathbf{7}}= S^2-3 (g-1), \quad n_{\mathbf{8}} =2 S^2-8 (g-1), \quad  n_{\mathbf{21}}=g,
\end{equation}
where we used  $2-2g=-KS -S^2$. 

For a Spin($8$)-model, the non-gravitational anomaly equations are 
\begin{equation}
-n_{\bf{8}_v}+\frac{1}{2} (n_{\bf{8}_s}+ n_{\bf{8}_c})=0, \quad 
6 (1-n_{\bf{28}})-n_{\bf{8}_v}-(n_{\bf{8}_s}+ n_{\bf{8}_c})=3 K\cdot S, \quad 
4  (1-n_{\bf{28}})-\frac{1}{2} (n_{\bf{8}_s}+n_{\bf{8}_c})=- S^2, 
\end{equation}
We can solve these equations for $n_{\bf{8}_v}$, $n_{\bf{28}}$, and $n_{\bf{8}_s}+n_{\bf{8}_c}$: 
 \begin{equation}
 n_{\bf{28}}=g,\quad n_{\bf{8}_v}=4(1-g) + S^2, \quad n_{\bf{8}_s}+n_{\bf{8}_c}=8(1-g)+2 S^2,
 \end{equation}
 where we used  $2-2g=-KS -S^2$. 
Imposing triality, will require  $n_{\bf{8}_v}=n_{\bf{8}_c}=n_{\bf{8}_s}$, which would give 
\begin{equation}
 n_{\bf{28}}=g,\quad n_{\bf{8}_v}=n_{\bf{8}_s}=n_{\bf{8}_c}=4(1-g) + S^2.
\end{equation}
\begin{rem}[Number of charged hypermultiplets in 5d and 6d for the Spin($8$)-model]
We notice that for a Spin($8$)-model,  comparing the triple intersection polynomial with the 5d prepotential does completely resolve the values of the number of charged hypermultiplets:  
 $n_{\bf{28}}=g$, $n_{\bf{8}_v}=n_{\bf{8}_s}=n_{\bf{8}_c}=4(1-g) + S^2$. 
In contrast, the  6d anomaly cancellation conditions are less constrained as they only impose   
$n_{\bf{28}}=g$,    $n_{\bf{8}_v}=4(1-g) + S^2$, and $n_{\bf{8}_s}+n_{\bf{8}_c}=8(1-g)+2 S^2$. 
To retrieve the result found in 5d, we have to impose by hand triality. 
Since $\bf{8}_s$ and $\bf{8}_c$ are dual to each other, we can also retrieve $n_{\bf{8}_s}=n_{\bf{8}_c}$ by referring to the CPT invariance.
\end{rem}

We are left with the pure gravitational  anomaly that requires checking  equation.

From Theorem \ref{lem:Hoddge}, we have the Hodge numbers
\
$$
\begin{array}{|c|c|c|c|}
\hline
&  h^{1,1} &h^{2,1} \\ \hline
\textnormal{G}_2   &13-K^2 & 13+29 K^2-18 S^2+48 (g-1)  \\\hline
\textnormal{Spin($7$) }& 14-K^2& 14+29 K^2-22 S^2+64 (g-1)\\\hline
\textnormal{Spin($8$)} & 15-K^2& 15+29 K^2-24 S^2+72 (g-1) \\\hline
\end{array}
$$
We also need the following data on the representations, their dimensions, and the number of their zero weights: 
$$
\begin{array}{|c||c|c||c|c|c||c|c|c|c|}
\cline{1-10}
\multicolumn{1}{|c||}{\mathfrak{g}}& \multicolumn{2}{c||}{\text{G$_2$}} & \multicolumn{3}{c||}{\text{Spin}(7)} & \multicolumn{4}{c|}{\text{Spin($8$)}}\\
\hline
\multicolumn{1}{|c||}{n_V=\text{dim}\ G}& \multicolumn{2}{c||}{$14$} & \multicolumn{3}{c||}{21} & \multicolumn{4}{c|}{28}\\
\hline
\mathbf{R}& \mathbf{14} & \mathbf{7} & \mathbf{21} &\mathbf{7} &\mathbf{8}  &\mathbf{28}&  \mathbf{8}_v &\mathbf{8}_s &\mathbf{8}_c\\
\hline 
\rm{dim}{\mathbf R} -\rm{dim} {\mathbf R}_{0}&12 & 6 & 18& 6 & 8 & 24& 8 & 8 & 8 \\
\hline
\end{array}
$$
The six-dimensional anomaly conditions gave the following information on the number of charged hypermultiplets: 
\begin{align}
\label{eq:nRg2}
\textnormal{G}_2: &&\quad  n_{\mathbf{7}}&=-10(g-1)+3S^2,    &     &n_{\mathbf{14}}=g, \\
\textnormal{Spin($7$)}: &&\quad  n_{\mathbf{7}}&= S^2-3 (g-1),  \quad n_{\mathbf{8}} &=2 S^2-8 (g-1), \quad & n_{\mathbf{21}}=g, \\
\textnormal{Spin($8$)}: &&\quad  n_{\mathbf{8_v}}&=S^2-4(g-1),  \quad n_{\mathbf{8_s}}+n_{\mathbf{8_c}}&=2S^2-8(g-1), \quad & n_{\mathbf{28}}=g. 
\end{align}

We recall that 
\begin{align}
n_{\text T}=9-K^2, \quad 
n_{\text H}=n_{\text H}^0 +n_{\text H}^{\text{ch}}, \quad n_{\text H}^0 =h^{2,1}(Y)+1, \quad 
 n^{\text{ch}}_{\text H} =
  \sum_{i} ( \mathrm{dim}\ \mathbf{R}_i-\mathrm{dim} \mathbf{R}_{i,0}) n_{\mathbf R_i}\nonumber
  \end{align}
We then observe that in each case, the gravitational anomaly is satisfied: 
\begin{equation}\nonumber
n_H-n_V+29n_T-273=0.
\end{equation}

When all anomaly cancellations are satisfied, we have 
 $$
\begin{aligned}
X^{(2)} 
=
3K \cdot S\  \mathrm{tr}_{\mathbf F}  F^2 , \quad 
X^{(4)} 
=-\frac{3}{4}S^2 \ (\mathrm{tr}_{\mathbf F}\ F^2)^2, 
\end{aligned}
$$
and the  anomaly cancellation polynomial factorizes as a perfect square \cite{Sadov:1996zm} 
$$
\begin{aligned}
I_8 &=\frac{K^2}{8} (\mathrm{tr} \   R^2)^2+\frac{1}{2}  K \cdot S\  \mathrm{tr}_{\mathbf F}  F^2\  (\mathrm{tr}\ R^2)+\frac{1}{2}S^2 \ (\mathrm{tr}_{\mathbf F}\ F^2)^2\\
&=\frac{1}{2} \Big( \frac{1}{2  }c_1(B) \ \mathrm{tr} \   R^2 -  S\  \mathrm{tr}_{\mathbf F}\ F^2 \Big)^2.
\end{aligned}
$$
Since I$_8$ factorizes, the anomaly can be canceled by the Green-Schwarz mechanism.

\subsection{Global anomaly for G$_2$-models} \label{sec:6dgravG2}

The non-triviality of the homotopy group $\pi_6(\text{G$_2$})=\mathbb{Z}/3\mathbb{Z}$ implies that G$_2$ gauge theories in  six-spacetime dimension can have global anomalies. 
The absence of these global anomalies require the following arithmetic condition 
\begin{equation}
1- \sum_i n_{\bf{R}_i} A_{\bf{R}_i}=0 \quad (\mathrm{mod}\ 3).
\end{equation}
This condition is automatically satisfied with the values we have obtained for $n_{\bf{14}}$ and $n_{\bf{7}}$, given by equation \eqref{eq:nRg2}. The global anamalies were investigated for the case of \text{SU(2)$\times$G$_2$}-models in a similar manner in \cite{SU2G2}.

\section*{Acknowledgements}
The authors are grateful to Paolo Aluffi, Lara Anderson, Patrick Jefferson, Shu-Heng Shao, Washington Taylor,  and Shing-Tung Yau for helpful discussions. 
M.E. and M.J.K. are grateful to  the Field Institute  and Virginia Tech for their hospitality during part of this research. 
M.J.K. 
would like to thank the organizers of the 2018 Summer Workshop at the Simons Center
for Geometry and Physics for their hospitality and support during part of
this work.  
The authors are thankful to all the participants of the workshop  ``A Three-Workshop Series on the Mathematics and Physics of F-theory'' supported by the National Science Foundation (NSF) grant DMS-1603247. 
M.E. is supported in part by the National Science Foundation (NSF) grant DMS-1406925  and DMS-1701635 ``Elliptic Fibrations and String Theory''.
Part of this research was conducted while Jagadeesan was an Economic Design Fellow at the Harvard Center of Mathematical Sciences and Applications. Jagadeesan gratefully acknowledges the support of a National Science Foundation Graduate Research Fellowship under Grant No. DGE-1745303.
M.J.K. was supported by the National Science Foundation (NSF) grant PHY-1352084.

\bibliography{mboyoBib}

\end{document}